\documentclass[DIV=11,titlepage=off,11pt,a4paper]{scrartcl}
\usepackage[utf8]{inputenc}
\usepackage{todonotes}

\usepackage{xr-hyper}
\usepackage[colorlinks]{hyperref}
\usepackage{amsmath,amsthm,amssymb,mathtools,dsfont}
\usepackage[multidot]{grffile}
\usepackage{subcaption}

\usepackage[capitalise]{cleveref}

\crefname{lemma}{Lemma}{Lemmas}
\crefname{proposition}{Proposition}{Propositions}
\crefname{definition}{Definition}{Definitions}
\crefname{theorem}{Theorem}{Theorems}
\crefname{conjecture}{Conjecture}{Conjectures}
\crefname{corollary}{Corollary}{Corollaries}
\crefname{example}{Example}{Examples}
\crefname{section}{Section}{Sections}
\crefname{appendix}{Appendix}{Appendices}
\crefname{figure}{Figure}{Figures}
\crefname{equation}{Equation}{Equations}
\crefname{table}{Table}{Tables}
\crefname{item}{Property}{Properties}
\crefname{remark}{Remark}{Remarks}

\newtheorem{theorem}{Theorem}
\newtheorem{definition}[theorem]{Definition}
\newtheorem{corollary}[theorem]{Corollary}

\newtheorem{lemma}[theorem]{Lemma}

\newtheorem{remark}[theorem]{Remark}
\RequirePackage{xparse}
\RequirePackage{amsmath}

\DeclareDocumentCommand\bra{ s m t\ket s g }
{ 
	\IfBooleanTF{#3}
	{ 
		\IfBooleanTF{#1}
		{ 
			\IfNoValueTF{#5}
			{\braket*{#2}{} \IfBooleanTF{#4}{*}{}}
			{\braket*{#2}{#5}}
		}
		{
			\IfBooleanTF{#4}
			{ 
				\IfNoValueTF{#5}
				{\braket{#2}{} *}
				{\braket*{#2}{#5}}
			}
			{\braket{#2}{\IfNoValueTF{#5}{}{#5}}} 
		}
	}
	{ 
		\IfBooleanTF{#1}
		{\vphantom{#2}\left\langle\smash{#2}\right\rvert}
		{\left\langle{#2}\right\rvert}
		\IfBooleanTF{#4}{*}{}
		\IfNoValueTF{#5}{}{#5}
	}
}

\DeclareDocumentCommand\ket{ s m }
{ 
	\IfBooleanTF{#1}
	{\vphantom{#2}\left\lvert\smash{#2}\right\rangle} 
	{\left\lvert{#2}\right\rangle} 
}

\DeclareDocumentCommand\innerproduct{ s m g }
{ 
	\IfBooleanTF{#1}
	{ 
		\IfNoValueTF{#3}
		{\vphantom{#2}\left\langle\smash{#2}\middle\vert\smash{#2}\right\rangle}
		{\vphantom{#2#3}\left\langle\smash{#2}\middle\vert\smash{#3}\right\rangle}
	}
	{ 
		\IfNoValueTF{#3}
		{\left\langle{#2}\middle\vert{#2}\right\rangle}
		{\left\langle{#2}\middle\vert{#3}\right\rangle}
	}
}
\DeclareDocumentCommand\braket{}{\innerproduct} 
	
\DeclareDocumentCommand\outerproduct{ s m g }
{ 
	\IfBooleanTF{#1}
	{ 
		\IfNoValueTF{#3}
		{\vphantom{#2}\left\lvert\smash{#2}\middle\rangle\!\middle\langle\smash{#2}\right\rvert}
		{\vphantom{#2#3}\left\lvert\smash{#2}\middle\rangle\!\middle\langle\smash{#3}\right\rvert}
	}
	{ 
		\IfNoValueTF{#3}
		{\left\lvert{#2}\middle\rangle\!\middle\langle{#2}\right\rvert}
		{\left\lvert{#2}\middle\rangle\!\middle\langle{#3}\right\rvert}
	}
}

\newcommand{\Umax}{\op U_\mathrm{max}}

\newcommand\hop{v}
\newcommand\os{u}
\newcommand\n{n} 
\newcommand\tcoeff{b}

\newcommand\ee{\mathrm e}
\newcommand\ii{\mathrm i}
\newcommand\dd{\mathrm d}
\newcommand\field{\mathds}
\newcommand\id{\mathds 1}

\newcommand\op{}
\newcommand\R{\mathcal R}

\DeclareMathOperator{\BigO}{O}
\DeclareMathOperator{\cost}{\mathcal T_\mathrm{cost}}
\DeclareMathOperator{\argmax}{argmax}

\DeclareMathOperator{\Tr}{Tr}
\DeclareMathOperator{\poly}{poly}

\newcommand\ad{a^{\dagger}}

\usepackage{qcircuit}
\DeclareDocumentCommand{\circuit}{m}{%
\Qcircuit @C=1.5em @R=1.2em {
#1
}
}

\newcommand{\cN}{\mathcal N}

\usepackage[section]{placeins}
\usepackage{floatpag}
\usepackage{subcaption}

\usepackage{tikz}
\usetikzlibrary{calc, shapes.geometric, arrows, decorations.markings}

\usepackage{booktabs,multirow}
\usepackage{authblk}

\DeclareDocumentCommand\bra{ s m t\ket s g }
{ 
	\IfBooleanTF{#3}
	{ 
		\IfBooleanTF{#1}
		{ 
			\IfNoValueTF{#5}
			{\braket*{#2}{} \IfBooleanTF{#4}{*}{}}
			{\braket*{#2}{#5}}
		}
		{
			\IfBooleanTF{#4}
			{ 
				\IfNoValueTF{#5}
				{\braket{#2}{} *}
				{\braket*{#2}{#5}}
			}
			{\braket{#2}{\IfNoValueTF{#5}{}{#5}}} 
		}
	}
	{ 
		\IfBooleanTF{#1}
		{\vphantom{#2}\left\langle\smash{#2}\right\rvert}
		{\left\langle{#2}\right\rvert}
		\IfBooleanTF{#4}{*}{}
		\IfNoValueTF{#5}{}{#5}
	}
}

\DeclareDocumentCommand\ket{ s m }
{ 
	\IfBooleanTF{#1}
	{\vphantom{#2}\left\lvert\smash{#2}\right\rangle} 
	{\left\lvert{#2}\right\rangle} 
}

\DeclareDocumentCommand\innerproduct{ s m g }
{ 
	\IfBooleanTF{#1}
	{ 
		\IfNoValueTF{#3}
		{\vphantom{#2}\left\langle\smash{#2}\middle\vert\smash{#2}\right\rangle}
		{\vphantom{#2#3}\left\langle\smash{#2}\middle\vert\smash{#3}\right\rangle}
	}
	{ 
		\IfNoValueTF{#3}
		{\left\langle{#2}\middle\vert{#2}\right\rangle}
		{\left\langle{#2}\middle\vert{#3}\right\rangle}
	}
}
\DeclareDocumentCommand\braket{}{\innerproduct} 
	
\DeclareDocumentCommand\outerproduct{ s m g }
{ 
	\IfBooleanTF{#1}
	{ 
		\IfNoValueTF{#3}
		{\vphantom{#2}\left\lvert\smash{#2}\middle\rangle\!\middle\langle\smash{#2}\right\rvert}
		{\vphantom{#2#3}\left\lvert\smash{#2}\middle\rangle\!\middle\langle\smash{#3}\right\rvert}
	}
	{ 
		\IfNoValueTF{#3}
		{\left\lvert{#2}\middle\rangle\!\middle\langle{#2}\right\rvert}
		{\left\lvert{#2}\middle\rangle\!\middle\langle{#3}\right\rvert}
	}
}

\newtheoremstyle{natcomms}
{3pt}
{3pt}
{\normalfont}
{}
{\itshape}
{.}
{.5em}
{}%

\theoremstyle{natcomms}

\renewcommand\P{\mathcal P}

\DeclareDocumentCommand{\circuit}{m}{%
\Qcircuit @C=1.5em @R=1.2em {
#1
}
}

\title{Hamiltonian Simulation Algorithms for Near-Term Quantum Hardware}
\author[1,2]{Laura Clinton \thanks{laura@phasecraft.io}}
\author[1,3]{Johannes Bausch \thanks{johannes@phasecraft.io}}
\author[1]{Toby Cubitt \thanks{toby@phasecraft.io}}
\affil[1]{PhaseCraft Ltd.}
\affil[2]{Department of Computer Science, University College London}
\affil[3]{CQIF, DAMTP, University of Cambridge}
\date{November 2019}

\begin{document}

\maketitle

\begin{abstract}
The quantum circuit model is the de-facto way of designing quantum algorithms.
Yet any level of abstraction away from the underlying hardware incurs overhead.
In this work, we develop quantum algorithms for Hamiltonian simulation ``one level below'' the circuit model, exploiting the underlying control over qubit interactions available in most quantum hardware and deriving analytic circuit identities for synthesising multi-qubit evolutions from two-qubit interactions. We then analyse the impact of these techniques under the standard error model where errors occur per gate, and an error model with a constant error rate per unit time.

To quantify the benefits of this approach, we apply it to time-dynamics simulation of the 2D spin Fermi-Hubbard model.
Combined with new error bounds for Trotter product formulas tailored to the non-asymptotic regime and an analysis of error propagation, we find that e.g.\ for a $5\times5$ Fermi-Hubbard lattice we reduce the circuit depth from $1,243,586$ using the best previous fermion encoding and error bounds in the literature, to $3,209$ in the per-gate error model, or the circuit-depth-equivalent to $259$ in the per-time error model.
This brings Hamiltonian simulation, previously beyond reach of current hardware for non-trivial examples, significantly closer to being feasible in the NISQ era.
\end{abstract}

\clearpage
\thispagestyle{empty}
\enlargethispage{5em}

\section*{Introduction}\label{sec:intro}
Quantum computing is on the cusp of entering the era in which quantum hardware can no longer be simulated effectively classically, even on the world's biggest supercomputers~\cite{Villalonga2020,Bremner2011,Aaronson2010,Bremner2017,Harrow2017}.
Google recently achieved the first so-called ``quantum supremacy'' milestone demonstrating this~\cite{GoogleAI}.
Whilst reaching this milestone is an impressive experimental physics achievement, the very definition of this goal allows it to be a demonstration that has no useful practical applications~\cite{Preskill2012}.
The recent Google results are of exactly this nature.
By far the most important question for quantum computing now is to determine whether there are useful applications of this class of noisy, intermediate-scale quantum (NISQ) hardware~\cite{Preskill2018}.

However, current quantum hardware is still extremely limited, with $\approx 50$ qubits capable of implementing quantum circuits up to a gate depth of $\approx 20$~\cite{GoogleAI}.
This is far too limited to run useful instances of even the simplest textbook quantum algorithms, let alone implement the error correction and fault-tolerance required for large-scale quantum computations.
Estimates of the number of qubits and gates required to run Shor's algorithm on integers that cannot readily be factored on classical computers place it -- and related number-theoretic algorithms -- well into the regime of requiring a fully scalable, fault-tolerant quantum computer~\cite{Haner2016,Roetteler2017}.
Studies of practically relevant combinatorial problems tell a similar story for capitalising on the quadratic speedup of Grover's algorithm~\cite{Montanaro2015}.
Quantum computers are naturally well-suited for simulation of quantum many-body systems~\cite{Feynman1982,Lloyd1996a} -- a task that is notoriously difficult on classical computers. Quantum simulation is likely to be one of the first practical applications of quantum computing.
But, whilst the number of qubits required to run interesting quantum simulations may be lower than for other applications, careful studies of the gate counts required for a quantum chemistry simulation of molecules that are not easily tractible classically~\cite{google-chemisty},
or for simple condensed matter models~\cite{Kivlichan2019}, remain far beyond current hardware.

With severely resource-constrained hardware such as this, squeezing every ounce of performance out of it is crucial.
The quantum circuit model is the standard way to design quantum algorithms, and quantum gates and circuits provide a highly convenient abstraction of quantum hardware.
Circuits sit at a significantly lower level of abstraction than even assembly code in classical computing.
But any layer of abstraction sacrifices some overhead for the sake of convenience.
The quantum circuit model is no exception.

In the underlying hardware, quantum gates are typically implemented by controlling interactions between qubits.
E.g.\ by changing voltages to bring superconducting qubits in and out of resonance; or by laser pulses to manipulate the internal states of trapped ions.
By restricting to a fixed set of standard gates, the circuit model abstracts away the full capabilities of the underlying hardware.
In the NISQ era, it is not clear this sacrifice is justified.
The Solovay-Kitaev theorem tells us that the overhead of any particular choice of universal gate set is at most poly-logarithmic~\cite{Kitaev1997,Dawson2005}.
But when the available circuit depth is limited to $\approx 20$, even a constant factor improvement could make the difference between being able to run an algorithm on current hardware, and being beyond the reach of foreseeable hardware.

The advantages of designing quantum algorithms ``one level below'' the circuit model become particularly acute in the case of Hamiltonian time-dynamics simulation.
To simulate evolution under a many-body Hamiltonian $\op H=\sum_{\langle i,j\rangle} \op h_{ij}$, the basic Trotterization algorithm~\cite{Lloyd1996a,Nielsen_and_Chuang} repeatedly time-evolves the system under each individual interaction $\op h_{ij}$ for a small time-step $\delta$,
\begin{equation}\label{eq:trotter-intro}
\ee^{-\ii  H T} \simeq \prod_{n=0}^{T/\delta} \left(\prod_{\langle i,j\rangle} \ee^{-\ii h_{ij}\delta}\right).
\end{equation}
To achieve good precision, $\delta$ must be small.
In the circuit model, each $\ee^{-\ii  h_{ij}\delta}$ Trotter step necessarily requires at least one quantum gate to implement.
Thus the required circuit depth -- and hence the total run-time -- is at least $T/\delta$.
Contrast this with the run-time if we were able to implement $\ee^{-\ii  h_{ij}\delta}$ directly in time $\delta$.
The total run-time would then be $T$, which improves on the circuit model algorithm by a factor of $1/\delta$.
This is ``only'' a constant factor improvement, in line with the Solovay-Kitaev theorem.
But this ``constant'' can be very large; indeed, it diverges to $\infty$ as the precision of the algorithm increases.

It is unrealistic to assume the hardware can implement $\ee^{-\ii  h_{ij}\delta}$ for any desired interaction $ h_{ij}$ and any time $\delta$.
Furthermore, the available interactions are typically limited to at most a handful of specific types, determined by the underlying physics of the device's qubit and quantum gate implementations.
And these interactions cannot be switched on and off arbitrarily fast, placing a limit on the smallest achievable value of $\delta$.
There are also experimental challenges associated with implementing gates with small $\delta$ with the same fidelities as those with $\delta \approx \BigO(1)$.

A major criticism of analogue computation (classical and quantum) is that it cannot cope with errors and noise.
The ``N'' in NISQ stands for ``noisy''; errors and noise will be a significant factor in all foreseeable quantum hardware.
But near-term hardware has few resources to spare even on basic error correction, let alone fault-tolerance.
Indeed, near-term hardware may not always have the necessary capabilities. E.g.\ the intermediate measurements required for active error-correction are not possible in all superconducting circuit hardware~\cite[Sec.~II]{GoogleAI+Supplementary}.

Algorithms that cope well with errors and noise, and still give reasonable results without active error correction or fault-tolerance, are thus critical for NISQ applications.

Designing algorithms ``one level below'' the circuit model can also in some cases reduce the impact of errors and noise during the algorithm.
Again, this benefit is particularly acute in Hamiltonian simulation algorithms.
If an error occurs on a qubit in a quantum circuit, a two-qubit gate acting on the faulty qubit can spread the error to a second qubit.
In the absence of any error-correction or fault-tolerance, errors can spread to an additional qubit with each two-qubit gate applied, so that after circuit depth $n$ the error can spread to all $n$ qubits.

In the circuit model, each $\ee^{-\ii  h_{ij}\delta}$ Trotter step requires at least one two-qubit gate.
So a single error can be spread throughout the quantum computer after simulating time-evolution for time as short as $\delta n$.
However, if a two-qubit interaction $\ee^{-\ii  h_{ij}\delta}$ is implemented directly, one would intuitively expect it to only ``spread the error'' by a small amount $\delta$ for each such time-step.
Thus we might expect it to take time $O(n)$ before the error can propagate to all $n$ qubits -- a factor of $1/\delta$ improvement. Another way of viewing this is that, in the circuit model, the Lieb-Robinson velocity~\cite{Lieb1972} at which effects propagate in the system is always $O(1)$, regardless of what unitary dynamics is being implemented by the overall circuit.
  In contrast, the Trotterized Hamiltonian evolution has the same Lieb-Robinson velocity as the dynamics being simulated: $O(1/\delta)$ in the same units.

The Fermi-Hubbard model is believed to capture, in a simplified toy model, key aspects of high-temperature superconductors, which are still less well understood theoretically than their low-temperature brethren.
Its Hamiltonian is given by a combination of on-site and hopping terms:
\begin{align}\label{eq:FH-H-intro}
  H_{\text{FH}} &\coloneqq  \sum_{i=1}^{N}  h_{\text{on-site}}^{(i)} \ + \sum_{i<j,\sigma}  h_{\text{hopping}}^{(i,j,\sigma)} \\
&\coloneqq  \os \sum_{i=1}^{N}   \ad_{i \uparrow} a_{i \uparrow} \ad_{i \downarrow} a_{i \downarrow}  +  \hop \sum_{i<j,\sigma}   \left(\ad_{i \sigma} a_{j \sigma} + \ad_{j\sigma} a_{i \sigma}\right). \nonumber
\end{align}
describing electrons with spin $\sigma = \uparrow$ or $\downarrow$ hopping between neighbouring sites on a lattice, with an on-site interaction between opposite-spin electrons at the same site.
The Fermi-Hubbard model serves as a particularly good test-bed for NISQ Hamiltonian simulation algorithms for a number of reasons~\cite[Sec.~IV]{Bauer2020}, beyond the fact that it is a scientifically interesting model in its own right:
\begin{enumerate}
\item%
  The Fermi-Hubbard model was a famous, well-studied condensed matter model long before quantum computing was proposed. It is therefore less open to the criticism of being an artificial problem tailored to fit the algorithm.
\item%
  It is a fermionic model, which poses particular challenges for simulation on (qubit-based) quantum computers. Most of the proposed practical applications of quantum simulation involve fermionic systems, either in quantum chemistry or materials science. So achieving quantum simulation of fermionic models is an important step on the path to practical quantum computing applications.
\item%
  There have been over three decades of research developing ever-more-sophisticated classical simulations of Fermi-Hubbard-model physics~\cite{LeBlanc2015}.
  This gives clear benchmarks against which to compare quantum algorithms.
  And it reduces the likelihood of there being efficient classical algorithms, which haven't been discovered because little interest or effort has been devoted to the model.
\end{enumerate}

The state-of-the-art quantum circuit-model algorithm for simulating the time dynamics of the 2D Fermi-Hubbard model on an $8\times 8$ lattice requires $\approx 10^{7}$ Toffoli gates ~\cite[Sec.~C:~Tb.~2]{Kivlichan2019}.
This includes the overhead for fault-tolerance, which is necessary for the algorithm to achieve reasonable precision with the gate fidelities available in current and near-term hardware. But it additionally incorporates performing phase estimation, which is a significant extra contribution to the gate count.
Thus, although this result is indicative of the scale required for standard circuit-model Hamiltonian simulation, a direct comparison of this result with time-dynamics simulation would be unfair.

To establish a fair benchmark, using available Trotter error bounds from the literature~\cite{Childs2017} with the best previous choice of fermion encoding in the literature~\cite{Verstraete2005}, we calculate that one could achieve a Fermi-Hubbard time-dynamics simulation on a $5\times 5$ square lattice, up to time $T=7$ and to within $10\%$ accuracy, using $50$ qubits and $1,243,586$ standard two-qubit gates.
This estimate assumes the effects of decoherence and errors in the circuit can be neglected, which is certainly over-optimistic.

Our results rely on developing more sophisticated techniques for synthesising many-body interactions out of the underlying one- and two-qubit interactions available in the quantum hardware (see 
Results). This gives us access to $\ee^{-\ii  h_{ij}\delta}$ for more general interactions $h_{ij}$. We then quantify the type of gains discussed here under two precisely defined error models, which correspond to different assumptions about the hardware. By using the aforementioned techniques to synthesise local Trotter steps, exploiting a recent fermion encoding specifically designed for this type of algorithm ~\cite{DK}, deriving tighter error bounds on the higher-order Trotter expansions that account for all constant factors, and carefully analysing analytically and numerically the impact and rate of spread of errors in the resulting algorithm, we improve on this by multiple orders of magnitude even in the presence of decoherence.
For example, we show that a $5\times 5$ Fermi-Hubbard time-dynamics simulation up to time $T=7$ can be performed to $10\%$ accuracy in what we refer to as a per-gate error model with $\approx 50$ qubits and the equivalent of circuit depth $72,308$.
This is a conservative estimate and based on analytic Trotter error bounds that we derive in this paper.
Using numerical extrapolation of Trotter errors, a circuit depth of $3,209$ can be reached.
In the second error model, which we refer to as a per-time error model, we prove rigorously that the same simulation is achievable in a circuit-depth-equivalent run-time of $1,686$; numerical error computations bring this down to $259$.
In the per-time model, for some parameter regimes we are also able to exploit the inherent partial error-detection properties of local fermionic encodings to enable error mitigation strategies to reduce the resource cost.
This brings Hamiltonian simulation, previously beyond reach of current hardware for non-trivial examples, significantly closer to being feasible in the NISQ era.
\section*{Results and Discussion}
\subsection*{Circuit Error Models}
We consider two error models for quantum computation in this work. The first error model assumes that noise occurs at a constant rate per gate, independent of the time it takes to implement that gate.
This is the standard error model in quantum computation theory, in which the cost of a computation is proportional to its circuit depth.
We refer to this model as the per gate error model.
The second error model assumes that noise occurs at a constant rate per unit time.
This is the traditional model of errors in physics, where dissipative noise is more commonly modelled by continuous-time master equations, which translates to the per-time error model.
In this model, the errors accumulate proportionately to the time the interactions in the system are switched on, thus with the total pulse lengths.
We refer to this as the per time error model We emphasise that these error models are not fundamentally about execution time, but about an error budget required to execute a particular circuit.
While it is clear that deeper circuits experience more decoherence, how much each gate contributes to it can be analysed from two different perspectives.
The two error models we study correspond to two difference models of how noise scales in quantum hardware.

Which of these more accurately models errors in practice is hardware-dependent.
For example, in NMR experiments, the per-time model is common~\cite{Khaneja_2001,Khaneja_2002,Khaneja2005}.
The per-time model is not without basis in more recent quantum hardware, too.
Recent work has developed and experimentally tested duration-scaled two-qubit gates using Qiskit Pulse and IBM devices~\cite{ibm-earnest2021,ibm-stenger2021}.
In~\cite{ibm-stenger2021} the authors experimentally observe braiding of Majorana zero modes using and IBM device and parameterized two-qubit gates.
They also find a relationship between relative gate errors and the duration of these parameterised gates which is further validated in~\cite{ibm-earnest2021}.
The authors of~\cite{ibm-earnest2021} explicitly attribute the reduction in error -- seen using these duration-scaled gates in place of CNOT gates -- to the shorter schedules of the scaled gates relative to the coherence time.

Nonetheless, the standard per-gate error model is also very relevant to current quantum hardware hardware.
Therefore, throughout this paper we carry out full error analyses of all our algorithms in both of these error models.

Both of these error models are idealisations.
Both are reasonable from a theoretical perspective and supported by certain experiments.
Analysing both error models allows different algorithm implementations to be compared fairly under different error regimes.
In particular, analysing both of these error models gives a more stringent test of new techniques than considering only the ``standard error model'' of quantum computation, which corresponds to the per gate model.

We show that in both error models, significant gains can be achieved using our new techniques.

In our analysis, for simplicity we treat single-qubit gates as a free resource in both error models.
There are three reasons for making this simplification, Firstly, single-qubit gates can typically be implemented with an order of magnitude higher fidelity in hardware, so contribute significantly less to the error budget than two-qubit gates. Secondly, they do not propagate errors to the same extent as two-qubit gates (cf.\ only costing T gates in studies of fault-tolerant quantum computation). Thirdly, any quantum circuit can be decomposed with at most a single layer of single-qubit gates between each layer of two-qubit gates. Thus including single-qubit gates in the per gate error model changes the absolute numbers by a constant factor $\leq 2$ in the worst case. Nor does it significantly affect comparisons between different algorithm designs. This is particularly true of product formula simulation algorithms, where the algorithms are composed of the same layers of gates repeated over and over.

Additionally, there is a benefit to utilising our synthesis techniques
regardless of error model. Decomposing the simulation into gates of the form  $\ee^{-\ii \op h_{ij}\delta}$ using these methods allows us to exploit the underlying error detection properties of fermionic encodings, as explained in Supplementary Methods and demonstrated in \Cref{fig:cost-delta-5x5} (see below).

\Cref{tab:NumericCost-intro-per-gate,tab:NumericCost-intro-per-time} compare these results, showing how the combination of sub-circuit algorithms, recent advances in fermion encodings (VC~$=$~Verstraete-Cirac encoding~\cite{Verstraete2005}, compact~$=$~encoding reported in~\cite{DK}), and tighter Trotter bounds (both analytic and numeric) successively reduce the run-time of the simulation algorithm. ($\cost=$ circuit-depth for per-gate error model, or sum of pulse lengths for per-time error model.)

\subsection*{Synthesis of Encoded Fermi-Hubbard Hamiltonian Trotter Layers}\label{sec:fermion-encoding-main}

To simulate fermionic systems on a quantum computer, one must encode the fermionic Fock space into qubits.
There are many encodings in the literature~\cite{Jordan1928} but we confine our analysis to two: the Verstraete-Cirac (VC) encoding~\cite{Verstraete2005}, and the compact encoding recently introduced in~\cite{DK}.
We have selected these two encodings as they minimise the maximum Pauli weight of the encoded interactions, which is a key factor in the efficiency of Trotter-based algorithms and of our sub-circuit techniques: weight-$4$ (VC) and weight-$3$ (compact), respectively.
By comparison, the classic Jordan-Wigner transformation~\cite{Jordan1928} results in a maximum Pauli weight that scales as as $\BigO\left(L\right)$ with the lattice size $L$; the Bravyi-Kitaev encoding~\cite{Bravyi2002} has interaction terms of weight $O(\log L)$; and the Bravyi-Kitaev superfast encoding~\cite{Bravyi2002} results in weight-8 interactions.

Under the compact encoding, the fermionic operators in \cref{eq:FH-H-intro} are mapped to operators on qubits arranged on two stacked square grids of qubits (one corresponding to the spin up, and one to the spin down sector, as shown in \cref{fig:intro-onsite}),
augmented by a face-centred ancilla in a checkerboard pattern, with an enumeration explained in \cref{fig:intro-ordering}.

The on-site, horizontal and vertical local terms in the Fermi-Hubbard Hamiltonian \cref{eq:FH-H-intro} are mapped under this encoding to qubit operators as follows:
\begin{align}
     \op h_{\text{on-site}}^{(i)} &\rightarrow \frac{\os}{4} \left(\id - Z_{i \uparrow}\right)\left(\id - Z_{i \downarrow}\right) \\
    \op h_{\text{hopping,hor}}^{(i,j,\sigma)} &\rightarrow \frac{\hop}{2}\left(
     X_{i,\sigma}X_{j,\sigma}Y_{f_{ij}',\sigma}+Y_{i,\sigma}Y_{j,\sigma}Y_{f_{ij}',\sigma}\right) \\
    \op h_{\text{hopping,vert}}^{(i,j,\sigma)} &\rightarrow \frac{\hop}{2}
        (-1)^{g(i,j)}\left( X_{i,\sigma}X_{j,\sigma}X_{f_{ij}',\sigma}+Y_{i,\sigma}Y_{j,\sigma}X_{f_{ij}',\sigma}\right),
\end{align}
where qubit $f'_{ij}$ is the face-centered ancilla closest to vertex $(i,j)$, and $g(i,j)$ indicates an associated sign choice in the encoding, as explained in~\cite{DK}.

If the VC encoding is used, the fermionic operators in \cref{eq:FH-H-intro} are mapped to qubits arranged on two stacked square grids of qubits (again with one corresponding to spin up, the other to spin down, as shown in \cref{fig:intro-VC}),
augmented by an ancilla qubit for each data qubit and with an enumeration explained in \cref{fig:intro-ordering-VC}.
In this case the on-site, horizontal and vertical local terms are mapped to
\begin{align}
     \op h_{\text{on-site}}^{(i)} &\rightarrow \frac{\os}{4} \left(\id - Z_{i \uparrow}\right)\left(\id - Z_{i \downarrow}\right) \\
    \op h_{\text{hopping,hor}}^{(i,j,\sigma)} &\rightarrow \frac{\hop}{2}\left(
     X_{i,\sigma}Z_{i',\sigma}X_{j,\sigma}+Y_{i,\sigma}Z_{i',\sigma}Y_{j,\sigma} \right)\\
    \op h_{\text{hopping,vert}}^{(i,j,\sigma)} &\rightarrow \frac{\hop}{2}\left(
        X_{i,\sigma}Y_{i',\sigma}Y_{j,\sigma}X_{j',\sigma}-Y_{i,\sigma}Y_{i',\sigma}X_{j,\sigma}X_{j',\sigma}\right),
\end{align}
where $i'$ indicates the ancilla qubit associated with qubit $i$.

In both encodings, we partition the resulting Hamiltonian $H$ -- a sum of on-site, horizontal and vertical qubit interaction terms on the augmented square lattice -- into $M=5$ layers $\op H=H_1 + H_2 + H_3 + H_4 + H_5$, as shown in \cref{fig:intro-compact,fig:intro-VC}.
The Hamiltonians for each layer do not commute with one another. Each layer is a sum of mutually-commuting local terms acting on disjoint subsets of the qubits. For instance, $H_5=\sum_i h_\mathrm{on-site}^{(i)}$ is a sum of all the two-local, non-overlapping, on-site terms.


The Trotter product formula $\mathcal P_p(T,\delta)$ comprises local unitaries, corresponding to the local interaction terms that make up the five layers of Hamiltonians that we decomposed the Fermi Hubbard Hamiltonian into.

In order to implement each step of the product formula as a sequence of gates, we would ideally simply execute all two-, three- (for the compact encoding), or four-local (for the VC encoding) interactions necessary for the time evolution directly within the quantum computer.
Yet this is an unrealistic assumption, as the quantum device is more likely to feature a very restricted set of one- and two-qubit interactions.

As outlined in the introduction, we assume in our model that arbitrary single qubit unitaries are available, and that we have access to the continuous family of gates $\{\exp(\ii t Z\otimes Z)\}$ for arbitrary values of $t$.
In contrast, the gates we wish to implement all have the form $\exp(\ii \delta Z^{\otimes k})$ for $k=3$ or~$4$.
(Or different products of $k$ Pauli operators, but these are all equivalent up to local unitaries, which we are assuming are available.)

It is well known that a unitary describing the evolution under any $k$-local Pauli interaction can be straightforwardly decomposed into CNOT gates and single qubit rotations~\cite[Sec.~4.7.3]{Nielsen_and_Chuang}.
For instance, we can decompose evolution under a $3$-local Pauli as
\begin{align}\label{eq:conj-method}
    \ee^{\ii \delta Z_1Z_2Z_3} &= \ee^{-\ii \pi/4 Z_1 X_2} \ee^{\ii \delta Y_2 Z_3} \ee^{\ii \pi/4 Z_1 X_2},
\end{align}
where we then further decompose the remaining $2$-local evolutions in \cref{eq:conj-method} using the exact same method as
\begin{align}
    \ee^{\ii \delta Y_2 Z_3}  &=  \ee^{-\ii \pi/4 Y_2 X_3} \ee^{ \ii \delta Y_3 } \ee^{\ii \pi/4 Y_2 X_3}.
\end{align}
This effectively corresponds to decomposing $\ee^{\ii \delta Z_1Z_2Z_3}$ into CNOT gates and single qubit rotations, as $\ee^{\pm \ii \pi/4 Z_i Z_j}$ is equivalent to a CNOT gate up to single qubit rotations.
To generate evolution under any $k$-local Pauli interaction we can simply iterate this procedure, which yields a constant overhead $\propto 2 (k-1) \times \pi/4$.

Can we do better?
Even optimized variants of Solovay-Kitaev to decompose multi-qubit gates~-- beyond introducing an additional error -- generally yield gate sequences multiple orders of magnitude larger, as e.g.\ demonstrated in~\cite{Pham2013}. While more recent results conjecture that an arbitrary three-qubit gate can be implemented with at most eight $\BigO(1)$ two-local entangling gates~\cite{Martinez2016}, this is still worse than the conjugation method for the particular case of a rank one Pauli interaction that we are concerned with.

For small pulse times $\delta$, the existing decompositions are thus inadequate, as they all introduce a gate cost $\Omega(1) + \BigO(\delta)$.
In this paper, we develop a series of analytic pulse sequence identities (see \cref{lem:Weight-Increase-Depth4-ap,lem:Weight-Increase-Depth5-ap}
in Supplementary Methods, which allow us to decompose the three-qubit and four-qubit gates as approximately The approximations in \cref{eq:intro-gatedec-3,eq:intro-gatedec-4} are shown to first order in $\delta$. Exact analytic expressions, which also hold for $\delta \geq 1$, are derived in Supplementary Methods. The constants in \cref{eq:intro-gatedec-4} have been rounded to the third significant figure.
\begin{align}
    \ee^{\ii \delta Z_1Z_2Z_3} &\approx \ee^{-\ii \sqrt{\delta/2} Z_1X_2} \ee^{\ii \sqrt{\delta/2} Y_2 Z_3} \nonumber \\
    &\hspace{11mm} \times \ee^{\ii \sqrt{\delta/2} Z_1X_2} \ee^{-\ii \sqrt{\delta/2} Y_2 Z_3},
    \label{eq:intro-gatedec-3} \\
    \ee^{\ii \delta Z_1Z_2Z_3Z_4} &\approx \ee^{- \ii 0.22  \delta^{2/3} Y_2 Z_3 Z_4} \ee^{- \ii 1.13 \delta^{1/3} Z_1 X_2} \ee^{ \ii 0.44 \delta^{2/3} Y_2 Z_3 Z_4} \nonumber \\
    &\hspace{11mm} \times \ee^{\ii 1.13 \delta^{1/3} Z_1 X_2} \ee^{- \ii 0.22 \delta^{2/3} Y_2 Z_3 Z_4}.
    \label{eq:intro-gatedec-4}
\end{align}
In reality we use the exact versions of these decompositions, which we also note are still exact for $\delta \geq 1$. The depth-5 decomposition in \cref{eq:intro-gatedec-4} yields the shortest overall run-time when breaking down higher-weight interactions in a recursive fashion, assuming that the remaining 3-local gates are decomposed using an expression similar to \cref{eq:intro-gatedec-3}.
We also carry out numerical studies that indicate that these decompositions are likely to be optimal. (See Supplementary Methods.) 
These circuit decompositions allow us to establish that, for a weight-k interaction term, there exists a pulse sequence which implements the evolution operator for time $\delta$ with an overhead $\propto \delta^{1/(k-1)}$, achieved by recursively applying these decompositions. While we have only made reference to interactions of the form $Z^{\otimes k}$, we remark that this is sufficient as we can obtain any other interaction term of the same weight, for example $ZXZ$, by conjugating $Z^{\otimes k}$ by single qubit rotations, $H$ and $SHS^\dagger$ in this example (where $H$ is a Hadamard and $S$ a phase gate).

For the interactions required for our Fermi-Hubbard simulation, the overhead of decomposing short-pulse gates with this analytic decomposition is $\propto \sqrt\delta$ for any weight-3 interaction term, and $\propto \delta^{1/3}$ for weight-4.
The asymptotic run-time is thus $\BigO(T \delta_0^{w})$ for $w=-1/2$ (compact encoding) or $w=-2/3$ (VC encoding).
We show the exact scaling for $k=3$ and $k=4$ in \cref{fig:intro-gatedec-overhead}, as compared to the standard conjugation method.

\begin{figure*}[htbp]
  \centering
  \includegraphics[width=\linewidth]{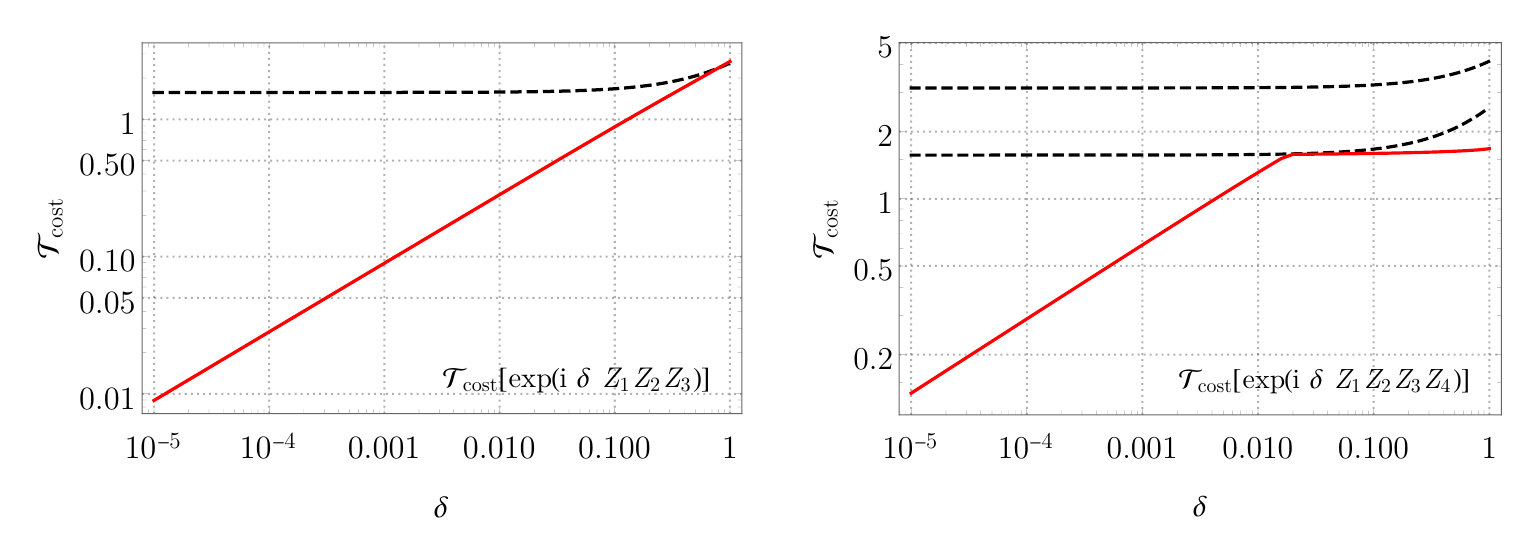}
    \caption{Gate decomposition cost $\cost$ for decomposing $\exp(\ii \delta Z^{\otimes 3})$ (left) and $\exp(\ii \delta Z^{\otimes 4})$ (right), for $\delta\in[10^{-5},1]$. The lower dashed line is the cost obtained by conjugation decomposition, $\pi/2+\delta$. The upper dashed line is the cost for a once-nested conjgation, $\pi+\delta$.
    Decomposing the four-local gate with an outer depth-5 and an inner depth-4 formula according to \cref{eq:intro-gatedec-3,eq:intro-gatedec-4} only saturates the lower conjugation cost bound.
    }
    \label{fig:intro-gatedec-overhead}
\end{figure*}

\subsection*{Tighter Error Bounds for Trotter Product Formulas}\label{sec:Trotter-bounds-main}

There are by now a number of sophisticated quantum algorithms for Hamiltonian simulation, achieving optimal asymptotic scaling in some or all parameters~\cite{Berry2015,Low2016,Berry2014}.
Recently, \cite{Childs2019} have shown that previous error bounds on Trotter product formulae were over-pessimistic.
They derived new bounds showing that the older, simpler, product-formula algorithms achieve almost the same asymptotic scaling as the more sophisticated algorithms.

For near-term hardware, achieving good asymptotic scaling is almost irrelevant; what matters is minimising the actual circuit depth for the particular target system being simulated.
Similarly, in the NISQ regime we do not have the necessary resources to implement full active error-correction and fault-tolerance.
But we can still consider ways of minimising the output error probability for the specific computation being carried out.
Simple product-formula algorithms allow good control of error propagation in the absence of active error-correction and fault-tolerance.
Furthermore, combining product-formula algorithms with our circuit decompositions allows us to exploit the error detection properties of fermionic encodings.
We can use this to relax the effective noise rates required for accurate simulations, especially if we are willing to allow the simulation to include some degree of simulated natural noise.
This is explained further the Supplementary Methods and the results of this technique are shown in \cref{fig:cost-delta-5x5}

For these reasons, we choose to implement the time evolution operator $U(T) \coloneqq \exp(-\ii T H)$ by employing Trotter product formulae $U(T) =: \P_p(\delta)^{T/\delta} + \R_p(T,\delta)$. Here, $\R_p\left(T,\delta\right)$ denotes the error term remaining from the approximate decomposition into a product of individual terms, defined directly as $\R_p\left(T,\delta\right) \coloneqq U(T) - \P_p\left(\delta\right)^{T/\delta}$.
This includes the simple first-order formula~\cite{Lloyd1996a}
\begin{align}
  \P_1\left(\delta\right)^{T/\delta} &\coloneqq \prod_{n=1}^{T/\delta} \prod_{i=1}^M \ee^{-\ii \op H_i \delta}, \\
\intertext{as well as higher-order variants~\cite{Suzuki1992,Suzuki1991,Childs2019}}
    \P_2\left(\delta\right) &\coloneqq \prod_{j=1}^M \ee^{-\ii \op H_j \delta/2} \prod_{j=M}^1 \ee^{-\ii \op H_j \delta/2},  \\
    \P_{2k}\left(\delta\right) &\coloneqq \P_{2k-2}\left(a_k \delta\right)^2\P_{2k-2}\left((1-4a_k) \delta\right) \P_{2k-2}\left(a_k \delta\right)^2
\end{align}
for $k\in\field N$, where the coefficients are given by $a_k \coloneqq1/\left(4-4^{1/\left(2k-1\right)}\right)$.
It is easy to see that, while for higher-order formulas not all pulse times equal $\delta$, they still asymptotically scale as $\Theta(\delta)$.
The product formula $\P_p\left(\delta\right)^{T/\delta}$ then approximates a time-evolution under $U(\delta)^{T/\delta} \approx U(T)$, and it describes the sequence of local unitaries to be implemented as a quantum circuit.

Choosing the Trotter step $\delta$ small means that corrections for every factor in this formula come in at $\BigO \left(\delta^{p+1} \right)$ for $p\in\{1, 2k : k\in\field N\}$. Since we have to perform $T/\delta$ many rounds, the overall error scales roughly as $\BigO \left(T\delta^p \right)$.
Yet this rough estimate is insufficient if we need to calculate the largest-possible $\delta$ for our Hamiltonian simulation.

The Hamiltonian dynamics remain entirely within one fermion number sector, as $H_\mathrm{FH}$ commutes with the total fermion number operator.
Let $\Lambda$ denote the number of fermions present in the simulation, such that $\| H_i|_{\Lambda\ \text{fermions}} \| \le \Lambda$ as shown in
\cref{thm:norm-bound}.
Let $M=5$ denote the number of non-commuting Trotter layers, and set $\epsilon_p(T,\delta) \coloneqq \| \mathcal R_p(T,\delta) \|$, and as shorthand $\epsilon_p(\delta) \coloneqq \epsilon_p(\delta,\delta)$, so that $\epsilon_p(T,\delta) = T/\delta \times \epsilon_p(\delta)$.

To obtain a bound on $\P_{p}\left(\delta\right)$, we apply the variation of constants formula~\cite[Th,~4.9]{RealAnalysis} to $\mathcal R_p(\delta)$, with the condition that $\P_{p}\left(0\right)=\id$, which always holds.
As in~\cite[sec.~3.2]{Childs2019}, for $\delta\ge0$, we obtain
\begin{equation}\label{eq:integral-representation-main}
    \P_{p}\left(\delta\right) = U\left(\delta\right) + \R_{p}\left(\delta\right) = \ee^{- \ii \delta \op H} + \int_{0}^{\delta} \ee^{- \ii \left(\delta-\tau\right)\op H} \op R_{p}\left(\tau\right) d\tau
\end{equation}
where the integrand $\op R_p\left(\tau\right)$ is defined as
\begin{align}
    \op R_{p}\left(\tau\right)\coloneqq\frac{d}{d\tau} \P_{p}\left(\tau\right) -\left(- \ii \op H\right) \P_{p}\left(\tau\right).
\end{align}
Now, if $\P_{p}\left(\delta\right)$ is accurate up to $p$\textsuperscript{th} order -- meaning that $\R_{p}\left(\delta\right) = \BigO\left(\delta^{p+1}\right)$ -- it holds that the integrand $\op R_{p}\left(\delta\right) = \BigO\left(\delta^p\right)$. This allows us to restrict its partial derivatives for all $0\le j\le p-1$ to $\partial_\tau^j \op R_{p}\left(0\right) = 0$. For full details see \cref{rem:order-conditions-ap}
and~\cite[Order Conditions]{Childs2019}.

Then, following~\cite{Childs2019}, we perform a Taylor expansion of $\op R_{p}\left(\tau\right)$ around $\tau=0$, simplifying the error bound $\epsilon_{p}(\delta) \equiv \| \mathcal R_{p}(\delta) \|$ to
\begin{align}
    \epsilon_{p}(\delta)
    &=\left\| \int_{0}^{\delta} \ee^{- \ii \left(\delta-\tau\right)\op H} \op R_{p}\left(\tau\right) \dd\tau \right\|
    \leq  \int_{0}^{\delta} \| \op R_{p}\left(\tau\right) \| \dd\tau\\
    &=\int_0^\delta \bigg( \| \op R_{p}\left(0\right) \| + \| \op R'_{p}\left(0\right) \| \tau + \ldots + \\
    &\hspace{15mm} \| \op R_{p}^{\left(p-1\right)}\left(0\right)\| \frac{\tau^{p-1}}{\left(p-1\right)!}
    + \| \op S_{p}\left(\tau, 0\right) \|
      \bigg) \dd\tau .
\end{align}
Here we use the aforementioned order condition that for all $0\le j\le p-1$ the partial derivatives satisfy $\partial_\tau^j \op R_{p}\left(0\right) = 0$, leaving all but the $p$\textsuperscript{th} or higher remainder terms -- $\op S_{p}\left(\tau, 0\right)$ -- equal to zero.
Thus
\begin{align}
    \epsilon_{p}(\delta) &\le \int_0^\delta \| \op S_{p}\left(\tau, 0\right) \| \dd \tau \nonumber \\
    &= p \int_0^\delta \int_0^1 \left(1-x\right)^{p-1} \| \op R_{p}^{\left(p\right)}\left(x\tau\right)\| \frac{\tau^{p}}{p!} \dd x \dd\tau,
    \label{eq:higher-trotter-error-1-intro}
\end{align}
where we used the integral representation for the Taylor remainder $\op S_{p}\left(\tau,0\right)$.

Motivated by this, we look for simple bounds on the $p$\textsuperscript{th} derivative of the integrand $\|\op R_{p}\left(\tau\right)\|$.
At this point our work diverges from~\cite{Childs2019} by focusing on obtaining bounds on $\|\op R_{p}\left(\tau\right)\|$ which have the tightest constants for NISQ-era system sizes, but which now are not optimal in system size.
(See \cref{fig:bounds-comparison}
and
\cref{lem:trotter-tech1-ap,lem:R-bound-Comms-ap}
in Supplementary Methods for details.)
We derive the following explicit error bounds (see \cref{th:trotter-error-ap,cor:trotter-error-ap}):
\begin{align}
    \epsilon_p(\delta)
    &\le \delta^{p+1} M^{p+1} \Lambda^{p+1} G_p \\
\intertext{where}
    G_p &\coloneqq \times \begin{cases}
        1 & p=1 \\
        \displaystyle \frac{2}{\left(p+1\right)!} \left( \frac{10}3 \right)^{\left(p+1\right)\left(p/2-1\right)} & \text{$p=2k$, $k\ge 1$},
    \end{cases}
    \\
\intertext{and}
    \epsilon_p\left(\delta\right)
    &\le \frac{2 \delta^{p+1} M^{p+1} \Lambda^{p+1} }{\left(p+1\right)!} H_p^{p+1}\\
\intertext{where}
    H_p &\coloneqq \prod_{i=1}^{p/2-1} \frac{ 4+4^{1/\left(2i+1\right)}}{\left| 4 - 4^{1/\left(2i+1\right)} \right| }.
\end{align}
The above expressions hold for generic Trotter formulae.
Using \cref{lem:R-bound-Comms-ap}
we can exploit commutation relations for the specific Hamiltonian at hand (whose structure determines $N$ and $n$, see Supplementary Methods). This yields the bound (see \cref{th:Trotter-Er-Commutator-ap}):
\begin{align}\label{eq:Trotter-Er-Commutator-intro}
    \epsilon_p \left(\delta \right)
    &\le  C_1 \frac{T \delta^{p}}{\left(p+1\right)!} + \nonumber\\
    &\hspace{7mm} C_2 \frac{T}{\delta} \int_0^\delta p \int_0^1 \left(1-x\right)^{p-1}  \frac{x \tau^{p+1}}{p!} \ee^{x \tau N B_p} \dd x \dd\tau \\
\intertext{where}
    C_1 &\coloneqq \n p B_p^{2} \Lambda^{p-1} N \left(M H_p - B_p + B_p \left(\frac{N}{\Lambda}\right)\right)^{p-1} \times \nonumber \\
    &\hspace{10mm}\left(\left(S_p M\right)^{2}-\left(S_p M\right)\right),\\[2mm]
    C_2 &\coloneqq \n  B_p^{2}\left( M H_p \Lambda\right)^{p} N \left(\left(S_p M\right)^{2}-\left(S_p M\right)\right),\\
\intertext{and}
    B_p &\coloneqq \begin{cases}
    1 & p=1 \\
    \frac12 & p=2 \\
    \frac12 \prod_{i=2}^k (1-4a_i) & \text{$p=2k$, $k\ge2$.}
    \end{cases}
\end{align}
These analytic error bounds are then combined with a Taylor-of-Taylor method, by which we expand the Taylor coefficient $R_p^{(p)}$ in \cref{eq:higher-trotter-error-1-intro} itself in terms of a power series to some higher order $q > p$, with corresponding series coefficients $R_p^{(q)}$, and a corresponding remainder-of-remainder error term $\epsilon_{p,q+1}$.
The tightest error expression we obtain is (see \cref{cor:taylor-error-bound-ap}
in Supplementary Methods)
\begin{equation}
    \epsilon_p(\delta) \le \sum_{l=p}^q \frac{\delta^{l+1}\Lambda^{l+1}}{(l+1)!} f(p,M,l)
    + \epsilon_{p,q+1}(\delta),
\end{equation}
where the $f(p,M,l)$ are exactly-calculated coefficients (using a computer algebra package) that exploit cancellations between the $M$ non-commuting Trotter layers, for a product formula of order $p$ and series expansion order $l$ (given in \cref{tab:coefficients-ap}).
The series' remainder $\epsilon_{p,q+1}$ therein is then derived from the analytic bounds in \cref{eq:Trotter-Er-Commutator-intro} (see Supplementary Methods for technical details).

Henceforth, we will assume the tightest choice of $\epsilon_p(\delta)$ amongst all the derived error expressions and choice of $p\in\{1,2,4\}$.
In order to guarantee a target error bound $\epsilon_p(T,\delta)\le \epsilon_\mathrm{target}$, we invert these explicitly derived error bounds and obtain a maximum possible Trotter step $\delta_0 = \delta_0(\epsilon_\mathrm{target})$.

\subsection*{Benchmarking the Sub-Circuit-Model}\label{sec:benchmarking-sub-circuit-model}

How significant is the improvement of the measures set out in previous sections, as benchmarked against state-of-the-art results from literature?
A first comparison is in terms of exact asymptotic bounds (which we derive in \cref{th:FH-optimum-dig,th:FH-optimum-analogue}),
in terms of the number of non-commuting Trotter layers $M$, fermion number $\Lambda$, simulation time $T$ and target error $\epsilon_\mathrm{target}$:

\noindent
Standard Circuit Synthesis:
\begin{align}
    \cost\left(\mathcal P_p(\delta)^{T/\delta}\right) &= \BigO\left(M^{2 + \frac1p} \Lambda^{1 + \frac1p} T^{1+\frac1p} \epsilon^{-\frac1p}_\mathrm{target}\right), \\
\intertext{Sub-circuit Synthesis:}
    \cost\left(\mathcal P_p(\delta)^{T/\delta}\right) &=\BigO\left(
     M^{\frac32 + \frac1{2p}} \Lambda^{\frac12 +\frac1{2p}} T^{1 + \frac1{2p}} \epsilon^{-\frac1{2p}}_\mathrm{target}
    \right).
\end{align}
Here we write $\cost$ for the ``run-time'' of the quantum circuits -- i.e.\ the sum of pulse times of all gates within the circuit. (See \cref{def:time-cost}
for a detailed discussion of the cost model we employ.)

\begin{table}[]
    \centering
    \begin{tabular}{c c c c}
    \toprule
      \parbox[c][][c]{8em}{\centering Fermion encoding}
      & \parbox[c][][c]{8em}{\centering Trotter bounds}
      & \parbox[c][][c]{8em}{\centering Standard decomposition}
      & \parbox[c][][c]{8em}{\centering Subcircuit decomposition} \\
      \midrule
     \multirow{3}{*}{VC}
     &\cite[Prop.~F.4.]{Childs2017} & 1,243,586 & 977,103 \\
      & analytic & 121,478 & 95,447 \\
      & numeric & 5,391 & 4,236 \\
     \midrule
     \multirow{2}{*}{compact}
         & analytic & 98,339 & 72,308 \\
         & numeric & 4,364 & 3,209 \\
    \end{tabular}
    \caption{Per-gate run-times. A comparison of the run-time $\cost$ for lattice size $L\times L$ with $L=5$, overall simulation time $T=7$ and target Trotter error $\epsilon_\mathrm{target} = 0.1$, with $\Lambda=5$ fermions and coupling strengths $|\os|, |\hop|\le r=1$.
    Obtained by minimising over product formulas up to $4$\textsuperscript{th} order.
    $\cost=$~circuit-depth for per-gate error model.
    In either gate decomposition case---standard and sub-circuit---we account single-qubit rotations as a free resource as explained in the Introduction; the value of $\cost$ depends only on the two-qubit gates/interactions. Two-qubit unitaries are counted by unit time per gate in the per gate error model.  Here compact and VC denote the choice of fermionic encoding.}
    \label{tab:NumericCost-intro-per-gate}
\end{table}
\begin{table}[]
    \centering
    \begin{tabular}{c c c c}
      \toprule
      \parbox[c][][c]{8em}{\centering Fermion encoding}
      & \parbox[c][][c]{8em}{\centering Trotter bounds}
      & \parbox[c][][c]{8em}{\centering Standard decomposition}
      & \parbox[c][][c]{8em}{\centering Subcircuit decomposition} \\
    \midrule
     \multirow{3}{*}{VC}
     &\cite[Prop.~F.4.]{Childs2017} & 976,710 & 59,830\\
      & analytic & 95,409 & 17,100 \\
      & numeric & 4,234 & 1,669 \\
     \midrule
     \multirow{2}{*}{compact}
         & analytic & 77,236 & 1,686 \\
         & numeric & 3,428 & 259 \\
    \end{tabular}
    \caption{Per-time run-times. A comparison of the run-time $\cost$ for lattice size $L\times L$ with $L=5$, overall simulation time $T=7$ and target Trotter error $\epsilon_\mathrm{target} = 0.1$, with $\Lambda=5$ fermions and coupling strengths $|\os|, |\hop|\le r=1$.
    Obtained by minimising over product formulas up to $4$\textsuperscript{th} order.
    $\cost=\cost(\P_p(\delta_0)^{T/\delta_0})$ for per-time error model.
    In either gate decomposition case---standard and sub-circuit---we account single-qubit rotations as a free resource; the value of $\cost$ depends only on the two-qubit gates/interactions. Two-qubit unitaries are counted by their respective pulse lengths. Here compact and VC denote the choice of fermionic encoding.}
    \label{tab:NumericCost-intro-per-time}
\end{table}

Beyond asymptotic scaling, and in order to establish a more comprehensive benchmark that takes into account potentially large but hidden constant factors, we employ our tighter Trotter error bounds that account for all constant factors, and
concretely target a $5\times 5$ Fermi-Hubbard Hamiltonian for overall simulation time $T=7$ (which is roughly the Lieb-Robinson time required for the ``causality-cone'' to spread across the whole lattice, and for correlations to potentially build up between any pair of sites), in the sector of $\Lambda=5$ fermions, and coupling strengths $|\os|, |\hop| \le r=1$ as given in \cref{eq:FH-H-intro}.
For this system, we choose the optimal Trotter product formula order $p$ that yields the lowest overall run-time, while still achieving a target error of $\epsilon_\mathrm{target} = 0.1$.

The results are given in \cref{tab:NumericCost-intro-per-time,tab:NumericCost-intro-per-gate}, where we emphasise that in order to maintain a fair comparison, we always account single-qubit gates as a free resource, for the reasons discussed in the Introduction, and two-qubit gates are either accounted at one unit of time per gate in the per-gate error model (making the run-time equal the circuit depth), or accounted at their pulse length for the per-time error model.

Our Trotter error bounds yield an order-of-magnitude improvement as compared to~\cite[Prop~F.4]{Childs2017}.
And even for existing gate decompositions by conjugation, the recently published lower-weight compact encoding yields a small but significant improvement.
The most striking advantage comes from utilising the sub-circuit sequence decompositions developed in this paper, in particular in conjunction with the lower-weight compact fermionic encoding.

Overall, the combination of Trotter error bounds, numerics, compact fermion encoding and sub-circuit-model algorithm design, allows us to improve the run-time of the simulation algorithm from $976,\!710$, to $259$ -- an improvement of more than three orders of magnitude over that obtainable using the previous state-of-the-art methods, and a further improvement over results in the pre-existing literature~\cite{Kivlichan2019}.

\begin{figure}[p]
  \centering
 \includegraphics[width=\linewidth]{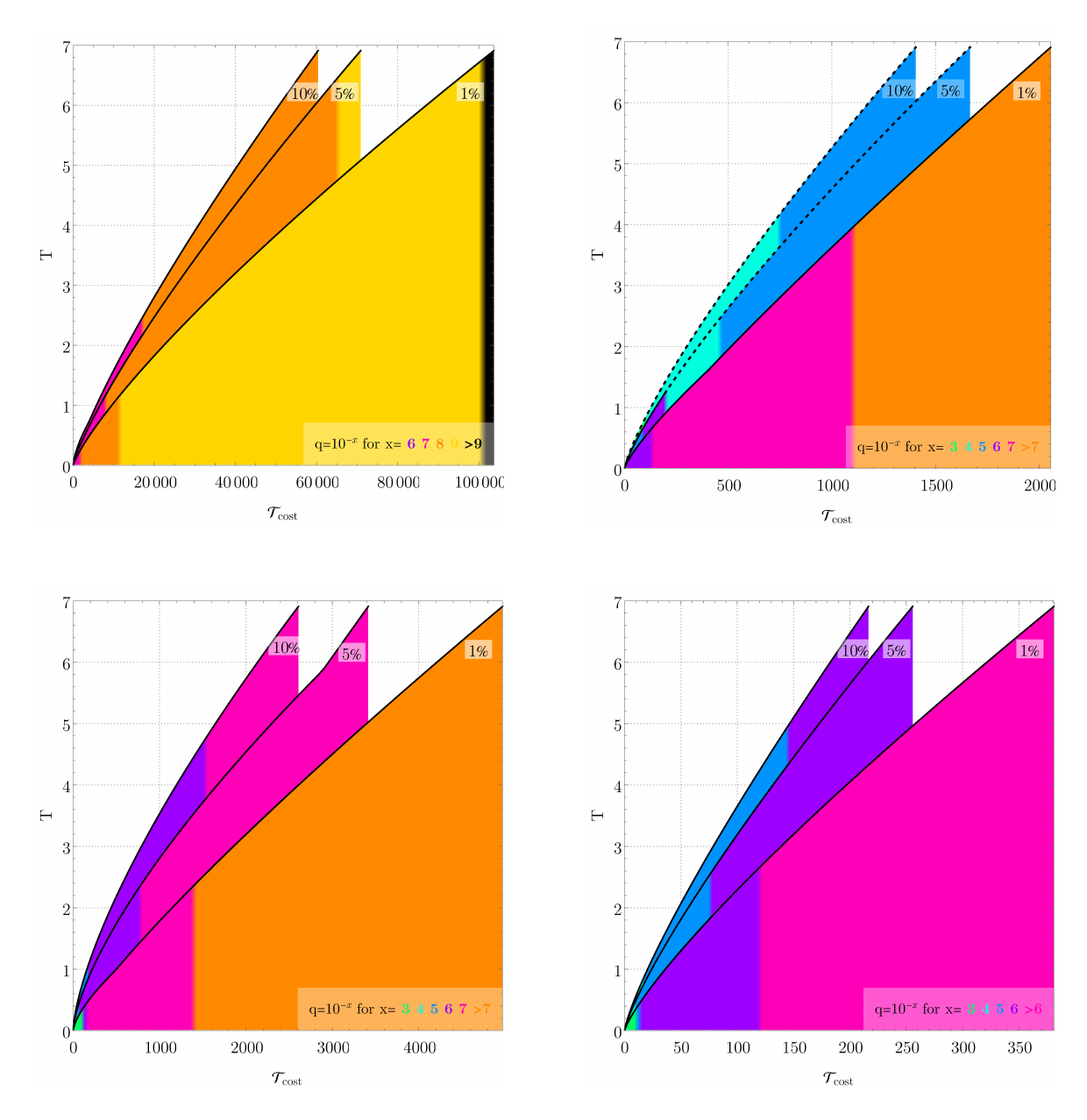}
 \caption{Target simulation time $T$ vs cost $\cost$ for a $5\times5$ lattice FH Hamiltonian $\op H_{\text{FH}}$ from
   \cref{def:unencoded-H}
   using encoding of~\cite{DK}. Per gate (left column) and per time (right column) error models.
   The three lines represent $1\%$, $5\%$, and $10\%$ Trotter error $\epsilon$ given in
  \cref{eq:trotter-epsilon-ap},
    minimized over formula order $p\in\{1,2,4,6\}$.
   Analytic Trotter bounds (top row), get $\delta_0$ from
   \cref{cor:trotter-error-ap,th:Trotter-Er-Commutator-ap,cor:taylor-error-bound-ap};
   numerical bounds (bottom row) by numerical extrapolation (see Supplementary Methods).
    Colors indicate achievable $T$ for a given noise parameter $q$, keeping Trotter and depolarizing errors below the $1\%$, $5\%$ or $10\%$ bound, accordingly. E.g.\ the purple section of the bottom right $1\%$ plot indicates that all $T$ in that range needs $q=10^{-6}$, with Trotter and decoherence error below $1\%$.
    Dashed lines indicate where error mitigation from \cite{DK} can reduce the noise requirements. Additional lattice sizes and details shown in \cref{fig:cost-analytic-delta-ignore-pulse-lengths,fig:cost-analytic-delta,fig:cost-numeric-delta-ignore-pulse-lengths,fig:cost-numeric-delta}
    .}
  \label{fig:cost-delta-5x5}
  \thisfloatpagestyle{empty}
\end{figure}

\subsection*{Sub-Circuit Algorithms on Noisy Hardware}\label{sec:noise-bounds-main}

As ours is a study of quantum simulation on near-term hardware, we cannot neglect decoherence errors that inevitably occur throughout the simulation.
To address this concern, we assume an iid noise model described by the qubit depolarizing channel
\begin{equation}\label{eq:intro-depol}
    \mathcal N_q(\rho) = (1-q)\rho + \frac q3 \big( X \rho X + Y \rho Y + Z \rho Z \big)
\end{equation}
applied to each individual qubit in the circuit, and after each gate layer in the Trotter product formula, such that the bit, phase, and combined bit-phase-flip probability $q$ is proportional to the elapsed time of the preceding layer.
Whilst this standard error model is simplistic, it is a surprisingly good match to the errors seen in some hardware~\cite{GoogleAI}.

Within this setting, a simple analytic decoherence error bound can readily be derived (see Supplementary Methods), by calculating the probability that zero errors appear throughout the circuit.
If $V$ denotes the volume of the circuit (defined as $\cost \times L^2$), then the depolarising noise parameter $q < 1-(1-\epsilon_\mathrm{target})^{1/V}$ -- i.e.\, it needs to shrink exponentially quickly with the circuit's volume.
We emphasise that this is likely a crude overestimate.
As briefly discussed at the start, one of the major advantages of sub-circuit circuits is that, under a short-pulse gate, an error is only weakly propagated due to the reduced Lieb-Robinson velocity (discussed further in~\cite{Error-Mapping}).

Yet irrespective of this overestimate, can we derive a tighter error bound by other means?
In~\cite{Error-Mapping}, the authors analyse how noise on the physical qubits translates to errors in the fermionic code space.
To first order and in the compact encoding, all of $\{X, Y, Z\}$ errors on the face, and $\{X, Y\}$ on the vertex qubits can be detected.
$Z$ errors on the vertex qubits result in an undetectable error, as evident from the form of $\op h_\mathrm{on-site}$ from
\cref{eq:h-onsite}
. It is shown in~\cite[Sec.~3.2]{Error-Mapping} that this $Z$ error corresponds to fermionic phase noise in the simulated Fermi-Hubbard model.

It is therefore a natural extension to the notion of simulation to allow for some errors to occur, if they correspond to physical noise in the fermionic space.
And indeed, as discussed more extensively in~\cite[Sec.~2.4]{Error-Mapping}, phase noise is a natural setting for many fermionic condensed matter systems coupled to a phonon bath~\cite{Ng2015,Kauch2020,Zhao2017,Melnikov2016,Openov2005,Fedichkin2004,Scully1993} and~\cite[Ch.~6.1\&eq.~6.17]{Wellington2014}.

How can we exploit the encoding's error mapping properties?
Under the assumption that $X$, $Y$ and $Z$ errors occur uniformly across all qubits, as assumed in \cref{eq:intro-depol}, each Pauli error occurs with probability $q/3$.
We further assume that we can measure all stabilizers (including a global parity operator) once at the end of the entire circuit, which can be done by dovetailing a negligible depth $4$ circuit to the end of our simulation (see
Supplementary Methods for more details).
We then numerically simulate a stochastic noise model for the circuit derived from aforementioned Trotter formula for a specific target error $\epsilon_\mathrm{target}$, for a Fermi-Hubbard Hamiltonian on an $L\times L$ lattice for $L\in\{3,5,10\}$.

Whenever an error occurs, we keep track of the syndrome violations they induce (including potential cancellations that happen with previous syndromes), using results from~\cite{Error-Mapping} on how Pauli errors translate to error syndromes with respect to the fermion encoding's stabilizers (summarized in
\cref{tab:error-mapping}
).
We then bin the resulting circuit runs into the following categories:
\begin{enumerate}
    \item detectable error: at least one syndrome remains triggered, even though some may have canceled throughout the simulation,
    \item undetectable phase noise: no syndrome was ever violated, and the only errors are $Z$ errors on the vertex qubits which map to fermionic phase noise, and
    \item undetectable non-phase noise: syndromes were at some point violated, but they all canceled.
    \item errors not happening in between Trotter layers: naturally, not all errors happen in between Trotter layers, so this category encompasses all those cases where errors happen in between gates in the gate decomposition.
\end{enumerate}

This categorization allows us to calculate the maximum depolarizing noise parameter $q$ to be able to run a simulation for time $T=\lfloor \sqrt 2 L \rfloor$ with target Trotter error $\epsilon_\mathrm{t} \le \epsilon_\mathrm{target} \in \{ 1\%, 5\%, 10\% \}$, where we allow the resulting undetectable non-phase noise and the errors not happening in between Trotter layers errors to also saturate this error bound, i.e.\ $\epsilon_\mathrm{s}\le\epsilon_\mathrm{target}$.
The overall error is thus guaranteed to stay below a propagated error probability of $( \epsilon_\mathrm{t}^2 + \epsilon_{s}^2)^{1/2} \in \{ 1.5\%, 7.1\%, 15\% \}$, respectively.

In order to achieve these decoherence error bounds, one needs to postselect ``good'' runs and discard ones where errors have occurred, as determined from the single final measurement of all stabilizers of the compact encoding.
The required overhead due to the postselected runs is mild, and shown in
\cref{fig:postsel-probs}.

We plot the resulting simulation cost vs.\ target simulation time in \cref{fig:cost-delta-5x5} and \cref{fig:cost-numeric-delta-ignore-pulse-lengths,fig:cost-analytic-delta-ignore-pulse-lengths,fig:cost-numeric-delta,fig:cost-analytic-delta},
  where we color the graphs according to the depolarizing noise rate required to achieve the target error bound.
For instance, in the tightest per-time error model (bottom right plot in \cref{fig:cost-delta-5x5}), a depolarizing noise parameter $q=10^{-5}$ allows simulating a $5\times5$ FH Hamiltonian for time $T\approx 5$, while satisfying a $15\%$ error bound, the required circuit-depth-equivalent is $\cost\approx 140$---and for time $T\approx 2.5$ for a $7.1\%$ error bound, for $\cost\approx 70$.

In this work, we have derived a method for designing quantum algorithms ``one level below'' the circuit model, by designing analytic sub-circuit identities to decompose the algorithm into.
As a concrete example, we applied these techniques to the task of simulating time-dynamics of the spin Fermi-Hubbard Hamiltonian on a square lattice. Together with Trotter product formulae error bounds applied to the recent compact fermionic encoding, we estimate these techniques provide a three orders of magnitude reduction in circuit-depth-equivalent. The authors of  \cite{Childsnew2019} have recently extended their work on error bounds in \cite{Childsnew2019}, beyond their results in \cite{Childs2019}.
We have not yet incorporated their new bounds into our analysis, and this may give further improvements over our analytic error bounds.

Naturally, any real world implementation on actual quantum hardware will allow and require further optimizations; for instance, all errors displayed within this paper are in terms of operator norm, which indicates the worst-case error deviation for any simulation. However, when simulating time dynamics starting from a specific initial configuration and a distinct final measurement setup, a lower error rate results.
We have accounted for this in a crude way, by analysing simulation of the Fermi-Hubbard model dynamics with initial states of bounded fermion number. But the error bounds -- even the numerical ones -- are certainly pessimistic for any specific computation.
Furthermore, while we already utilize numerical simulations of Trotter errors, more sophisticated techniques such as Richardson extrapolation for varying Trotter step sizes might show promise in improving our findings further.

It is conceivable that other algorithms that require small unitary rotations will similarly benefit from designing the algorithms ``one level below'' the circuit model.
Standard circuit decompositions of many interesting quantum algorithms will remain unfeasible on real hardware for some time to come.
Whereas our sub-circuit-model algorithms, with their shorter overall run-time requirements and lower error-propagation even in the absence of error correction, potentially bring these algorithms and applications within reach of near-term NISQ hardware.

\subsection*{Acknowledgements}
The authors thank Joel Klassen for providing the proof of \cref{thm:norm-bound}
, and for many useful discussions.

\subsection*{Author Contributions}
The authors L. C., J. B. and T. C. contributed equally to this work.

\subsection*{Competing Interests}
The authors declare no competing interests.

\subsection*{Code Availability}
The code to support the findings in this work is available upon request from the authors.

\newpage
\appendix
\part*{Supplementary Figures}

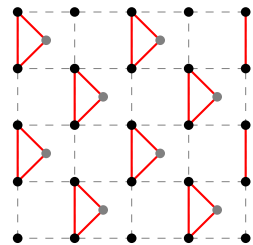
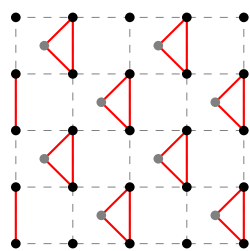
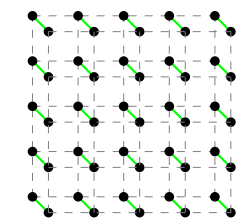
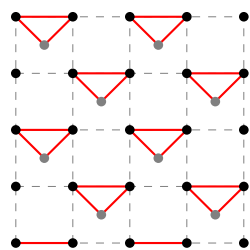
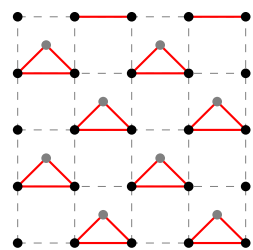
\begin{figure}[htbp]
\begin{subfigure}[b]{0.5\linewidth}
\centering
\begin{tikzpicture}[x=2cm,y=2cm,scale=.65]
 \tikzstyle{every node}=[font=\small]
\draw[dashed][step=2cm,gray, thin] (0,0) grid (2,2);

\draw[solid][step=1cm,black, thick] (0,2) -- (2,2);
\draw[solid][step=1cm,black, thick] (0,1) -- (2,1);
\draw[solid][step=1cm,black, thick] (0,0) -- (2,0);

\draw[solid][step=1cm,black, thick] (2,2) -- (2,1);
\draw[solid][step=1cm,black, thick] (0,0) -- (0,1);

\node[label={[label distance=-2.5]45:1}] at (0,2)[fill=black,circle,scale=0.5]{};
\node[label={[label distance=-2.5, color=gray]45:$f'_{12}$}]at (0.5,1.5)[fill=gray,circle,scale=0.5]{};
\node[label={[label distance=-2.5]45:2}] at (1,2)[fill=black,circle,scale=0.5]{};
\node[label={[label distance=-2.5]45:3}] at (2,2)[fill=black,circle,scale=0.5]{};
\node[label={[label distance=-2.5]45:4}] at (2,1)[fill=black,circle,scale=0.5]{};
\node[label={[label distance=-2.5, color=gray]45:$f'_{54}$}] at (1.5,0.5)[fill=gray,circle,scale=0.5]{};
\node[label={[label distance=-2.5]45:5}] at (1,1)[fill=black,circle,scale=0.5]{};
\node[label={[label distance=-2.5]45:6}] at (0,1)[fill=black,circle,scale=0.5]{};
\node[label={[label distance=-2.5]45:7 $\ldots$}] at (0,0)[fill=black,circle,scale=0.5]{};

\end{tikzpicture}
\caption{Qubit numbering.}\label{fig:intro-ordering}
\end{subfigure}
\vspace{2em}
\begin{subfigure}[b]{0.5\linewidth}
\centering
\begin{tikzpicture}[rotate=90,x=-1cm,scale=.75]

\draw[dashed][step=1cm,gray,very thin] (0,0) grid (4,4);

\foreach \x in {0,2}{
    \foreach \y in {0,2,4}{
        \draw[solid][step=1cm,red,thick] (\x,\y) -- (\x+1,\y);
        }
    }
\foreach \x in {1,3}{
    \foreach \y in {1,3}{
        \draw[solid][step=1cm,red,thick] (\x,\y) -- (\x+1,\y);
        }
    }
\foreach \x in {1.5,3.5}{
    \foreach \y in {0.5,2.5}{
        \draw[solid][step=1cm,red,thick] (\x,\y) -- (\x-0.5,\y+0.5);
         \draw[solid][step=1cm,red,thick] (\x,\y) -- (\x+0.5,\y+0.5);
        }
    }
\foreach \x in {0.5,2.5}{
    \foreach \y in {1.5,3.5}{
        \draw[solid][step=1cm,red,thick] (\x,\y) -- (\x-0.5,\y+0.5);
         \draw[solid][step=1cm,red,thick] (\x,\y) -- (\x+0.5,\y+0.5);
        }
    }
\foreach \x in {1.5,3.5}{
    \foreach \y in {0.5,2.5}{
        \node at (\x,\y)[fill=gray,circle,scale=0.35]{};
        }
    }
\foreach \x in {0.5,2.5}{
    \foreach \y in {1.5,3.5}{
        \node at (\x,\y)[fill=gray,circle,scale=0.35]{};
        }
    }
\foreach \x in {0,1,2,3,4}{
    \foreach \y in {0,1,2,3,4}{
        \node at (\x,\y)[fill=black,circle,scale=0.35]{};
        }
    }
\end{tikzpicture}
\caption{Hopping terms in $\op H_3$.}
\end{subfigure}
\vspace{2em}
\begin{subfigure}[b]{0.5\linewidth}
\centering
\begin{tikzpicture}[rotate=90,x=-1cm,scale=.75]

\draw[dashed][step=1cm,gray,very thin] (0,0) grid (4,4);

\foreach \x in {1,3}{
    \foreach \y in {0,2,4}{
        \draw[solid][step=1cm,red,thick] (\x,\y) -- (\x+1,\y);
        }
    }

\foreach \x in {0,2}{
    \foreach \y in {1,3}{
        \draw[solid][step=1cm,red,thick] (\x,\y) -- (\x+1,\y);
        }
    }

\foreach \x in {1.5,3.5}{
    \foreach \y in {0.5,2.5}{
        \draw[solid][step=1cm,red,thick] (\x,\y) -- (\x-0.5,\y-0.5);
         \draw[solid][step=1cm,red,thick] (\x,\y) -- (\x+0.5,\y-0.5);
        }
    }

\foreach \x in {0.5,2.5}{
    \foreach \y in {1.5,3.5}{
        \draw[solid][step=1cm,red,thick] (\x,\y) -- (\x-0.5,\y-0.5);
         \draw[solid][step=1cm,red,thick] (\x,\y) -- (\x+0.5,\y-0.5);
        }
    }

\foreach \x in {1.5,3.5}{
    \foreach \y in {0.5,2.5}{
        \node at (\x,\y)[fill=gray,circle,scale=0.35]{};
        }
    }

\foreach \x in {0.5,2.5}{
    \foreach \y in {1.5,3.5}{
        \node at (\x,\y)[fill=gray,circle,scale=0.35]{};
        }
    }

\foreach \x in {0,1,2,3,4}{
    \foreach \y in {0,1,2,3,4}{
        \node at (\x,\y)[fill=black,circle,scale=0.35]{};
        }
    }
\end{tikzpicture}
\caption{Hopping terms in $\op H_4$.}
\end{subfigure}
\hfill
\begin{subfigure}[b]{0.5\linewidth}
\centering
\begin{tikzpicture}[scale=0.6]
\draw[dashed][step=1cm,gray,very thin] (0,0) grid (4,4);

\foreach \x in {0,1,2,3,4}{
    \foreach \y in {0,1,2,3,4}{
        \draw[solid][step=1cm,green,thick] (\x,\y) -- (\x+0.35,\y-0.35);
        }
    }

\foreach \x in {0,1,2,3,4}{
    \foreach \y in {0,1,2,3,4}{
        \node at (\x,\y)[fill=black,circle,scale=0.35]{};
        }
    }

\foreach \x in {0.35,1.35,2.35,3.35,4.35}{
    \foreach \y in {-0.35,0.65,1.65,2.65,3.65}{
        \node at (\x,\y)[fill=black,circle,scale=0.35]{};
        }
    }
\begin{scope}[shift={(0.35,-0.35)}]
\draw[dashed][step=1cm,gray,very thin] (0,0) grid (4,4);
\end{scope}
\end{tikzpicture}
\caption{On-site terms in $H_5$.}\label{fig:intro-onsite}
\end{subfigure}
\vspace{2em}
\begin{subfigure}[b]{0.5\linewidth}
\centering
\begin{tikzpicture}[scale=.75]

\draw[dashed][step=1cm,gray,very thin] (0,0) grid (4,4);

\foreach \x in {0,2}{
    \foreach \y in {0,2,4}{
        \draw[solid][step=1cm,red,thick] (\x,\y) -- (\x+1,\y);
        }
    }

\foreach \x in {1,3}{
    \foreach \y in {1,3}{
        \draw[solid][step=1cm,red,thick] (\x,\y) -- (\x+1,\y);
        }
    }

\foreach \x in {1.5,3.5}{
    \foreach \y in {0.5,2.5}{
        \draw[solid][step=1cm,red,thick] (\x,\y) -- (\x-0.5,\y+0.5);
         \draw[solid][step=1cm,red,thick] (\x,\y) -- (\x+0.5,\y+0.5);
        }
    }

\foreach \x in {0.5,2.5}{
    \foreach \y in {1.5,3.5}{
        \draw[solid][step=1cm,red,thick] (\x,\y) -- (\x-0.5,\y+0.5);
         \draw[solid][step=1cm,red,thick] (\x,\y) -- (\x+0.5,\y+0.5);
        }
    }

\foreach \x in {1.5,3.5}{
    \foreach \y in {0.5,2.5}{
        \node at (\x,\y)[fill=gray,circle,scale=0.35]{};
        }
    }

\foreach \x in {0.5,2.5}{
    \foreach \y in {1.5,3.5}{
        \node at (\x,\y)[fill=gray,circle,scale=0.35]{};
        }
    }

\foreach \x in {0,1,2,3,4}{
    \foreach \y in {0,1,2,3,4}{
        \node at (\x,\y)[fill=black,circle,scale=0.35]{};
        }
    }
\end{tikzpicture}
\caption{Hopping terms in $\op H_1$.}
\end{subfigure}
\vspace{2em}
\begin{subfigure}[b]{0.5\linewidth}
\centering
\begin{tikzpicture}[scale=.75]

\draw[dashed][step=1cm,gray,very thin] (0,0) grid (4,4);

\foreach \x in {1,3}{
    \foreach \y in {0,2,4}{
        \draw[solid][step=1cm,red,thick] (\x,\y) -- (\x+1,\y);
        }
    }

\foreach \x in {0,2}{
    \foreach \y in {1,3}{
        \draw[solid][step=1cm,red,thick] (\x,\y) -- (\x+1,\y);
        }
    }

\foreach \x in {1.5,3.5}{
    \foreach \y in {0.5,2.5}{
        \draw[solid][step=1cm,red,thick] (\x,\y) -- (\x-0.5,\y-0.5);
         \draw[solid][step=1cm,red,thick] (\x,\y) -- (\x+0.5,\y-0.5);
        }
    }

\foreach \x in {0.5,2.5}{
    \foreach \y in {1.5,3.5}{
        \draw[solid][step=1cm,red,thick] (\x,\y) -- (\x-0.5,\y-0.5);
         \draw[solid][step=1cm,red,thick] (\x,\y) -- (\x+0.5,\y-0.5);
        }
    }

\foreach \x in {1.5,3.5}{
    \foreach \y in {0.5,2.5}{
        \node at (\x,\y)[fill=gray,circle,scale=0.35]{};
        }
    }

\foreach \x in {0.5,2.5}{
    \foreach \y in {1.5,3.5}{
        \node at (\x,\y)[fill=gray,circle,scale=0.35]{};
        }
    }

\foreach \x in {0,1,2,3,4}{
    \foreach \y in {0,1,2,3,4}{
        \node at (\x,\y)[fill=black,circle,scale=0.35]{};
        }
    }
\end{tikzpicture}
\caption{Hopping terms in $\op H_2$.}
\end{subfigure}

\caption{Compact encoding: qubit enumeration, and five mutually non-commuting interaction layers.}\label{fig:intro-compact}
\end{figure}

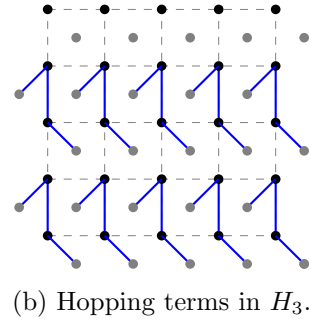
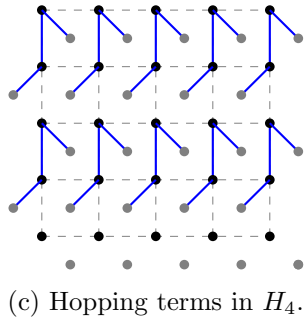
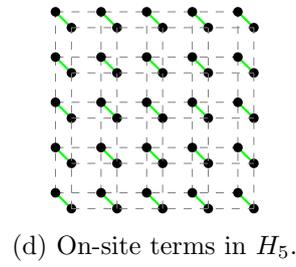
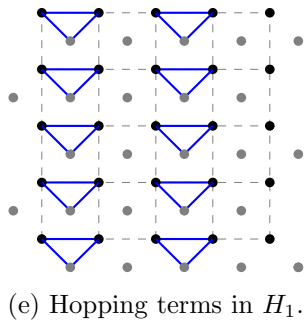
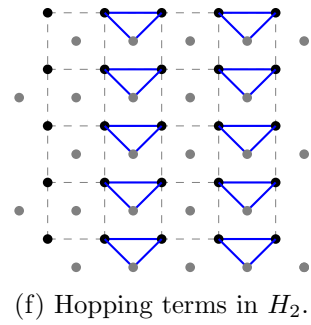
\begin{figure}[htbp]
\begin{subfigure}[b]{0.5\linewidth}
\centering
\begin{tikzpicture}[x=2cm,y=2cm,scale=.65]
\draw[dashed][step=2cm,gray, thin] (0,0) grid (2,2);

\draw[solid][step=1cm,black, thick] (0,2) -- (2,2);
\draw[solid][step=1cm,black, thick] (0,1) -- (2,1);
\draw[solid][step=1cm,black, thick] (0,0) -- (2,0);

\draw[solid][step=1cm,black, thick] (2,2) -- (2,1);
\draw[solid][step=1cm,black, thick] (0,0) -- (0,1);

\node[label={[label distance=0.1mm]45:1}] at (0,2)[fill=black,circle,scale=0.5]{};

\node[label={[label distance=0.1mm, color=gray]45:1'}] at (0.5,1.5)[fill=gray,circle,scale=0.5]{};

\node[label={[label distance=0.1mm]45:2}] at (1,2)[fill=black,circle,scale=0.5]{};

\node[label={[label distance=0.1mm, color=gray]45:2'}] at (1.5,1.5)[fill=gray,circle,scale=0.5]{};

\node[label={[label distance=0.1mm]45:3}] at (2,2)[fill=black,circle,scale=0.5]{};

\node[label={[label distance=0.1mm, color=gray]45:3'}] at (2.5,1.5)[fill=gray,circle,scale=0.5]{};

\node[label={[label distance=0.1mm]45:4}] at (2,1)[fill=black,circle,scale=0.5]{};

\node[label={[label distance=0.1mm, color=gray]45:4'}] at (1.5,0.5)[fill=gray,circle,scale=0.5]{};

\node[label={[label distance=0.1mm]45:5}] at (1,1)[fill=black,circle,scale=0.5]{};

\node[label={[label distance=0.1mm, color=gray]45:5'}] at (0.5,0.5)[fill=gray,circle,scale=0.5]{};

\node[label={[label distance=0.1mm]45:6}] at (0,1)[fill=black,circle,scale=0.5]{};

\node[label={[label distance=0.1mm, color=gray]45:6'}] at (-0.5,0.5)[fill=gray,circle,scale=0.5]{};

\node[label={[label distance=0.1mm]45:7}] at (0,0)[fill=black,circle,scale=0.5]{};

\node[label={[label distance=0.1mm, color=gray]45:7'}] at (0.5,-0.5)[fill=gray,circle,scale=0.5]{};

\node[label={[label distance=0.1mm]45:8}] at (1,0)[fill=black,circle,scale=0.5]{};

\node[label={[label distance=0.1mm, color=gray]45:8'}] at (1.5,-0.5)[fill=gray,circle,scale=0.5]{};

\node[label={[label distance=0.1mm]45:9}] at (2,0)[fill=black,circle,scale=0.5]{};

\node[label={[label distance=0.1mm, color=gray]45:9'}] at (2.5,-0.5)[fill=gray,circle,scale=0.5]{};

\end{tikzpicture}
\caption{Qubit numbering.}\label{fig:intro-ordering-VC}
\end{subfigure}
\vspace{2em}
\begin{subfigure}[b]{0.5\linewidth}
\centering
\begin{tikzpicture}[scale=0.75]

\draw[dashed][step=1cm,gray,very thin] (0,0) grid (4,4);

\foreach \x in {0,1,2,3,4}{
    \foreach \y in {0,2}{
        \draw[solid][step=1cm,blue,thick] (\x,\y) -- (\x,\y+1);
        }
    }

\foreach \x in {0,1,2,3,4}{
    \foreach \y in {0,1,2,3,4}{
        \node at (\x,\y)[fill=black,circle,scale=0.35]{};
        }
    }

\foreach \x in {-0.5,0.5,1.5,2.5,3.5}{
    \foreach \y in {0.5,2.5}{
         \draw[solid][step=1cm,blue,thick] (\x,\y) -- (\x+0.5,\y+0.5);
        }
    }

\foreach \x in {0.5,1.5,2.5,3.5,4.5}{
    \foreach \y in {-0.5,1.5}{
         \draw[solid][step=1cm,blue,thick] (\x,\y) -- (\x-0.5,\y+0.5);
        }
    }

\foreach \x in {0.5,1.5,2.5,3.5}{
    \foreach \y in {-0.5,0.5,1.5,2.5,3.5}{
        \node at (\x,\y)[fill=gray,circle,scale=0.35]{};
        }
    }
\node at (4.5,3.5)[fill=gray,circle,scale=0.35]{};
\node at (-0.5,2.5)[fill=gray,circle,scale=0.35]{};
\node at (4.5,1.5)[fill=gray,circle,scale=0.35]{};
\node at (-0.5,0.5)[fill=gray,circle,scale=0.35]{};
\node at (4.5,-0.5)[fill=gray,circle,scale=0.35]{};
\node at (4.5,-0.5)[fill=gray,circle,scale=0.35]{};
\end{tikzpicture}
\caption{Hopping terms in $\op H_3$.}
\end{subfigure}
\vspace{2em}
\begin{subfigure}[b]{0.5\linewidth}
\centering
\begin{tikzpicture}[scale=0.75]

\draw[dashed][step=1cm,gray,very thin] (0,0) grid (4,4);

\foreach \x in {0,1,2,3,4}{
    \foreach \y in {1,3}{
        \draw[solid][step=1cm,blue,thick] (\x,\y) -- (\x,\y+1);
        }
    }

\foreach \x in {0,1,2,3,4}{
    \foreach \y in {0,1,2,3,4}{
        \node at (\x,\y)[fill=black,circle,scale=0.35]{};
        }
    }

\foreach \x in {-0.5,0.5,1.5,2.5,3.5}{
    \foreach \y in {0.5,2.5}{
         \draw[solid][step=1cm,blue,thick] (\x,\y) -- (\x+0.5,\y+0.5);
        }
    }

\foreach \x in {0.5,1.5,2.5,3.5,4.5}{
    \foreach \y in {1.5,3.5}{
         \draw[solid][step=1cm,blue,thick] (\x,\y) -- (\x-0.5,\y+0.5);
        }
    }

\foreach \x in {0.5,1.5,2.5,3.5}{
    \foreach \y in {-0.5,0.5,1.5,2.5,3.5}{
        \node at (\x,\y)[fill=gray,circle,scale=0.35]{};
        }
    }
\node at (4.5,3.5)[fill=gray,circle,scale=0.35]{};
\node at (-0.5,2.5)[fill=gray,circle,scale=0.35]{};
\node at (4.5,1.5)[fill=gray,circle,scale=0.35]{};
\node at (-0.5,0.5)[fill=gray,circle,scale=0.35]{};
\node at (4.5,-0.5)[fill=gray,circle,scale=0.35]{};
\node at (4.5,-0.5)[fill=gray,circle,scale=0.35]{};

\end{tikzpicture}
\caption{Hopping terms in $\op H_4$.}
\end{subfigure}
\hfill
\begin{subfigure}[b]{0.5\linewidth}
\centering
\begin{tikzpicture}[scale=0.6]
\draw[dashed][step=1cm,gray,very thin] (0,0) grid (4,4);

\foreach \x in {0,1,2,3,4}{
    \foreach \y in {0,1,2,3,4}{
        \draw[solid][step=1cm,green,thick] (\x,\y) -- (\x+0.35,\y-0.35);
        }
    }

\foreach \x in {0,1,2,3,4}{
    \foreach \y in {0,1,2,3,4}{
        \node at (\x,\y)[fill=black,circle,scale=0.35]{};
        }
    }

\foreach \x in {0.35,1.35,2.35,3.35,4.35}{
    \foreach \y in {-0.35,0.65,1.65,2.65,3.65}{
        \node at (\x,\y)[fill=black,circle,scale=0.35]{};
        }
    }
\begin{scope}[shift={(0.35,-0.35)}]
\draw[dashed][step=1cm,gray,very thin] (0,0) grid (4,4);
\end{scope}
\end{tikzpicture}
\caption{On-site terms in $H_5$.}\label{fig:intro-onsite-vc}
\end{subfigure}
\vspace{2em}
\begin{subfigure}[b]{0.5\linewidth}
\centering
\begin{tikzpicture}[scale=0.75]
\draw[dashed][step=1cm,gray,very thin] (0,0) grid (4,4);

\foreach \x in {0,2}{
    \foreach \y in {0,1,2,3,4}{
        \draw[solid][step=1cm,blue,thick] (\x,\y) -- (\x+1,\y);
        }
    }

\foreach \x in {0,1,2,3,4}{
    \foreach \y in {0,1,2,3,4}{
        \node at (\x,\y)[fill=black,circle,scale=0.35]{};
        }
    }
\foreach \x in {0.5,2.5}{
    \foreach \y in {-0.5,0.5,1.5,2.5,3.5}{
        \draw[solid][step=1cm,blue,thick] (\x,\y) -- (\x-0.5,\y+0.5);
         \draw[solid][step=1cm,blue,thick] (\x,\y) -- (\x+0.5,\y+0.5);
        }
    }

\foreach \x in {0.5,1.5,2.5,3.5}{
    \foreach \y in {-0.5,0.5,1.5,2.5,3.5}{
        \node at (\x,\y)[fill=gray,circle,scale=0.35]{};
        }
    }
\node at (4.5,3.5)[fill=gray,circle,scale=0.35]{};
\node at (-0.5,2.5)[fill=gray,circle,scale=0.35]{};
\node at (4.5,1.5)[fill=gray,circle,scale=0.35]{};
\node at (-0.5,0.5)[fill=gray,circle,scale=0.35]{};
\node at (4.5,-0.5)[fill=gray,circle,scale=0.35]{};
\node at (4.5,-0.5)[fill=gray,circle,scale=0.35]{};
\end{tikzpicture}
\caption{Hopping terms in $\op H_1$.}
\end{subfigure}
\vspace{2em}
\begin{subfigure}[b]{0.5\linewidth}
\centering
\begin{tikzpicture}[scale=0.75]

\draw[dashed][step=1cm,gray,very thin] (0,0) grid (4,4);

\foreach \x in {1,3}{
    \foreach \y in {0,1,2,3,4}{
        \draw[solid][step=1cm,blue,thick] (\x,\y) -- (\x+1,\y);
        }
    }

\foreach \x in {0,1,2,3,4}{
    \foreach \y in {0,1,2,3,4}{
        \node at (\x,\y)[fill=black,circle,scale=0.35]{};
        }
    }

\foreach \x in {1.5,3.5}{
    \foreach \y in {-0.5,0.5,1.5,2.5,3.5}{
        \draw[solid][step=1cm,blue,thick] (\x,\y) -- (\x-0.5,\y+0.5);
         \draw[solid][step=1cm,blue,thick] (\x,\y) -- (\x+0.5,\y+0.5);
        }
    }

\foreach \x in {0.5,1.5,2.5,3.5}{
    \foreach \y in {-0.5,0.5,1.5,2.5,3.5}{
        \node at (\x,\y)[fill=gray,circle,scale=0.35]{};
        }
    }
\node at (4.5,3.5)[fill=gray,circle,scale=0.35]{};
\node at (-0.5,2.5)[fill=gray,circle,scale=0.35]{};
\node at (4.5,1.5)[fill=gray,circle,scale=0.35]{};
\node at (-0.5,0.5)[fill=gray,circle,scale=0.35]{};
\node at (4.5,-0.5)[fill=gray,circle,scale=0.35]{};
\node at (4.5,-0.5)[fill=gray,circle,scale=0.35]{};
\end{tikzpicture}
\caption{Hopping terms in $\op H_2$.}
\end{subfigure}
\caption{VC encoding: qubit enumeration, and five mutually non-commuting interaction layers.}\label{fig:intro-VC}
\end{figure}

\newpage

\part*{Supplementary Methods}
\section*{The Sub-Circuit Model and Error Models}\label{sec:cost}
In this section we introduce the sub-circuit model, which we employ throughout this paper. We analyse it under two different error models; these respective models are applicable to NISQ devices with differing capabilities. Before defining these we introduce the mathematical definition of the sub-circuit model.
\begin{definition}[Sub-circuit Model]\label{def:short-pulse-circuit}
  Given a set of qubits $Q$, a set $I \subseteq Q \times Q$ specifying which pairs of qubits may interact, a fixed two qubit interaction Hamiltonian $\op h$, and a minimum switching time $t_{\text{min}}$, a sub-circuit pulse-sequence $\op C$ is a quantum circuit of $L$ pairs of alternating layer types $\op C = \prod_l^L \op U_l \op V_l $ with $\op U_l = \prod_{i \in Q} \op u_i^l$ being a layer of arbitrary single qubit unitary gates, and  $\op V_l = \prod_{ ij  \in \Gamma_l} \op v_{ij} \left(t_{ij}^l \right)$ being a layer of non-overlapping, variable time, two-qubit unitary gates:
  \begin{align}
    \op v_{ij}(t)=\ee^{\ii t \op h_{ij}}
  \end{align}
with the set $\Gamma_l \subseteq I$ containing no overlapping pairs of qubits, and $t \geq t_{min}$. Throughout this paper we assume $\op h_{ij} = Z_iZ_j$. As all $\sigma_i \sigma_j$ are equivalent to $Z_i Z_j$ up to single qubit rotations this can be left implicit and so we take  $\op h_{ij} =\sigma_i \sigma_j$.
\end{definition}



The traditional quantum circuit model measures its run-time in layer count. This also applies in the sub-circuit-model.
\begin{definition}[Circuit Depth] \label{def:cost-circuit depth}
Under a per-gate error model the cost of a sub-circuit pulse-sequence $\op C$ is defined as
\begin{align}
\cost(\op C) := L,
\end{align}
or simply the circuit depth.
\end{definition}

However, unlike the traditional quantum circuit model, the sub-circuit-model also allows for a different run-time metric for any given circuit $C$. Depending on the details of the underlying hardware, it can be appropriate to measure run-time as the total physical duration of the two-qubit interaction layers. This is justified for many implementations: for example superconducting qubits have interaction time scales of $\sim 50 - 500 \text{ns}$ \cite{kjaergaard2019}, while the single qubit energy spacing is on the order of $\sim 5 \text{Ghz}$, which gives a time scale for single qubit gates of $\sim 0.2 \text{ns}$.
\begin{definition}[Run-time] \label{def:time-cost}
The physical run-time of a sub-circuit pulse-sequence $\op C$ is defined as
\begin{align}
\cost(\op C) := \sum_l^L \max_{ij \in \Gamma_l}\left(t_{ij}^l\right)
\end{align}
The run-time is normalised to the physical interaction strength, so that $\vert h\vert = 1$.
\end{definition}

For both run-time and circuit depth we assume single qubit layers contribute a negligible amount to the total time duration of the circuit and we can cost standard gates according to both metrics as long as they are written in terms of a sub-circuit pulse-sequence. For example, according to \Cref{def:time-cost} a CNOT gate has $\cost = \pi/4$ as it is equivalent to $\ee^{-\ii \frac{\pi}{4} ZZ}$ up to single qubit rotations.

How does this second cost model affect the time complexity of algorithms? I.e., given a circuit $\op C$, does $\cost(\op C)$ ever deviate so significantly from $\op C$'s gate depth count that the circuit would have to be placed in a complexity class lower?
Under reasonable assumptions on the shortest pulse time we prove in the following that this is not the case.
\begin{remark}\label{rem:cost-model-overhead}
Let $\{ \op C_x \}_{x\in\field N}$ be a family of quantum circuits ingesting input of size $x$.
Denote with $m(x)$ the circuit depth of $\op C_x$; and let $\delta_0 = \delta_0(x):=\min_{l\in[L]}\max_{ij\in\Gamma_l}(\tau^t_{ij})$ be the shortest layer time pulse present in the circuit $\op C_x$, according to \cref{def:time-cost}.
Then if
\begin{align}
    \delta_0(x) = \begin{cases}
    \BigO(1) \\
    1/\poly x \\
    1/\exp(\poly x) \\
    \end{cases}
    \quad\Longrightarrow\quad
    m(x) = \cost(\op C) \times \begin{cases}
    \BigO(1) \\
    \BigO(\poly x) \\
    \BigO(\exp(\poly x))
    \end{cases}
\end{align}
Furthermore $\cost(\op C) = \BigO(m(x))$.
\end{remark}
\begin{proof}
Clear since $m(x) = \BigO(\cost(\op C)/\delta_0(x))$.
The second claim is trivial.
\end{proof}
An immediate consequence of using the cost model metric and the overhead of counting gates from \cref{rem:cost-model-overhead} can be summarised as follows.
\begin{corollary}
Let $\epsilon>0$. Any family of short-pulse circuits $\{ \op C_x \}$ with $\delta_0(x) = \BigO(1)$ can be approximated by a family of circuits $\{ \tilde{\op C}_x \}$ made up of gates from a fixed universal gate set; and such that $\tilde{\op C}_x$ approximates $\op C_x$ in operator norm to precision $\epsilon$ in time $\BigO(\log^4(\cost(\op C_x)/\epsilon))$.
\end{corollary}
\begin{proof}
By \cref{rem:cost-model-overhead}, there are $m(x) = \cost(\op C)\times\BigO(1)$ layers of gates in $\op C$; now apply Solovay-Kitaev to compile it to a universal target gate set.
\end{proof}
Indeed, we can take this further and show that complexity classes like BQP are invariant under an exchange of the two metrics ``circuit depth'' and ``$\cost$''; if e.g.\ $\delta_0(x) = 1/\poly x$, then again invoking Solovay-Kitaev lets one upper-bound and approximate any circuit while only introducing an at most poly-logarithmic overhead in circuit depth.
However, a stronger result than this is already known, independent of any lower bound on pulse times, which we cite here for completeness.
\begin{remark}[Poulin et. al. \cite{Poulin2011}]\label{rem:bqp-invariant-under-cost}
A computational model based on simulating a local Hamiltonian with arbitrarily-quickly varying local terms is equivalent to the standard circuit model.
\end{remark}
\newcommand{\Ucirc}{\op U_\mathrm{circ}}
\section*{Sub-Circuit Synthesis of Multi-Qubit Interactions}\label{ap:decomp-deets-ap}
\subsection*{Analytic Pulse Sequence Identities}
In this section we introduce the analytic pulse sequence identities we use to decompose local Trotter steps $\ee^{- \ii \delta h}$. Their recursive application allows us to establish, that for a $k$-qubit Pauli interaction $\op h$, there exists sub-circuit pulse-sequence  $C:=\prod_l^{L} U_l V_l$ which implements the evolution operator $\ee^{-\ii \delta \op h}$. Most importantly, for any target time $\delta \geq 0$ the run-time of that circuit is bounded as
\begin{align}\label{eq:general-k-cost}
\cost \left(C \right) \leq \mathcal{O}\left( \delta^{\frac{1}{k-1}}\right),
\end{align}
according to the notion of run-time established in \Cref{def:time-cost}.

For $k=2^n +1$ where $n \in \mathbb{Z}$ as noted by \cite{Dur2007}, this can be done inexactly using a well know identity from Lie algebra. For Hermitian operators $\op A$ and $\op B$ we have
\begin{align}
\ee^{-\ii t \op B}\ee^{-\ii t \op A}\ee^{\ii t \op B}\ee^{\ii t \op A} = \ee^{ t^2 \left[\op A, \op B \right]} + \mathcal{O}\left( t^3 \right).
\end{align}
We make this exact for all $t \in [0,2\pi]$ for anti-commuting Pauli interactions in \Cref{lem:Weight-Increase-Depth4-ap}, and use \Cref{lem:Weight-Increase-Depth5-ap} to extend it to all $k \in \mathbb{Z}$.
\begin{lemma}[Depth 4 Decomposition]\label{lem:Weight-Increase-Depth4-ap}
  Let $\op U(t) = \ee^{\ii t \op H}$ be the time-evolution operator for time $t$ under a Hamiltonian $\op H = \tfrac{1}{2\ii}[ \op h_1, \op h_2]$, where
  $\op h_1$ and $\op h_2$ anti-commute and both square to identity.
  For $0 \leq \delta \leq \pi/2$ or $\pi \leq \delta \leq 3\pi/2$, $\op U(t)$ can be decomposed as
  \begin{align}
    \op U(t) =\ee^{\ii t_1 \op h_1} \ee^{\ii t_2 \op h_2} \ee^{\ii t_2 \op h_1} \ee^{\ii t_1 \op h_2}
  \end{align}
  with pulse times $t_1,t_2$ given by
  \begin{align}
    t_1 &=\frac{1}{2} \tan^{-1}\left(\frac{1}{\sin (t)+\cos (t)}, \; \pm\frac{\sqrt{\sin (2 t)}}{\sin (t)+\cos (t)}\right)+\pi  c\\
    t_2 &=\frac{1}{2} \tan^{-1}\left(\cos (t)-\sin (t), \; \mp\sqrt{\sin (2 t)}\right)+\pi  c,
  \end{align}
  where $c \in \mathbb{Z}$, and corresponding signs are taken in the two expressions.

  For $\pi/2 \leq t \leq \pi$ or $3\pi/2 \leq t \leq 2 \pi$, $\op U(t)$ can be decomposed as
  \begin{align}
    \op U(t) & =\ee^{\ii t_1 \op h_1} \ee^{\ii t_2 \op h_2} \ee^{-\ii t_2 \op h_1} \ee^{-\ii t_1 \op h_2}
  \end{align}
  with pulse times $t_1,t_2$ given by
  \begin{align}
    t_1 &=\frac{1}{2} \tan^{-1}\left(\frac{1}{\cos (t)-\sin (t)}, \; \pm\frac{\sqrt{-\sin (2 t)}}{\cos (t)-\sin (t)}\right)+\pi  c\\
    t_2 &=\frac{1}{2} \tan^{-1}\left(\sin (t)+\cos (t), \; \pm\sqrt{-\sin (2 t)}\right)+\pi  c,
  \end{align}
  where $c \in \mathbb{Z}$, and corresponding signs are taken in the two expressions.
\end{lemma}
\begin{proof}
  Follows similarly to \cref{lem:Rank-Increase-Depth3-ap}.
\end{proof}

\begin{lemma}[Depth 5 Decomposition]\label{lem:Weight-Increase-Depth5-ap}
  Let $\op U(t)$ be the time-evolution operator for time $t$ under a Hamiltonian $\op H = \tfrac{1}{2\ii} [ \op h_1,\op h_2]$.
  If $\op h_1$ and $\op h_2$ anti-commute and both square to identity, then $U(t)$ can be decomposed as
  \begin{align}
    \op U(t) & =\ee^{\ii t_1 \op h_2} \ee^{- \ii \phi \op h_1} \ee^{\ii t_2 \op h_2} \ee^{ \ii \phi \op h_1} \ee^{\ii t_1 \op h_2}
  \end{align}
  with pulse times $t_1,t_2,\phi$ given by
  \begin{align}
    t_1 &= \frac{1}{2} \tan^{-1}\left(\pm\sqrt{2} \sec (t) \csc (2 \phi ) \sqrt{\cos (2 t)-\cos (4 \phi )}, \; -2 \tan (t) \cot (2 \phi )\right)+\pi  c \\
    t_2 &= \tan^{-1}\left(\pm\frac{\csc (2 \phi ) \sqrt{\cos (2 t)-\cos (4 \phi )}}{\sqrt{2}}, \; \sin (t) \csc (2 \phi )\right)+2 \pi  c
  \end{align}
  where $c \in \mathbb{Z}$, and corresponding signs are taken in the two expressions.
\end{lemma}
\begin{proof}
  Follows similarly to \cref{lem:Rank-Increase-Depth3-ap}.
\end{proof}

\subsection*{Pulse-Time Bounds on Analytic Decompositions}
In our later analysis we apply these methods to the interactions in the Fermi-Hubbard Hamiltonian. Depending on the fermionic encoding used, these interaction terms are at most $3$-local or $4$-local. \Cref{fig:circuits} depicts exactly how \cref{lem:Weight-Increase-Depth4-ap,lem:Weight-Increase-Depth5-ap} are used to decompose  3-local and 4-local interactions of the form $Z^{\otimes k}$.

We establish bounds on the run-time (\cref{def:time-cost}) of these circuits. The exact run-time of the circuit $C_a$ -- defined in \cref{fig:circuits-a} -- follows directly from \cref{def:time-cost} as
\begin{align}
\cost(C_a(t)) =  2  |t^a_1(t)| +  2 |t^a_2(t)|.
\end{align}
We have labelled the functions $t_i(t)$ from \cref{lem:Weight-Increase-Depth4-ap} as $t^a_i(t)$ in order to distinguish them from those given in \cref{lem:Weight-Increase-Depth5-ap}, which are now labelled $t^b_i(t)$. This is to avoid confusion when using both identities in the one circuit, such as in circuit $C_b$ where we use \cref{lem:Weight-Increase-Depth4-ap} to decompose the remaining 3-local gates.

The exact run-time of the circuit $C_b$ -- defined in \cref{fig:circuits-b} -- is left in terms of $\cost(C_a)$ and again follows directly from \cref{def:time-cost} as
\begin{align}
\cost(C_b(t,\phi)) = 2 \cost(C_a(t^b_1(t,\phi))) + \cost(C_a(t^b_2(t,\phi))) +2 |\phi| .
\end{align}
\Cref{lem:Weight-Increase-Depth4-param-bounds-ap,lem:Weight-Increase-Depth5-param-bounds-ap} bound these two functions and determine the optimal choice of the free pulse-time $\phi$. Inserting these bounds into the above $\cost$ expressions gives
\begin{equation}\label{eq:decomp-bounds}
\cost(C(t)) \leq \begin{cases}
            2 \sqrt{2 t} & C=C_a \\
            7 \sqrt[3]{t} & C=C_b
           \end{cases}.
\end{equation}
As $Z^{\otimes k}$ is equivalent to any $k$-local Pauli term up to single qubit rotations, these bounds hold for any three or four local Pauli interactions.

\begin{figure}[t]
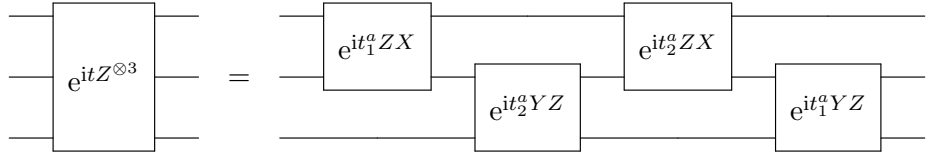
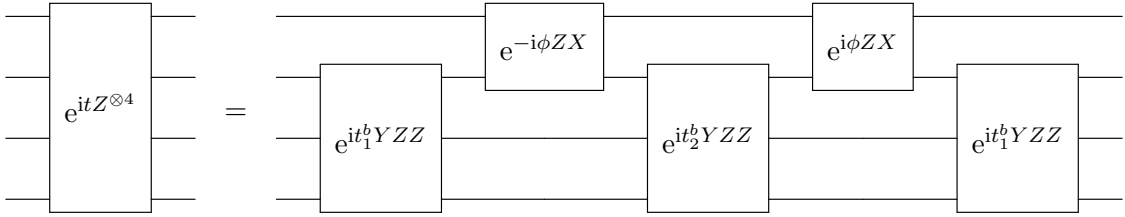

\begin{subfigure}[b]{\textwidth}
\[
\circuit{
& \multigate{2}{ \ee^{\ii t Z^{\otimes 3}}} & \qw \\
& \ghost{\ee^{\ii t Z^{\otimes 3}}} & \qw \\
& \ghost{\ee^{\ii t Z^{\otimes 3}}} & \qw
}
\quad\raisebox{-.82cm}{$=$} \quad
\circuit{
& \multigate{1}{\ee^{ \ii  t^a_1 Z X}} & \qw & \multigate{1}{\ee^{\ii  t^a_2 Z X}} & \qw & \qw \\
& \ghost{\ee^{\ii t^a_1 Z X}} & \multigate{1}{\ee^{\ii  t^a_2 Y Z}} & \ghost{\ee^{\ii  t^a_1 Z X}} & \multigate{1}{\ee^{\ii t^a_1 Y Z}} & \qw \\
& \qw & \ghost{\ee^{\ii t^a_2 Y Z}} & \qw & \ghost{\ee^{\ii t^a_1 Y Z}} & \qw
}
\]
\caption{Circuit $C_a$. The pulse times $t^a_i(t)$ are defined in \cref{lem:Weight-Increase-Depth4-ap} and the run-time is bounded as $\cost(C_a(t))  \leq 2 \sqrt{2 t}$. }\label{fig:circuits-a}
\end{subfigure}
\begin{subfigure}[b]{\textwidth}
\[
\circuit{
& \multigate{3}{ \ee^{\ii t Z^{\otimes 4}}} & \qw \\
& \ghost{\ee^{\ii t Z^{\otimes 4}}} & \qw \\
& \ghost{\ee^{\ii t Z^{\otimes 4}}} & \qw \\
& \ghost{\ee^{\ii t Z^{\otimes 4}}} & \qw
}
\quad\raisebox{-1.26cm}{$=$} \quad
\circuit{
& \qw & \multigate{1}{\ee^{-\ii \phi Z X}} & \qw & \multigate{1}{\ee^{\ii \phi Z X}} & \qw & \qw \\
& \multigate{2}{\ee^{ \ii t^b_1 Y Z Z}}  & \ghost{\ee^{-\ii  \phi Z X}} & \multigate{2}{\ee^{\ii t^b_2 Y Z Z}} & \ghost{\ee^{\ii  \phi Z X}} & \multigate{2}{\ee^{ \ii t^b_1 Y Z Z}} & \qw \\
& \ghost{\ee^{ \ii t^b_1 Y Z Z}} & \qw & \ghost{\ee^{\ii  t^b_2 Y Z Z}} & \qw & \ghost{\ee^{\ii t^b_1 Y Z Z}} & \qw \\
& \ghost{\ee^{ \ii t^b_1 Y Z Z}} & \qw & \ghost{\ee^{\ii  t^b_2 Y Z Z}} & \qw & \ghost{\ee^{\ii t^b_1 Y Z Z}} & \qw
}
\]
\caption{Circuit $C_b$. The pulse times $t^b_i(t)$ are defined in \cref{lem:Weight-Increase-Depth5-ap}. The three local gates are further decomposed using \cref{lem:Weight-Increase-Depth4-ap}, though this is not shown here. If $\phi = \left(\frac{1}{4}(3+2\sqrt{2}) t\right)^{1/3}$ then the run-time is bounded as $\cost(C_b(t))  \leq 7 \sqrt[3]{t}$.}\label{fig:circuits-b}
\end{subfigure}
\caption{The definitions of circuits $C_a(t)$ and $C_b(t)$ -- which respectively generate evolution under a three and four local Pauli interaction for target time $t \geq 0$.}\label{fig:circuits}
\end{figure}

\begin{lemma}\label{lem:Weight-Increase-Depth4-param-bounds-ap}
  Let $\op H$ be as in \cref{lem:Weight-Increase-Depth4-ap}.
  For $0 \leq t \leq \pi/2$, the pulse times $t_i$ in \cref{lem:Weight-Increase-Depth4-ap} can be bounded by
  \begin{align}
    |t_1| & \leq \sqrt{\frac{t}{2}}\\
    |t_2| + |t_1|  & \leq  \sqrt{2t}.
  \end{align}
\end{lemma}
\begin{proof}
Choosing the negative $t_1(t)$ and corresponding $t_2(t)$ solution from \cref{lem:Weight-Increase-Depth4-ap}
and Taylor expanding about $t=0$ gives
  \begin{align}
    t_1(t) &= 
          -\sqrt{\frac{t}{2}} + R_1(t)\\
    t_2(t) &= 
          \sqrt{\frac{t}{2}} + R_2(t).
  \end{align}
Basic calculus shows that $t_1$ is always negative and $t_2$ is always positive for $0 \leq t \leq \pi/2$, thus
\begin{align}
    |t_2| + |t_1| = t_2 - t_1
    = \sqrt{2t} + R_{12}(t).
\end{align}
Then it can be shown that the Taylor remainders $R_1$ and $R_{12}$ are positive and negative, respectively, giving the stated bounds.
\end{proof}

\begin{lemma}\label{lem:Weight-Increase-Depth5-param-bounds-ap}
Let $\op H$ be as in \cref{lem:Weight-Increase-Depth5-ap}. For  $0 \leq t \leq t_c$ and $\phi = (c t)^{1/3}$, the pulse times $t_i$ in \Cref{lem:Weight-Increase-Depth5-ap} can be bounded by
\begin{align}
    2\sqrt{2|t_2|} + 4 \sqrt{2|t_1|} + 2 |\phi|  & \leq  3 (6 + 4 \sqrt{2})^{1/3} t^{1/3} \\
     &\leq 7 t^{1/3},
\end{align}
where $c=\frac{1}{4}(3+2\sqrt{2})$ and $t_c \sim 0.33$.
\end{lemma}
\begin{proof}
This follows similarly to \Cref{lem:Weight-Increase-Depth4-param-bounds-ap}.  We choose the positive branch of the $\pm$ solutions for pulse times with $t_1(t)$ and $t_2(t)$ given in \Cref{lem:Weight-Increase-Depth5-ap},
and freely set $\phi=(ct)^{1/3}$ for some positive constant $c \in \mathbb{R}$. Within the range $0 \leq t \leq t_c$ we have real pulse times $t_1 \leq 0$ and $t_2 \geq 0$. We can then Taylor expand the following about $t=0$ to find
\begin{align}
    2\sqrt{2|t_2|} + 4 \sqrt{2|t_1|} + 2 |\phi|   &= 2\sqrt{2t_2} + 4 \sqrt{-2t_1} + 2(ct)^{1/3} \\
    &=\frac{2  \left(\sqrt{c}+\sqrt{2}+1\right)}{\sqrt[6]{c}} t^{1/3}+ R(t).
\end{align}
Choosing $c$ to minimise the first term in this expansion, and again showing that $R \leq 0$, leads to the stated result
\begin{align}
    2\sqrt{2|t_2|} + 4 \sqrt{2|t_1|} + 2 |\phi|  & \leq  3 (6 + 4 \sqrt{2})^{1/3} t^{1/3}\\
    &\leq 7 t^{1/3}
\end{align}
where $\phi = (c t)^{1/3}$ and $c=\frac{1}{4}(3+2\sqrt{2})$. This is valid only with the region $0 \leq t \leq t_c$ where $t_c \approx 0.33$.
\end{proof}
\begin{theorem}
   For a set of qubits $Q$, a set $I \subseteq Q \times Q$ specifying which pairs of qubits may interact, and a fixed two qubit interaction Hamiltonian $h_{ij}$, if $H$ is a $k$-body Pauli Hamiltonian then the following holds:
   
   For all $t$ there exists a quantum circuit of $L$ pairs of alternating layer types $\op C = \prod_l^L \op U_l \op V_l $ with $\op U_l = \prod_{i \in Q} \op u_i^l$ being a layer of arbitrary single qubit unitary gates, and  $\op V_l = \prod_{ ij  \in \Gamma_l} \op v_{ij} \left(t_{ij}^l \right)$ being a layer of non-overlapping, variable time, two-qubit unitary gates $\op v_{ij}(t)=\ee^{\ii t \op h_{ij}}$
with the set $\Gamma_l \subseteq I$ containing no overlapping pairs of qubits such that $C = \ee^{\ii t H}$ and
   \begin{align}
   \cost(C) \leq \BigO \left(|t|^{\frac{1}{k-1}} \right),
   \end{align}
   where 
   \begin{align}
   \cost(\op C) := \sum_l^L \max_{ij \in \Gamma_l}\left(t_{ij}^l\right).
   \end{align}
\end{theorem}
\begin{proof}
The proof of this claim follows from first noting that for any $t<0$ one can conjugate $\ee^{- \ii t H}$ by a single Pauli operator which anti-commutes with $H$ in order to obtain $\ee^{ \ii t H}$. Therefore we can consider w.l.o.g.  we can take $t>0$ as we have done up until now.

The sub-circuit $C$ which implements $\ee^{ \ii t H}$ is constructed recursively using the Depth $5$ decomposition. We note that the Depth $5$ decomposition has an important feature. The free choice of $\phi$ allows us to avoid incurring a fixed root overhead with every iterative application of this decomposition. That is when using it to decompose any $\ee^{\ii t Z^{\otimes k}}$, we can always choose $h_1$ as  a 2-local interaction and $h_2$ as a $(k-1)$-local interaction. We can choose $\phi \propto  t^{\frac{1}{k-1}}$ and a similar analysis as in \Cref{lem:Weight-Increase-Depth5-param-bounds-ap} will show that this leaves the remaining pulse-times as $t_i \propto  t^{1-\frac{1}{k-1}}$. This can be iterated to decompose the remaining gates, all of the form of evolution under $(k-1)$-local interactions for times $\propto t^{1-\frac{1}{k-1}}$. At each iteration we choose to $h_1$ as a 2-local interaction and $\phi \propto  t^{\frac{1}{k-1}}$. Hence after $k-2$ iterations we will have established the claim that $ \cost(C) \leq \BigO \left(|t|^{\frac{1}{k-1}} \right)$.
\end{proof}

\subsection*{Optimality}\label{ap:pulse-sequence-optimality}
An obvious question to ask at this point is whether the proposed decompositions are optimal, in the sense that they minimise the total run-time $\cost$ while reproducing the target gate $\op h$ exactly.
A closely related question is then whether relaxing the condition that we want to simulate the target gate without any error allows us to reduce the scaling of $\cost$ with regards to the target time $\delta$.

In this section we perform a series of numerical studies which indicate that the exact decompositions described in this section are indeed optimal within some parameter bounds, and that relaxing the goal to approximate implementations gives no benefit.

The setup is precisely as outlined in
before: for $\op U_\mathrm{target}=\exp(\ii T Z^{\otimes k})$ for some locality $k>1$ and time $T>0$, we iterate over all possible gate sequences of width $k$ and length $n$, the set of which we call $U_{n,k}$.
For each sequence $\op U \in U_{n,k}$, we perform a grid search over all parameter tuples $(t_1,\ldots,t_n)\in [-\pi/2,\pi/2]^n$ and $\delta\in[0,\pi/10]$, and calculate the parameter tuple $(\epsilon(\op U), \cost)$, where $\cost$ is given in \cref{def:time-cost}, and
\begin{align}
    \epsilon(\op U) := \left\| \op U - U_\mathrm{target} \right\|_2.
\end{align}
The results are binned into brackets over $(\delta,\cost) \in [\pi/10, n\pi / 2]$ and their minimum within each bracket is taken.
This procedure yields two outcomes:
\begin{enumerate}
    \item For each target time $\delta$ and each target error $\epsilon>0$, it yields the smallest $\cost$, depth $n$ circuit with error less than $\epsilon$, and
    \item for each target time $\delta$ and each $\cost$, the smallest error possible with any depth $n$ gate decomposition and total pulse time less than $\cost$.
\end{enumerate}

This algorithm scales exponentially both in $k$ and $n$, and polynomial in the number of grid search subdivisions.
The following optimisations were performed.
\begin{enumerate}
    \item We remove duplicate gate sequences under permutations of the qubits (since $\op U_\mathrm{target}$ is permutation symmetric).
    \item We restrict ourselves to two-local Pauli gates, since any one-local gate can always be absorbed by conjugations, and
    \item We remove mirror-symmetric sequences (since Paulis are self-adjoint).
    \item For $n>4$ we switch to performing a random sampling algorithm instead of grid search, since the number of grid points becomes too large.
\end{enumerate}

Results for $k=3$ and $n=3,4,5$ are plotted in \cref{fig:depth-3-4-numerics,fig:depth-5-numerics}.
\begin{figure}[t]
    \centering
    \includegraphics[width=.48\textwidth]{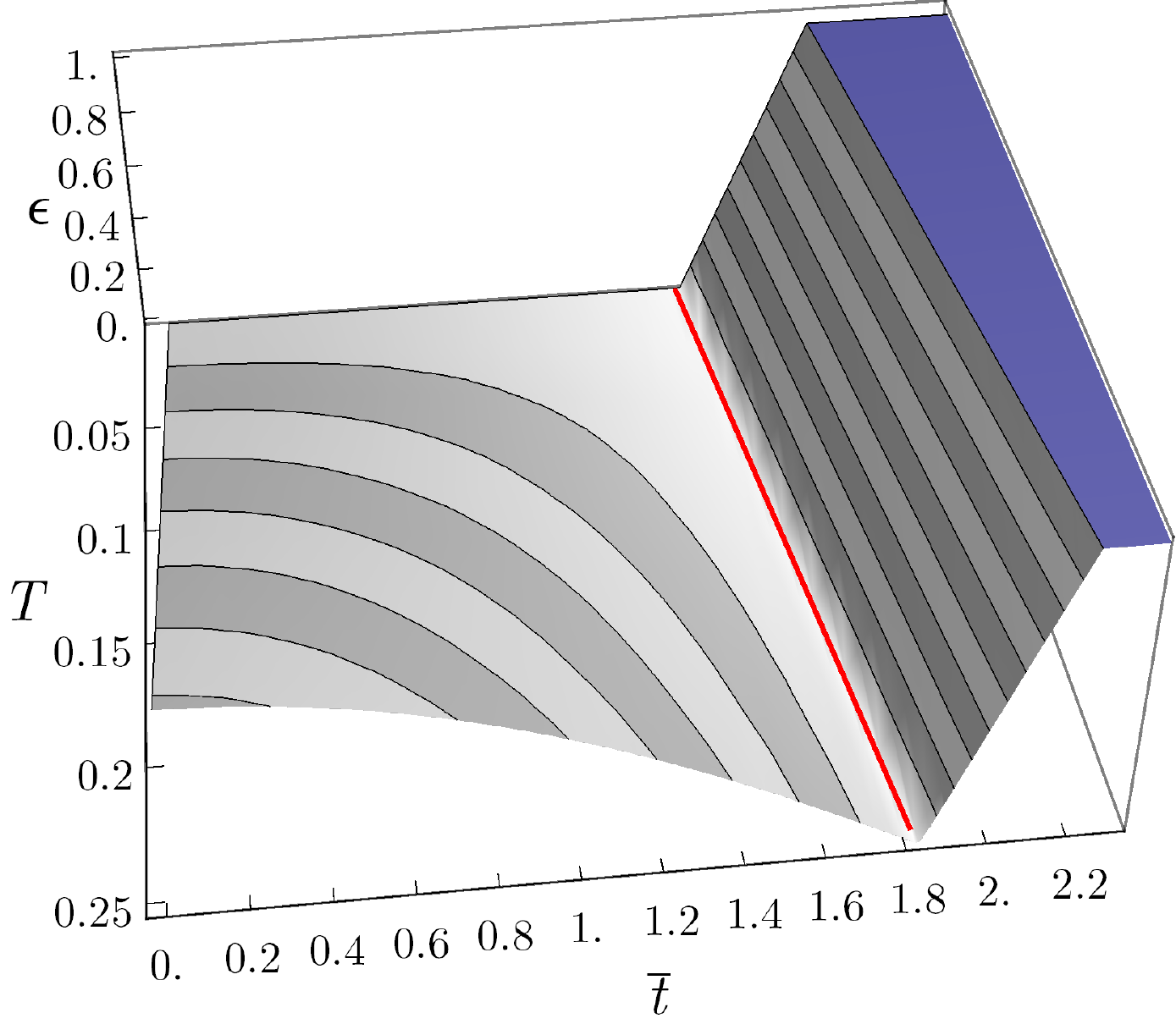}
    \hfill
    \includegraphics[width=.48\textwidth]{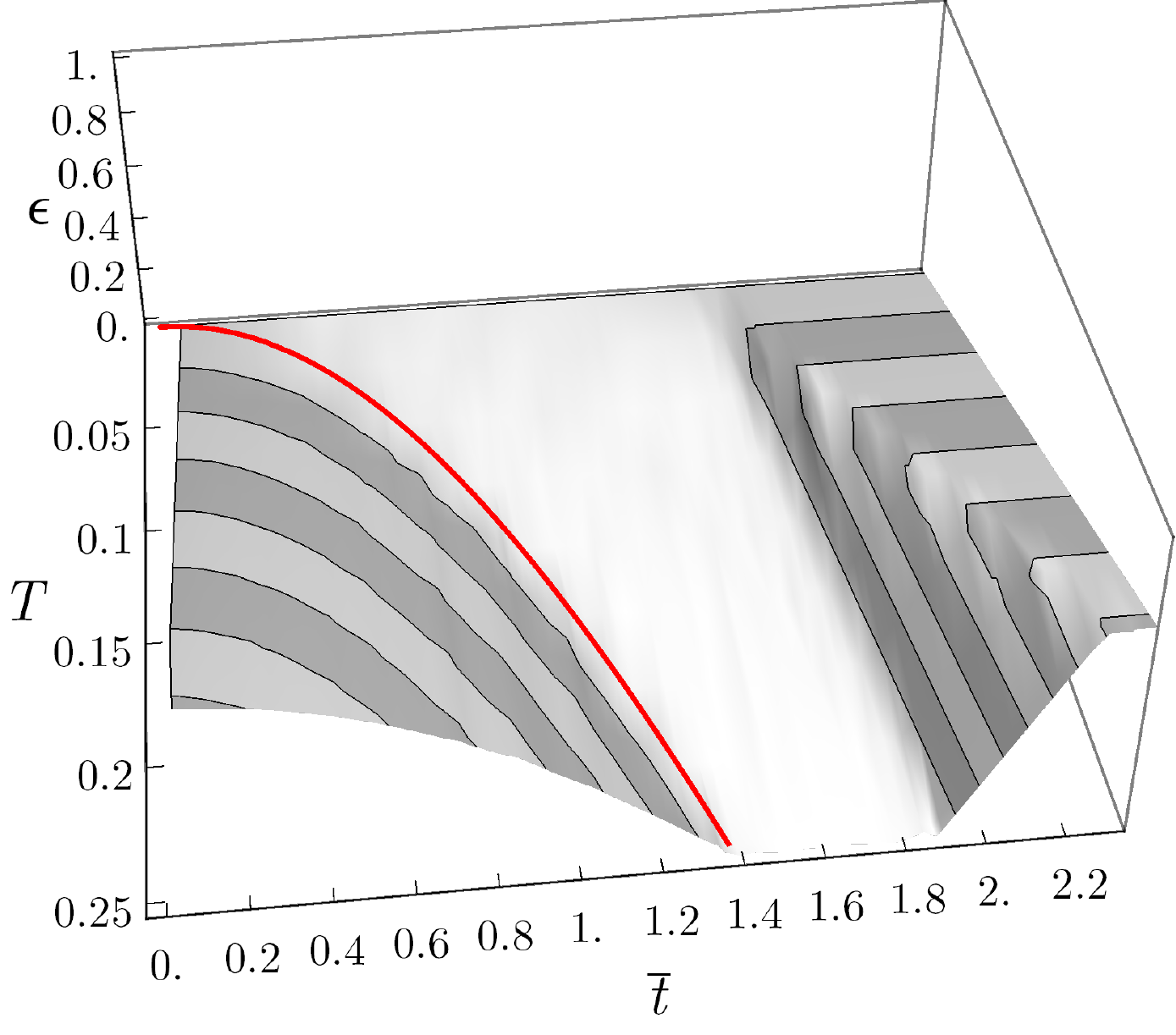}
    \caption{Numerical calculation of gate decomposition errors of the $\op U_\mathrm{target}=\exp(\ii T Z^{\otimes 3})$ gate, with a pulse sequence of depth 3 (left) and depth 4 (right).
    Plotted in red are the optimal analytical decompositions given by CNOT conjugation and \cref{lem:Weight-Increase-Depth4-ap}, respectively.}
    \label{fig:depth-3-4-numerics}
\end{figure}
\begin{figure}[t]
    \centering
    \includegraphics[width=.48\textwidth]{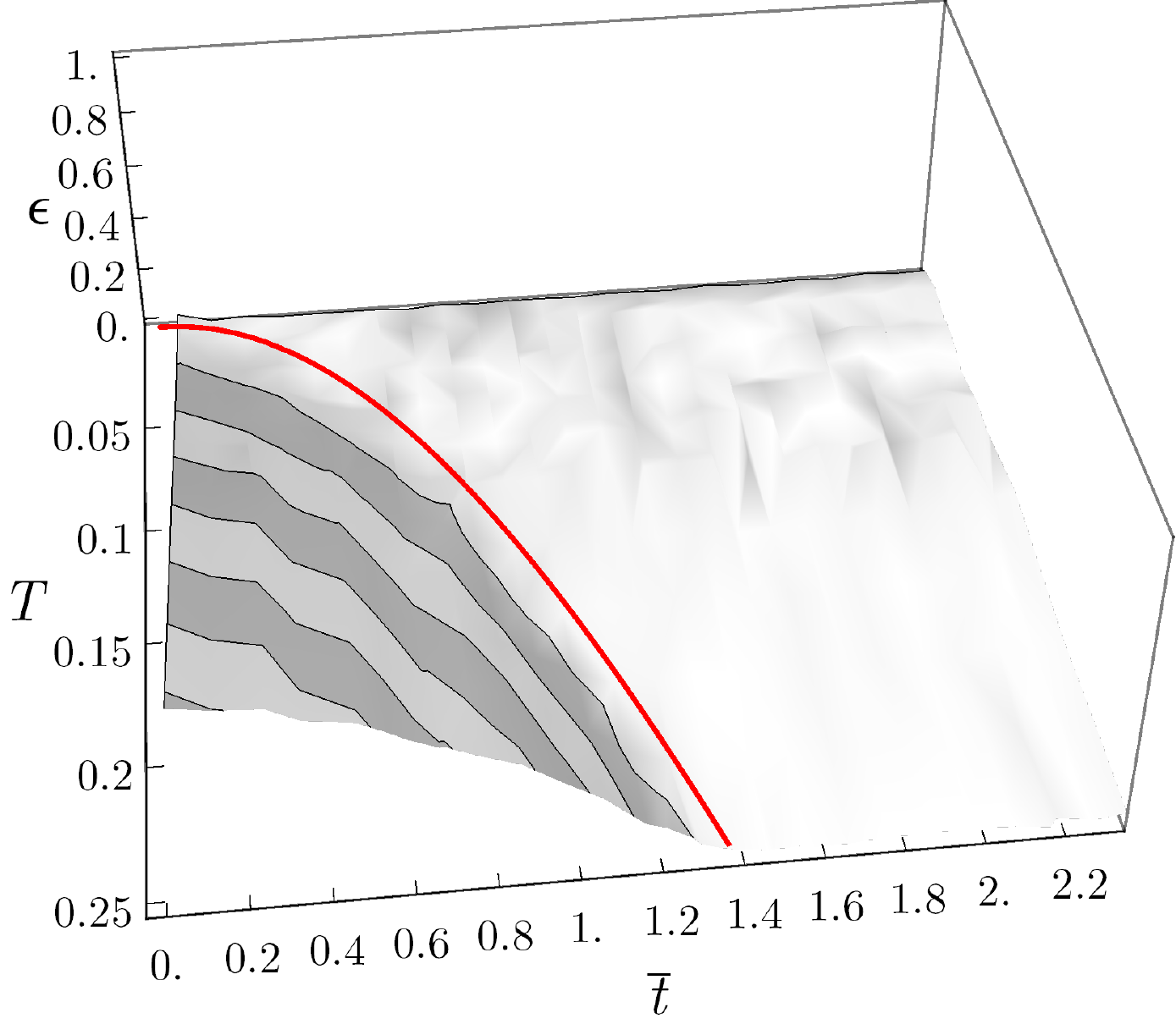}
    \hfill
    \includegraphics[width=.48\textwidth]{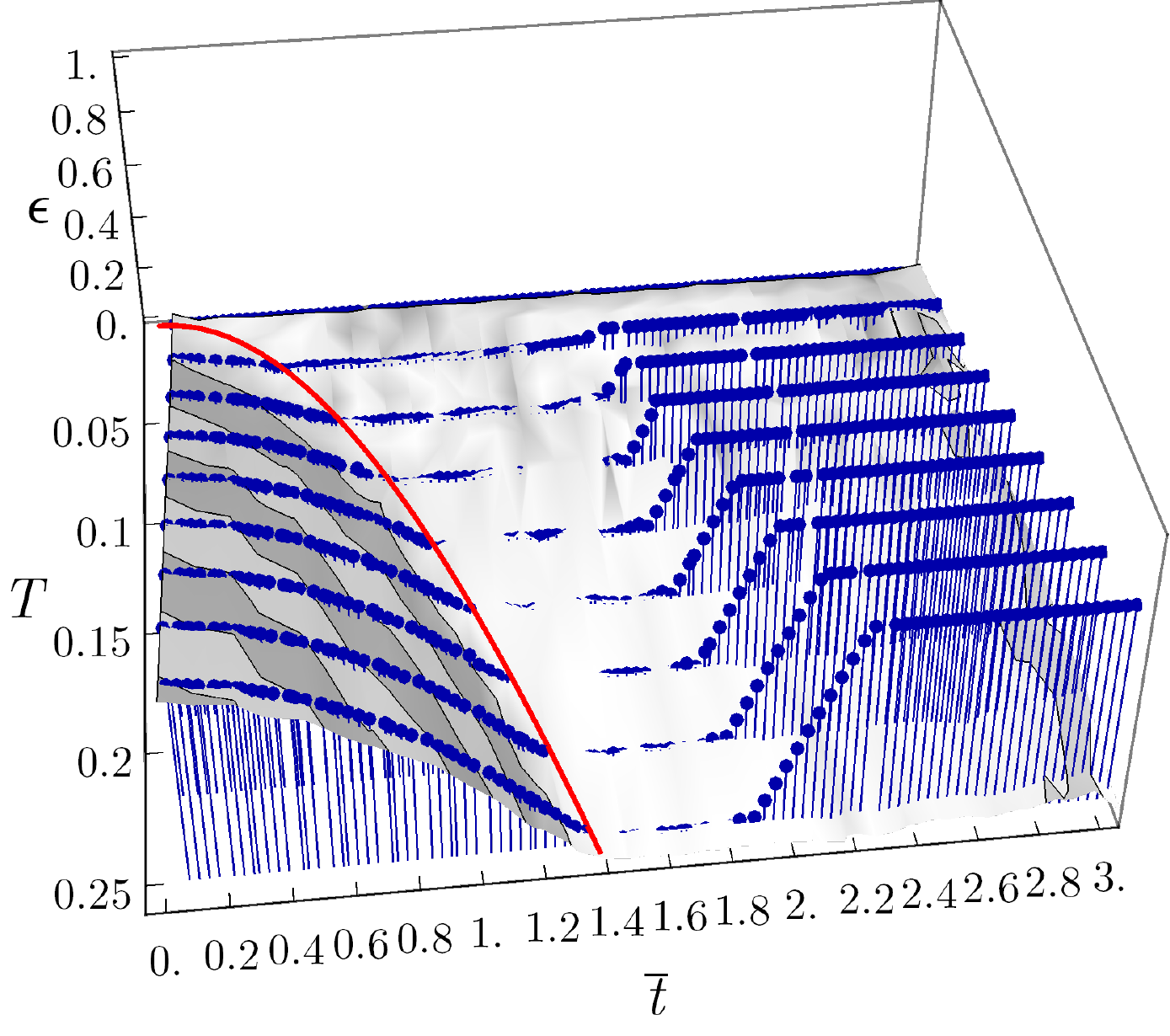}
    \caption{Numerical calculation of gate decomposition errors of the $\op U_\mathrm{target}=\exp(\ii T Z^{\otimes 3})$ gate, with a pulse sequence of depth 5.
    Plotted in red is the optimal analytical decompositions given for a depth 4 sequence in \cref{lem:Weight-Increase-Depth4-ap}; the blue lines are an overlay over the optimal depth 4 sequences from \cref{fig:depth-3-4-numerics}.}
    \label{fig:depth-5-numerics}
\end{figure}
As can be seen (plotted as red line), for $n=3$ the optimal zero-error decomposition has $\cost = \pi + \delta$ from CNOT conjugation.
For $n=4$, the optimal decomposition is given by the implicitly-defined solution in \cref{lem:Weight-Increase-Depth4-ap}, with a $\cost \propto \sqrt{\delta}$ dependence.
For the depth 5 sequences, it appears that the same optimality as for depth 4 holds.
In contrast to $n=3$ and $n=4$, there is now a zero error solution for all $\cost$ greater than the optimum threshold.



\section*{Suzuki-Trotter Formulae Error Bounds}\label{ap:trotter-deets-ap}
\subsection*{Existing Trotter Bounds}
Trotter error bounds have seen a spate of dramatic and very exciting improvements in the past few years \cite{Childs2017, Childs2019, Childsnew2019}. However, among these recent improvements we could not find a bound that was exactly suited to our purpose.

We wanted bounds which took into account the commutation relations between interactions in the Hamiltonian, as we know this leads to tighter error bounds \cite{Childs2019} \cite{Childs2017}. However, we needed exact constants in the error bound when applied to $2$D lattice Hamiltonians, such as the $2$D Fermi-Hubbard model. For this reason we could not directly apply the results of \cite{Childs2019} which only explicitly obtains constants for $1$D lattice Hamiltonians.

Additionally, we needed to be able to straightforwardly compute the bound for any higher order Trotter formula. This ruled out using the commutator bounds of \cite{Childs2017} as they become difficult to compute at higher orders. Furthermore these bounds require each Trotter layer to consist of a single interaction, meaning we wouldn't be able to exploit the result of \cref{thm:norm-bound}.

We followed the notation and adapted the methods of \cite{Childs2019} to derive bounds that meet the above criteria. Additionally we incorporate our own novel methods to tighten our bounds in \cref{cor:trotter-error-ap,cor:taylor-error-bound-ap} and \cref{thm:norm-bound}.

The authors of \cite{Childs2019} have recently extended their work further in \cite{Childsnew2019}. We have not yet seen whether they will further tighten our analysis, though we are keen to do this in future work.

\subsection*{Hamiltonian Simulation by Trotterisation}
In this section we derive our bounds for Trotter error. The standard approach to implementing time-evolution under a local Hamiltonian $\op H = \sum_i \op h_i$ on a quantum computer is to ``Trotterise'' the time evolution operator $\op U(T)=e^{-\ii \op H T}$.
Assuming that the Hamiltonian breaks up into $M$ mutually non-commuting layers $\op H = \sum_{i=1}^M \op H_i$ -- i.e.\ such that $\forall i\neq j\, [ \op H_i, \op H_j ] \neq 0$ -- Trotterizing in its basic form means expanding
\begin{equation}\label{eq:trotter-basic-ap}
  U(T) := \ee^{-\ii \op H T}
       = \prod_{n=1}^{T/\delta} \prod_{i=1}^M \ee^{-\ii \op H_i \delta} + \R_1\left(T,\delta\right)
       = \P_1\left(\delta\right)^{T/\delta} + \R_1\left(T,\delta\right)
\end{equation}
and then implementing the approximation $\P_1\left(\delta\right)^{T/\delta}$ as a quantum circuit.
Here $\R_1\left(T,\delta\right)$ denotes the error term remaining from the approximate decomposition into a product of individual terms.
$R_1\left(T,\delta\right) := U(T) - P_1\left(\delta\right)^{T/\delta}$ is simply defined as the difference to the exact evolution $U(T)$.
For $M$ mutually non-commuting layers of interactions $\op H_i$, we must perform $M$ sequential layers per Trotter step.

\Cref{eq:trotter-basic-ap} is an example of a first-order product formula, and is derived from the Baker-Campbell-Hausdorff identity
\begin{align}
    \ee^{\op A + \op B} &= \ee^{\op A}\ee^{\op B}\ee^{[ \op A, \op B]/2}\cdots
    \quad\text{and}\quad
    \ee^{\op A + \op B} = \ee^{\left(\delta \op A + \delta\op B\right)/\delta}
    = \left[ \ee^{\delta \op A + \delta \op B} \right]^{1/\delta}.
\end{align}
Choosing $\delta$ small in \cref{eq:trotter-basic-ap} means that corrections for every factor in this formula come in at $\BigO \left(\delta^2 \right)$ i.e.\ in the form of a commutator, and since we have to perform $1/\delta$ many rounds of the sequence $\ee^{\delta\op A}\ee^{\delta\op B}$ the overall error scales roughly as $\BigO \left(\delta \right)$.

Since its introduction in \cite{Lloyd1996a}, there have been a series of improvements, yielding higher-order expansions with more favourable error scaling.
For a historical overview of the use of Suzuki-Trotter formulas in the context of Hamiltonian simulation, we direct the reader to the extensive overview given in \cite[sec.~2.2.1]{Yung2014}.
In the following, we discuss the most recent developments for higher order product formulas, and analyse whether they yield an improved overall time and error scaling with respect to our introduced cost model.

To obtain higher-order expansions, Suzuki et.\ al.\ derived an iterative expression for product formulas in \cite{Suzuki1992,Suzuki1991}.
For the $\left(2k\right)$\textsuperscript{th} order, it reads \cite{Childs2019}
\begin{align}
    \P_2\left(\delta\right) &:= \prod_{j=1}^M \ee^{-\ii \op H_j \delta/2} \prod_{j=M}^1 \ee^{-\ii \op H_j \delta/2}, \label{eq:P-2-ap} \\
    \P_{2k}\left(\delta\right) &:= \P_{2k-2}\left(a_k \delta\right)^2\P_{2k-2}\left((1-4a_k) \delta\right) \P_{2k-2}\left(a_k \delta\right)^2, \label{eq:P-2k-ap}
\end{align}
where the coefficients are given by $a_k :=1/\left(4-4^{1/\left(2k-1\right)}\right)$.
The product limits indicate in which order the product is to be taken.
The terms in the product run from right to left, as gates in a circuit would be applied, so that $\prod_{j=1}^L \op A_j = \op A_L\cdots\op A_1$.

\subsection*{Error Analysis of Higher-Order Formulae}\label{subsec:trotter-error-ap}
We need an expression for the error $\R_p\left(T,\delta\right)$ arising from approximating the exact evolution $U(T)$ by a $p$\textsuperscript{th} order product formula $\P_p\left(\delta\right)$ repeated $T/\delta$ times.
As a first step, we bring the latter into the form:
\begin{align}
    \P_p\left(\delta\right) &:= \prod_{j=1}^{S} \P_{p,j} \left(\delta\right)=\P_{p,S}\left(\delta\right) \ldots \P_{p,2}\left(\delta\right) \P_{p,1}\left(\delta\right), \label{eq:higher-trotter-ap}\\
    \P_{p,j}\left(\delta\right) &:=\prod_{i=1}^{M} \op U_{ij}\left(\delta\right)
    \quad\text{where}\quad
    \op U_{ij}\left(\delta\right) :=\ee^{-\ii \delta  \tcoeff_{ji} \op H_i}.
    \label{eq:higher-trotter-2-ap}
\end{align}
As before, $M$ denotes the number of non-commuting \emph{layers} of interactions in the local Hamiltonian.
$S=S_p$ is the number of  \emph{stages}; the number of $\P_{p,j}(\delta)$ in a $p$\textsuperscript{th} order decomposition from \cref{eq:P-2-ap} or \cref{eq:P-2k-ap}. Here we note that we count a single stage as either $\prod_{i=1}^{M} \op U_{ij}\left(\delta\right)$ or $\prod_{i=M}^{1} \op U_{ij}\left(\delta\right)$, so that a second order formula is composed of $2$ stages.

\begin{lemma}\label{rem:bounds-on-trotter-coefficients-ap}
For a $p$\textsuperscript{th}-order decomposition with $p=1$ or $p=2k$, $k\ge 1$, we have $\sum_{j=1}^S \tcoeff_{ji}\left(p\right) = 1$ for all $i=1,\ldots,M$.
Furthermore, the Trotter coefficients $\tcoeff_{ji}$ satisfy
\begin{align}
    \max_{ij} \{ |\tcoeff_{ji}| \} \le B_p \le \begin{cases}
    1 & p=1 \\
    \displaystyle \frac12 \left( \frac23 \right)^{k-1} & \text{$p=2k$, $k\ge1$}
    \end{cases}
\end{align}
where
\begin{align}
    B_p := \begin{cases}
    1 & p=1 \\
    \frac12 & p=2 \\
    \frac12 \prod_{i=2}^k (1-4a_i) & \text{$p=2k$, $k\ge2$.}
    \end{cases}
\end{align}
\end{lemma}

\begin{proof}
The first claim is obviously true for the first order formula in \cref{eq:trotter-basic-ap}.
For higher orders, by \cite[Th.~3]{Childs2019} and \cref{eq:trotter-basic-ap}, we have that the first derivative
\begin{align}
\frac\partial{\partial x} \P_p\left(x\right)\bigg|_{x=0} = -\ii\sum_{i=1}^M \op H_i.
\end{align}
Similarly, from \cref{eq:P-2-ap,eq:P-2k-ap}, we have that
\begin{align}
\frac\partial{\partial x}\P_p\left(x\right)\bigg|_{x=0} =
\frac\partial{\partial x}\prod_{j=1}^S \prod_{i=1}^M \op U_{ij}\left(x\right) \bigg|_{x=0}  =-\ii \sum_{j=1}^S\sum_{i=1}^M \tcoeff_{ji}\op H_i.
\end{align}
Equating both expressions for the first derivative of $\P_p\left(x\right)$ at $x=0$ and realising that they have to hold for any $\op H_i$ yields the claim.

The second claim is again obviously true for a first order expansion, and follows immediately from \cref{eq:P-2-ap} for $p=2$.
Expanding \cref{eq:P-2k-ap} for $\P_{2k}\left(\delta\right)$ all the way down to a product of $\P_2$ terms, the argument of each of the resulting factors will be a product of $k-1$ terms of $a_{k'}$ or $1-4a_{k'}$ for $k'\le k$.
We further note that for $k\ge2$, $|a_k|\le|1-4a_{k}|$, as well as $|a_k|\le 1/2$ and $|1-4a_{k}|\le 2/3$, which can be shown easily.
The $\tcoeff_{ji}$ can thus be upper-bounded by $B_p$, which in turn is upper-bounded by $\left(1/2\right)\left(2/3\right)^{k-1}$ -- where the final factor of $\left(1/2\right)$ is obtained from the definition of $\P_2$.
\end{proof}

Since we are working with a fixed product formula order $p$ for the remainder of this section, we will drop the order subscript in the following and write $\P_p=\P$, $\R_p=\R$ for simplicity.
Assuming $\| \op H_i \| \le \Lambda$ for all $i=1,\ldots,M$, and setting the error
\begin{align}\label{eq:trotter-epsilon-ap}
    \epsilon_p\left(T,\delta\right) := \| \R\left(T,\delta\right) \| = \| U(T) - \P\left(\delta\right)^{T/\delta} \|,
\end{align}
we can derive an expression for the $p$\textsuperscript{th} order error term.
First, note that approximation errors in circuits accumulate at most linearly in \cref{eq:trotter-epsilon-ap}.
Thus it suffices to analyse a single $\delta$ step of the approximation, i.e.\ $U\left(\delta\right) = \P\left(\delta\right) + \R\left(\delta,\delta\right)$.
Then
\begin{equation}\label{eq:epsilon-single-step-ap}
  \epsilon_p\left(\delta\right) := \epsilon_p\left(\delta,\delta\right) = \| U\left(\delta\right) - \P\left(\delta\right) \|
\end{equation}
so that
\begin{equation}
  \epsilon_p\left(T,\delta\right) \le \frac{T}{\delta} \epsilon_p\left(\delta\right).
\end{equation}
We will denote $\epsilon_p\left(\delta\right)$ simply by $\epsilon$ in the following.

To obtain a bound on $\P\left(\delta\right)$, we apply the variation of constants formula with the condition that $\P\left(0\right)=I$, which always holds.
As in \cite[sec.~3.2]{Childs2019}, for $\delta\ge0$, we obtain
\begin{equation}\label{eq:integral-representation-ap}
    \P\left(\delta\right) = U\left(\delta\right) + \R\left(\delta\right) = \ee^{- \ii \delta \op H} + \int_{0}^{\delta} \ee^{- \ii \left(\delta-\tau\right)\op H} \op R\left(\tau\right) d\tau
\end{equation}
where the integrand $\op R\left(\tau\right)$ is defined as
\begin{align}
    \op R\left(\tau\right):=\frac{d}{d\tau} \P\left(\tau\right) -\left(- \ii \op H\right) \P\left(\tau\right).
\end{align}

Now, if $\P\left(\delta\right)$ is accurate up to $p$\textsuperscript{th} order -- meaning that $\R\left(\delta\right) = \BigO\left(\delta^{p+1}\right)$ -- it holds that
the integrand $\op R\left(\delta\right) = \BigO\left(\delta^p\right)$.
This allows us to restrict its partial derivatives, as the following shows.

\begin{lemma}\label{rem:order-conditions-ap}
  For a product formula accurate up to $p$\textsuperscript{th} order -- i.e.\ for which $\op R\left(\delta\right)=\BigO\left(\delta^p\right)$ -- the partial derivatives $\partial_\tau^j \op R\left(0\right) = 0$ for all $0\le j\le p-1$.
\end{lemma}

\begin{proof}
  We note that $\op R\left(\delta\right)$ is analytic, which means that we can expand it as a Taylor series $\op R\left(\delta\right)=\sum_{j=0}^\infty \op a_j \delta^i$.
  We proceed by induction.
  If $\op a_0 \neq 0$, then clearly $\op R\left(0\right)\neq 0$, which contradicts the assumption that $\op R\left(\delta\right)$ is accurate up to $p$\textsuperscript{th} order.
  Now assume for induction that $\forall j<j'<p-1: a_j = 0$ and $\op a_{j'}\neq 0$.
  Then
  \begin{align}
    \frac{\op R\left(\delta\right)}{T^{j'}} = \op a_{j'} + \sum_{i=1}^\infty \op a_{i+j'} T^i
    \xrightarrow{\delta\rightarrow 0} \op a_{j'} \neq 0,
  \end{align}
  which again contradicts that $\op R\left(0\right) = \BigO\left(\delta^p\right)$.
  The claim follows.
\end{proof}

Performing a Taylor expansion of $\op R\left(\tau\right)$ around $\tau=0$, the error bound $\epsilon$ given in \cref{eq:epsilon-single-step-ap} simplifies to
\begin{align}
    \epsilon &= \left\| \int_{0}^{\delta} \ee^{- \ii \left(\delta-\tau\right)\op H} \op R\left(\tau\right) \dd\tau \right\|
    \leq  \int_{0}^{\delta} \| \op R\left(\tau\right) \| \dd\tau\\
    &=\int_0^\delta \left( \| \op R\left(0\right) \| + \| \op R'\left(0\right) \| \tau + \ldots + \| \op R^{\left(p-1\right)}\left(0\right)\| \frac{\tau^{p-1}}{\left(p-1\right)!}
    + \| \op S\left(\tau, 0\right) \|
      \right) \dd\tau,
\end{align}
Further by \cref{rem:order-conditions-ap} all but the $p$\textsuperscript{th} or higher remainder terms $\op S\left(\tau, 0\right)$ equal zero, so
\begin{align}
    \epsilon &\le \int_0^\delta \| \op S\left(\tau, 0\right) \| \dd \tau = p \int_0^\delta \int_0^1 \left(1-x\right)^{p-1} \| \op R^{\left(p\right)}\left(x\tau\right)\| \frac{\tau^{p}}{p!} \dd x \dd\tau,
    \label{eq:higher-trotter-error-1-ap}
\end{align}
where we used the integral representation for the Taylor remainder $\op S\left(\tau,0\right)$.

Motivated by this, we look for a simple expression for the $p$\textsuperscript{th} derivative of the integrand $\op R\left(\tau\right)$, which capture this in the following technical lemma.
\begin{lemma}\label{lem:trotter-tech1-ap}
  For a product formula accurate to $p$\textsuperscript{th} order, having $S=S_p$ stages for $M$ non-commuting Hamiltonian layers with the upper-bound $\|\op H_i\|\le\Lambda$, the error term $\op R\left(\tau\right)$ satisfies
  \begin{align}
    \left\| \frac{\partial^p}{\partial\tau^p}\op R\left(\tau\right)\right\| \le \left(S M\right)^{p+1} \Lambda^{p+1} \begin{cases}
        2 & p=1 \\
        \displaystyle \frac{1}{2^p} \left(\frac23\right)^{\left(p+1\right)\left(p/2-1\right)} & \text{$p=2k$ for $k\ge1$}.
    \end{cases}
  \end{align}
\end{lemma}

\begin{proof}
  We first express $\P\left(\tau\right)$ from \cref{eq:higher-trotter-ap,eq:higher-trotter-2-ap} with a joint index set $\Sigma=[S]\times[M]$ as
  \begin{align}
    \P\left( \tau\right) = \prod_{j=1}^{S}\prod_{i=1}^{M}\op U_{ij}\left(\tau\right)
    =\prod_{I \in \Sigma}\op U_{I}\left(\tau\right).
  \end{align}
  Then the $\left(p+1\right)$\textsuperscript{th} derivative of this with respect to $\tau$ is
  \begin{align}
    \P^{\left(p+1\right)}\left(\tau\right) = \sum_{\alpha:\,|\alpha| = p+1} \binom{p+1}{\alpha} \prod_{I}\op U_{I}^{\left(\alpha_I\right)}\left( \tau\right)
    \label{eq:P-multiindex-ap}
  \end{align}
  where $\alpha$ is a multiindex on $\Sigma$, and $|\alpha|=\sum_{I\in\Sigma} \alpha_{I}$.
  Following standard convention, the multinomial coefficient for a multiindex is defined as
  \begin{align}
    \binom{p+1}{\alpha} = \frac{\left(p+1\right)!}{\alpha!} =\frac{\left(p+1\right)!}{\prod_{I\in\Sigma}\alpha_I!}.
  \end{align}
  We can similarly express $\op H$ with the same index set $\sigma$, and as a derivative of $\op U$ via
  \begin{equation}
    \op H = \sum_{i=1}^{S}\op H_i
    = \sum_{j=1}^{S} \sum_{i=1}^{M} \tcoeff_{ji}\op  H_i
    = \ii \sum_{j=1}^{S} \sum_{i=1}^{M} \op U_{ij}^{\left(1\right)}\left(0\right)
    = \ii \sum_{I\in\Sigma} \op U_{I}^{\left(1\right)}\left(0\right)
    \label{eq:H-multiindex-ap}
  \end{equation}
  where we used the fact that $\sum_{j=1}^S \tcoeff_{ji} = 1$ by \cref{rem:bounds-on-trotter-coefficients-ap}, and the exponential expression of $\op U_I$ from \cref{eq:higher-trotter-2-ap}.

  Now we can combine \cref{eq:P-multiindex-ap,eq:H-multiindex-ap} as in \cref{eq:higher-trotter-error-1-ap} to obtain the $p$\textsuperscript{th} derivative of the integrand $\op R\left(\tau\right)$:
  \begin{equation}\label{eq:R-bound-ap}
    \op R^{\left(p\right)}\left(\tau\right) = \sum_{\alpha:\,|\alpha| = p+1} \binom{p+1}{\alpha} \prod_{I} \op U_{I}^{\left(\alpha_I\right)}\left( \tau\right) - \sum_{I} \op U_{I}^{\left(1\right)}\left(0\right)\sum_{\beta:\,|\beta| = p} \binom{p}{\beta} \prod_{I} \op U_{I}^{\left(\beta_I\right)}\left( \tau\right).
  \end{equation}
  Noting that $\|\op U_{I}^{\left(\beta_I\right)}\left(\tau\right)\|=\| \op U_{I}^{\left(\beta_I\right)}\left(0\right)\|$, and further $\op U_I^{\left(x\right)}\left(0\right)\op U_I^{\left(y\right)}\left(0\right) = \op U_I^{\left(x+y\right)}\left(0\right)$, we have
  \begin{align}
    \sum_{J} \left\| \op U_{J}^{\left(1\right)}\left(0\right) \right\| \sum_{\beta:\,|\beta| = p} \binom{p}{\beta} \prod_{I} \left\| \op U_{I}^{\left(\beta_I\right)}\left(0\right) \right\|
    &= \sum_{\beta:\,|\beta|=p+1}\binom{p+1}{\beta} \prod_I \left\| \op U_I^{\left(\beta_I\right)}\left(0\right) \right\|.
  \end{align}
  We can therefore bound the norm of $\op R^{\left(p\right)}$ as follows:
  \begin{align}
    \left\|\op R^{\left(p\right)}\left(\tau\right)\right\| &\leq \sum_{\alpha:\,|\alpha| = p+1}\binom{p+1}{\alpha} \prod_{I} \left\|\op U_{I}^{\left(\alpha_I\right)}\left(0\right)\right\| \\
    &\phantom{=}\hspace{.8cm} + \sum_{I} \left\|\op U_{I}^{\left(1\right)}\left(0\right)\right\|\sum_{\beta:\,|\beta| = p} \binom{p}{\beta}\prod_{I} \left\|\op U_{I}^{\left(\beta_I\right)}\left(0\right)\right\| \\
    & =  2 \sum_{\alpha:\,|\alpha| = p+1}\binom{p+1}{\alpha} \prod_{I} \left\|\op U_{I}^{\left(\alpha_I\right)}\left(0\right)\right\| \\
    &= 2 \sum_{\alpha:\,|\alpha| = p+1}\binom{p+1}{\alpha} \prod_{j=1}^S\prod_{i=1}^M \left|\tcoeff_{ji}\right|^{\alpha_{ij}} \left\| \op H_{i}\right\|^{\alpha_{ij}}.
  \end{align}

  By \cref{rem:bounds-on-trotter-coefficients-ap}, we know that $|\tcoeff_{ji}|=1$ when $p=1$ and $|\tcoeff_{ji}|\le \left(2/3\right)^{p/2-1}/2$ for all $j,i$ when $p=2k$ for $k\ge1$.
  Hence for $p=1$
  \begin{align}
    \left\| \op R^{\left(1\right)}\left(\tau\right)\right\| \le 2 \left(S M\right)^{2} \Lambda^{2},
  \end{align}
  and for $p=2k$ for $k\ge1$
  \begin{align}
    \left\| \op R^{\left(p\right)}\left(\tau\right)\right\| \le 2 \sum_{\alpha:\,|\alpha|=p+1} \binom{p+1}{\alpha} \left[ \left(\frac23 \right)^{p/2-1} \frac{\Lambda}{2} \right]^{|\alpha|}
    =: C_p\left(S,M\right) \left( \frac23 \right)^{\left(p+1\right)\left(p/2-1\right)} \frac{\Lambda^{p+1}}{2^p},
  \end{align}
  where $C_p\left(S,M\right)$ is the sum of the multinomial coefficients of length $p\in\field N$; a simple expression can be obtained by reversing the multinomial theorem, since
  \begin{align}
    \sum_{\alpha:\,|\alpha|+p+1} \binom{p+1}{\alpha} = \left(\underbrace{1+1+\ldots+1}_{|\Sigma|\ \text{terms}}\right)^{p+1} = |\Sigma|^{p+1} = \left(SM\right)^{p+1}.
  \end{align}
  \qedhere
\end{proof}

To obtain the final error bounds, we combine \cref{lem:trotter-tech1-ap} with the integral representation in \cref{eq:higher-trotter-error-1-ap}.

\begin{theorem}[Trotter Error]\label{th:trotter-error-ap}
  For a $p$\textsuperscript{th} order product formula $\P_p$ for $p=1$ or $p=2k$, $k\ge1$, with the same setup as in \cref{lem:trotter-tech1-ap}, a bound on the approximation error for the exact evolution $U(T)$ with $T/\delta$ rounds of the product formula $\P_p\left(\delta\right)$ is given by
  \begin{align}
    \epsilon_p\left(T,\delta\right) \le \frac {T}{\delta}\, \delta^{p+1} M^{p+1} \Lambda^{p+1} \times \begin{cases}
        1 & p=1 \\
        \displaystyle \frac{2}{\left(p+1\right)!} \left( \frac{10}3 \right)^{\left(p+1\right)\left(p/2-1\right)} & \text{$p=2k$, $k\ge 1$}.
    \end{cases}
  \end{align}
\end{theorem}

\begin{proof}
  We can use the bound on $\op R^{\left(p\right)}$ derived in \cref{lem:trotter-tech1-ap} and perform the integration over $\tau$ and $x$ in \cref{eq:higher-trotter-error-1-ap}, to obtain
  \begin{align}
    \epsilon \le \| \op R^{\left(p\right)} \| \int_0^\delta p \int_0^1  \left(1-x\right)^{p-1} \frac{\tau^p}{p!} \dd x \dd \tau
    = \frac{\delta^{p+1}}{\left(p+1\right)!}\| \op R^{\left(p\right)} \|.
  \end{align}
  By \cref{rem:order-conditions-ap}, for Trotter formulae of order $p=1$ we have precisely one stage, i.e.\ $S=1$, and $\tcoeff_{ji}=1$ for all $i,j$.
  This, together with \cref{lem:trotter-tech1-ap,eq:epsilon-single-step-ap}, yields the first bound.

  The number of stages in higher order formulae can be upper-bounded by \cref{eq:P-2-ap,eq:P-2k-ap}, giving $S_p \le 2\times 5^{p/2-1}$.
  Together with \cref{lem:trotter-tech1-ap,eq:epsilon-single-step-ap}, this yields the second bound.
\end{proof}

We remark that tighter bounds than the ones in \cref{th:trotter-error-ap} are achievable for any given product formula, where the form of its coefficients $\tcoeff_{ji}$ are explicitly available and not merely bounded as in \cref{rem:bounds-on-trotter-coefficients-ap}.
Summing up these stage times exactly is therefore an immediate way to obtain an improved error bound.
Furthermore, the triangle inequality on $\| \op R^{\left(p\right)}\left(\tau\right) \|$ in the proof of \cref{lem:trotter-tech1-ap} is a crude overestimate: it looses information about (i).~terms that could cancel between the two multi-index sums, and (ii.)~any commutation relations between the individual Trotter stages.

In the following subsection, we will provide a tighter error analysis, featuring more optimal but less clean analytical expressions which we can nonetheless evaluate efficiently numerically.

\subsection*{Explicit Summation of Trotter Stage Coefficients}
For the recursive Suzuki-Trotter formula in \cref{eq:P-2k-ap} we can immediately improve the error bound by summing the stage coefficients $\tcoeff_{ij}$ up exactly, instead of bounding them as in \cref{rem:bounds-on-trotter-coefficients-ap}.

\begin{corollary}[Trotter Error]\label{cor:trotter-error-ap}
For the recursive product formula in \cref{eq:P-2k-ap} and $p=2k$ for $k\ge1$,
\begin{align}
    \epsilon_p\left(T,\delta\right) \le \frac{2 T \delta^p M^{p+1} \Lambda^{p+1} }{\left(p+1\right)!} H_p^{p+1}
    \quad\text{where}\quad
    H_p := \prod_{i=1}^{p/2-1} \frac{ 4+4^{1/\left(2i+1\right)}}{\left| 4 - 4^{1/\left(2i+1\right)} \right| }.
\end{align}
\end{corollary}
\begin{proof}
This follows from explicitly summing up the magnitudes of all the $\tcoeff_{ji}$'s obtained by solving the recursive definition of the product formula, which can easily be verified to satisfy $\sum_{ij} |\tcoeff_{ij}\left(p\right)| = M H_p$.
Then from \cref{lem:trotter-tech1-ap},
\begin{align}
    \left\| \op R^{\left(p\right)}\left(\tau\right) \right\| \le 2 \Lambda^{p+1}\!\!\!\! \sum_{\alpha:\,|\alpha|=p+1} \binom{p+1}{\alpha} \prod_{j=1}^S \prod_{i=1}^M \left| \tcoeff_{ji}^{\alpha_{ij}} \right|
    = 2\Lambda^{p+1} \left( \sum_{j=1}^S \sum_{i=1}^M |\tcoeff_{ji}| \right)^{p+1},
\end{align}
and the claim follows as before.
\end{proof}

For later reference, we note that it is straightforward to generalise the error bound in \cref{cor:trotter-error-ap} for the case of a \emph{higher} derivative $\op R^{(q)}$, $q\ge p$, but still for a $p$\textsuperscript{th} order formula: the bound simply reads
\begin{equation}\label{eq:p-q-error-1-ap}
\epsilon_{p,q}(T,\delta) \le \frac{2 T \delta^q M^{q+1} \Lambda^{q+1} }{\left(q+1\right)!} H_p^{q+1}.
\end{equation}

\subsection*{Commutator Bounds}\label{ap:commutator_bounds}
Our analysis thus far has completely neglected the underlying structure of the Hamiltonian. In this subsection we establish commutator bounds which are easily applicable to $D$-dimensional lattice Hamiltonians.

We begin with the following technical lemmas.
\begin{lemma}\label{lem:trotter-tech-comm-ap}
For a product formula accurate to $p$\textsuperscript{th} order, having $S=S_p$ stages for $M$ non-commuting Hamiltonian layers with the upper-bound $\|\op H_i\|=\Lambda_I$, the error term $\op R\left(\tau\right)$ satisfies
\begin{align}
    \left\| \frac{\partial^p}{\partial\tau^p}\op R\left(\tau\right)\right\| &\leq \sum_{J} \sum_{\beta:\,|\beta| = p} \binom{p}{\beta} \sum_{I=J+1}^{SM}  \left(B_p \Lambda\right)^{p-\beta_{I}} \left\|\left[\op U_{J}^{\left(1\right)}\left(0\right) \ , \op U_I^{\left(\beta_I\right)}\left( \tau\right)\right]\right\|.
\end{align}
\end{lemma}
\begin{proof}
As shown in \cref{lem:trotter-tech1-ap},
\begin{align}
    \op R^{\left(p\right)}\left(\tau\right) = \sum_{\alpha:\,|\alpha| = p+1} \binom{p+1}{\alpha} \prod_{I} \op U_{I}^{\left(\alpha_I\right)}\left( \tau\right) - \sum_{J} \op U_{J}^{\left(1\right)}\left(0\right) \sum_{\beta:\,|\beta| = p} \binom{p}{\beta} \prod_{I} \op U_{I}^{\left(\beta_I\right)}\left( \tau\right).
\end{align}
We begin by commuting every $\op U_{J}^{\left(1\right)}\left(0\right)$ past $\prod^{SM}_{I=J+1} \op U_I^{\left(\beta_I\right)}\left( \tau\right)$. Consider this for some fixed $J$ in the sum of over $J$. That is consider rewriting a particular summand from the second term above to obtain
\begin{align}
    \op U_{J}^{\left(1\right)}&\left(0\right) \sum_{\beta:\,|\beta| = p} \binom{p}{\beta} \prod_{I} \op U_{I}^{\left(\beta_I\right)}\left( \tau\right) \\
    &= \sum_{\beta:\,|\beta| = p} \binom{p}{\beta} \op U_{J}^{\left(1\right)}\left(0\right) \left(\op U_{sm}^{\left(\beta_{sm}\right)}\left( \tau\right) \ldots \op U_{J+1}^{\left(\beta_{J+1}\right)}\left( \tau\right) \right)\left(\op U_{J}^{\left(\beta_{J}\right)}\left( \tau\right) \ldots \op U_{1}^{\left(\beta_{1}\right)}\left( \tau\right) \right)\\
    &=\sum_{\beta:\,|\beta| = p} \binom{p}{\beta}  \left(\op U_{sm}^{\left(\beta_{sm}\right)}\left( \tau\right) \ldots \op U_{J+1}^{\left(\beta_{J+1}\right)}\left( \tau\right)\right) \left( \op U_{J}^{\left(\beta_{J}+1\right)}\left( \tau\right) \ldots \op U_{1}^{\left(\beta_{1}\right)}\left( \tau\right) \right) \\
    &\mspace{30mu} + \sum_{\beta:\,|\beta| = p} \binom{p}{\beta} \bigg[\op U_{J}^{\left(1\right)}\left(0\right)\ ,\ \prod^{SM}_{I=J+1} \op U_I^{\left(\beta_I\right)}\left( \tau\right)\bigg] \prod^{J}_{I=1} \op U_I^{\left(\beta_I\right)}\left( \tau\right).
\end{align}
Now, by inserting this into the full expression for $ \op R^{\left(p\right)}\left(\tau\right)$, we obtain
\begin{align}
    \op R^{\left(p\right)}\left(\tau\right) &= \sum_{\alpha:\,|\alpha| = p+1} \binom{p+1}{\alpha} \prod_{I} \op U_{I}^{\left(\alpha_I\right)}\left( \tau \right) - \sum_{\beta:\,|\beta| = p+1} \binom{p+1}{\beta} \prod_{I} \op U_{I}^{\left(\beta_I\right)}\left( \tau\right) \\
    &\mspace{30mu} -\sum_{J} \sum_{\beta:\,|\beta| = p} \binom{p}{\beta} \bigg[\op U_{J}^{\left(1\right)}\left(0\right)\ ,\ \prod^{SM}_{I=J+1} \op U_I^{\left(\beta_I\right)}\left( \tau\right)\bigg] \prod^{J}_{I=1} \op U_I^{\left(\beta_I\right)}\left( \tau\right) \\
    &= -\sum_{J} \sum_{\beta:\,|\beta| = p} \binom{p}{\beta} \bigg[\op U_{J}^{\left(1\right)}\left(0\right)\ ,\ \prod^{SM}_{I=J+1} \op U_I^{\left(\beta_I\right)}\left( \tau\right)\bigg] \prod^{J}_{I=1} \op U_I^{\left(\beta_I\right)}\left( \tau\right)\\
    &= -\sum_{J} \sum_{\beta:\,|\beta| = p} \binom{p}{\beta} \sum_{I=J+1}^{SM} \prod^{SM}_{K=I+1} \op U_K^{\left(\beta_K\right)} \left[\op U_{J}^{\left(1\right)}\left(0\right) \ , \op U_I^{\left(\beta_I\right)}\left( \tau\right)\right] \prod^{I-1}_{K=1} \op U_K^{\left(\beta_K\right)}.
\end{align}
Taking the norm of this expression gives
\begin{align}
    \left\| \frac{\partial^p}{\partial\tau^p}\op R\left(\tau\right)\right\| &\leq \sum_{J} \sum_{\beta:\,|\beta| = p} \binom{p}{\beta} \sum_{I=J+1}^{SM} \left( B_p  \Lambda\right)^{\sum_{K=I+1}^{SM} \beta_K} \left\|\left[\op U_{J}^{\left(1\right)}\left(0\right) \ , \op U_I^{\left(\beta_I\right)}\left( \tau\right)\right]\right\|  \left(B_p \Lambda\right)^{\sum_{K=1}^{I-1}\beta_{K}} \\
    & =\sum_{J} \sum_{\beta:\,|\beta| = p} \binom{p}{\beta} \sum_{I=J+1}^{SM}  \left(B_p \Lambda\right)^{p-\beta_{I}} \left\|\left[\op U_{J}^{\left(1\right)}\left(0\right) \ , \op U_I^{\left(\beta_I\right)}\left( \tau\right)\right]\right\|.
\end{align}
This completes the proof.
\end{proof}

\begin{lemma}\label{lem:zero-time-commutator-bounds-ap}
If every pair of Hamiltonians can be written as $\op H_I = \sum_{i=1}^{N} \op h^{I}_{i}$ and $\op H_J = \sum_{i=1}^{N} \op h^{J}_{i}$, where for any $i$ we have $\| \op h^{I}_{i} \|=\| \op h^{J}_{i} \|=1$ and for any fixed term $\op h^{J}$ there are at most $\n$ terms in $\op H_I$ which do not commute with that specific term, then
\begin{align}
    \left\|\left[\op U_{J}^{\left(1\right)}\left(0\right) \ , \op U_I^{\left(\beta_I\right)}\left( 0\right)\right]\right\| & \leq 2 \n \beta_I N^{\beta_I} B_{p}^{\beta_I + 1}.
\end{align}
\end{lemma}
\begin{proof}
First note that
\begin{align}
\op U_{J}^{\left(1\right)}\left(0\right) = -\ii \tcoeff_{J} \op H_J = -\ii \tcoeff_{J} \sum_{i=1}^{N}  \op h^{J}_{i}
\end{align}
and
\begin{align}
\op U_{I}^{\left(\beta_I\right)}\left(0\right) = \left(-\ii \tcoeff_{I} \op H_I\right)^{\beta_I} = \left(-\ii \tcoeff_{I}\right)^{\beta_I}\left( \sum_{i=1}^{N} \op h^{I}_{i}\right)^{\beta_I}.
\end{align}

Consider a fixed term in $\op U_{J}^{\left(1\right)}\left(0\right)$ such as $-\ii \tcoeff_J \op h^{J}$, where we have dropped the subscript $i$. As there are $N$ of these, we can bound the norm of the commutator as follows
\begin{align}
    \left\|\left[\op U_{J}^{\left(1\right)}\left(0\right) , \op U_I^{\left(\beta_I\right)}\left( 0\right)\right]\right\| & \leq N B_p \left\|\left[ \op h^{J} , \op U_I^{\left(\beta_I\right)}\left( 0\right)\right]\right\|.
\end{align}
Where we have also used the triangle inequality and the fact that $\tcoeff_J \leq B_p$.

Now consider fully expanding  the $\op U_I^{\left(\beta_I\right)}$ so that it is a sum of $N^{\beta_I}$ norm-$1$ Hamiltonians with coefficients upper-bounded by $\left(B_p\right)^{\beta_I}$. As only
$\n$ of the $N$ normalised Hamiltonians do not commute with $\op h^{J}$, the number of Hamiltonians in the expanded $\op U_{I}^{\left(\beta_I\right)}$ which do not commute with $\op h^{J}$ can be upper-bounded by $\n \beta_I N^{\beta_I -1}$. Here we have assumed that if any of the $\n$ non-commuting terms appear at any point in the expansion (the $\n$), then that term will not commute with $\op h^{J}$ regardless of whatever other terms appear (the $N^{\beta_I -1}$). We can over-count by repeating this for each term expanded (the $\beta_{I}$). This gives
\begin{align}
    \left\|\left[\op U_{J}^{\left(1\right)}\left(0\right), \op U_I^{\left(\beta_I\right)}\left( 0\right)\right]\right\| & \leq 2 \n \beta_I N^{\beta_I} B_{p}^{\beta_I + 1}.
\end{align}
The extra factor of $2$ comes from bounding the commutators of the norm~1 Hamiltonians via triangle inequality.
\end{proof}

\begin{lemma}\label{lem:R-bound-Comms-ap}
If every pair of Hamiltonians can be written as $\op H_I = \sum_{i=1}^{N} \op h^{I}_{i}$ and $\op H_J = \sum_{i=1}^{N} \op h^{J}_{i}$, where all $\| \op h^{I}_{i} \|=\| \op h^{J}_{i} \|=1$, and if additionally for any fixed term $\op h^{J}$ there are at most $\n$ terms $\op h^{I}$ which do not commute with $\op h^{J}$, then
\begin{align}
    \left\| \frac{\partial^p}{\partial\tau^p}\op R\left(\tau\right)\right\| &\leq  \n p B_p^{p+1} \Lambda^{p-1} N \left(\left(SM-1\right)+\frac{N}{\Lambda}\right)^{p-1} \left(\left(SM\right)^{2}-\left(SM\right)\right) \\
    &+ \n \tau B_p^{p+2} \Lambda^{p} N  \ee^{\tau N B_p} \left(\left(SM\right)^{p+2}-\left(SM\right)^{p+1}\right).
\end{align}
\end{lemma}
\begin{proof}
We must obtain a simplified form for the bounded commutator appearing in \cref{lem:trotter-tech-comm-ap}. We can sequentially expand this commutator and use the triangle inequality to write it as
\begin{align}
    \left\|\left[\op U_{J}^{\left(1\right)}\left(0\right) , \op U_I^{\left(\beta_I\right)}\left(\tau\right)\right]\right\| &= \left\|\left[\op U_{J}^{\left(1\right)}\left(0\right) , \op U_I^{\left(\beta_I\right)}\left(0\right) \ee^{\ii \tau \tcoeff_I \op H_{I}} \right]\right\|
    \\ &\leq   \left\|\left[\op U_{J}^{\left(1\right)}\left(0\right) , \op U_I^{\left(\beta_I\right)}\left(0\right)\right]\right\| + B_p^{\beta_{I}} \Lambda^{\beta_I} \left\|\left[\op U_{J}^{\left(1\right)}\left(0\right) , \ee^{\ii \tau \tcoeff_I \op H_{I}} \right]\right\|.
\end{align}
We can use \cref{lem:zero-time-commutator-bounds-ap} to bound the first term.
The commutator in the second term can be bounded as follows:
\begin{align}
    \left\|\left[\op U_{J}^{\left(1\right)}\left(0\right) , \ee^{\ii \tau \tcoeff_I \op H_{I}} \right]\right\| &= \left\|\left[\op U_{J}^{\left(1\right)}\left(0\right) , I+ \ii \tau \tcoeff_I \op H_{I}+\frac{1}{2!} \left(\ii \tau \tcoeff_I \op H_{I}\right)^2 + \ldots \right]\right\| \\
    &\leq \sum_{k=1}^{\infty} \frac{\tau^k}{k!}\left\| \left[\op U_{J}^{\left(1\right)}\left(0\right) , \op U_{I}^{\left(k\right)}\left(0\right) \right] \right\| \\
     &\leq \sum_{k=1}^{\infty}  \frac{2 \n \tau^k}{\left(k-1\right)!}  N^{k} B_{p}^{k + 1} \\
     &= 2 \n \tau N B^{2}_p \ee^{ N B_p \tau} .
\end{align}
Where we have used \cref{lem:zero-time-commutator-bounds-ap} to bound the norm of the commutator of $\op U_{J}^{\left(1\right)}\left(0\right) $ and $ \op U_{I}^{\left(k\right)}\left(0\right) $ by $2 \n k N^{k} B^{k+1}_p$ and simplified the resulting expression.
The first term can be bounded directly with \cref{lem:zero-time-commutator-bounds-ap}, so we obtain
\begin{align}
        \left\|\left[\op U_{J}^{\left(1\right)}\left(0\right) , \op U_I^{\left(\beta_I\right)}\left(\tau\right)\right]\right\| & \leq  2 \n \beta_I N^{\beta_I} B^{\beta_I + 1}_p + 2 \n \tau N \Lambda^{\beta_I} B^{2+\beta_I}_p \ee^{N B_p  \tau}.
\end{align}
Now by using this to bound the result of \cref{lem:trotter-tech-comm-ap} we obtain
\begin{align}
    \left\| \frac{\partial^p}{\partial\tau^p}\op R\left(\tau\right)\right\| &\leq \sum_{J} \sum_{\beta:\,|\beta| = p} \binom{p}{\beta} \sum_{I=J+1}^{SM}  \left(B_p \Lambda\right)^{p-\beta_{I}} \left\|\left[\op U_{J}^{\left(1\right)}\left(0\right) \ , \op U_I^{\left(\beta_I\right)}\left( \tau\right)\right]\right\| \\
    &\leq \sum_{J} \sum_{\beta:\,|\beta| = p} \binom{p}{\beta} \sum_{I=J+1}^{SM}  \left(B_p \Lambda\right)^{p-\beta_{I}} \left( 2 \n \beta_I N^{\beta_I} B^{\beta_I + 1}_p + 2 \n \tau N \Lambda^{\beta_I} B^{2+\beta_I}_p \ee^{N B_p  \tau} \right)\\
    &= \sum_{J} \sum_{\beta:\,|\beta| = p} \binom{p}{\beta} \sum_{I=J+1}^{SM}\left( 2 \n \beta_I \left(\frac{N}{\Lambda}\right)^{\beta_I} \Lambda^{p}B^{p+1}_p + 2 \n \tau N \Lambda^{p} B^{p+2}_p \ee^{N B_p  \tau} \right).
\end{align}
To simplify this expression, we must simplify an expression of the form
\begin{align}
    \sum_{\beta:\,|\beta| = p} \binom{p}{\beta} \beta_I x^{\beta_I}
\end{align}
where in our case $x= N/\Lambda$.
This can be done by rewriting this expression in terms of a derivative with respect to $x$ and reversing the multinomial theorem, which gives
\begin{align}
    \sum_{\beta:\,|\beta| = p} \binom{p}{\beta} \beta_I x^{\beta_I} &= x\frac{d}{dx}  \sum_{\beta:\,|\beta| = p}  \binom{p}{\beta} x^{\beta_I} \\
    &= x\frac{d}{dx} \left(\underbrace{1+\ldots+1+x}_{SM \ \text{terms}}\right)^{p}\\
    &=p x \left( SM - 1 + x\right)^{p-1}.
\end{align}
Using this and performing the summation over $J$ and $I$ simplifies the expression for $\left\| \op R^{\left(p\right)}\left(\tau \right) \right\|$ to
\begin{align}
    \left\| \frac{\partial^p}{\partial\tau^p}\op R\left(\tau\right)\right\| &\leq p \n  B_p^{p+1} \Lambda^{p-1} N \left(SM-1+\frac{N}{\Lambda}\right)^{p-1} \left(\left(SM\right)^{2}-\left(SM\right)\right) \\
    &+  \tau \n B_p^{p+2} \Lambda^{p} N   \left(\left(SM\right)^{p+2}-\left(SM\right)^{p+1}\right) \ee^{\tau N B_p}.
\end{align}
\end{proof}

Now we can use the preceding lemmas to establish a commutator bound for higher order Trotter formulae. Although it is cumbersome looking, it is easy to evaluate.
\begin{theorem}[Commutator Error Bound]\label{th:Trotter-Er-Commutator-ap}
  Let $\op H = \sum_{i=1}^{M} \op H_i$ with $\| \op H_i \| \leq \Lambda$ be a Hamiltonian with $M$ mutually commuting layers $\op H_I = \sum_{i=1}^{N} \op h^{I}_{i}$.
  Assume that for any $i$, $\| \op h^{I}_{i} \|=\| \op h^{J}_{i} \| \leq 1$.
  Additionally, assume that for any fixed term $\op h^{J}$ there exist at most $\n$ terms $\op h^{I}$ which do not commute with $\op h^{J}$.

  Then, for a $p$\textsuperscript{th} order product formula $\P_p$ with $p=1$ or $p=2k$, $k\ge1$ used to approximate the evolution operator under $H$, the approximation error for the exact evolution $U(T)$ with $T/\delta$ rounds of the product formula $\P_p\left(\delta\right)$ is bounded by
  \begin{align}
    \epsilon_p \left(T,\delta \right) \leq  C_1 \frac{T \delta^{p}}{\left(p+1\right)!} +  C_2 \frac{T}{\delta} \int_0^\delta p \int_0^1 \left(1-x\right)^{p-1}  \frac{x \tau^{p+1}}{p!} \ee^{x \tau N B_p} \dd x \dd\tau
  \end{align}
  with
  \begin{align}
    C_1 &= \n p B_p^{p+1} \Lambda^{p-1} N \left(\left(SM-1\right)+\frac{N}{\Lambda}\right)^{p-1} \left(\left(SM\right)^{2}-\left(SM\right)\right)\\
    C_2 &=\n  B_p^{p+2} \Lambda^{p} N   \left(\left(SM\right)^{p+2}-\left(SM\right)^{p+1}\right).
  \end{align}
\end{theorem}
\begin{proof}
  The error formula for a single Trotter step is given by \cref{eq:higher-trotter-error-1-ap} as
  \begin{align}
    \epsilon_p \left(\delta \right) &\le p \int_0^\delta \int_0^1 \left(1-x\right)^{p-1} \| \op R^{\left(p\right)}\left(x\tau\right)\| \frac{\tau^{p}}{p!} \dd x \dd\tau.
  \end{align}
  Evaluating this using \cref{lem:R-bound-Comms-ap} and then substituting the resultant expression in $\epsilon_p \left(T,\delta \right)\leq (T/\delta) \epsilon_p \left(\delta \right)$ gives the stated expression.
\end{proof}

For later reference, we note that it is straightforward to generalise the error bound in \cref{th:Trotter-Er-Commutator-ap}, by incorporating similar techniques to \cref{cor:trotter-error-ap} in order to sum up the $|\tcoeff_{ij}|$ exactly, instead of simply bounding them by $B_p$.
Additionally, we can also generalise to the case of a \emph{higher} derivative $\op R^{(q)}$, $q\ge p$, but still for a $p$\textsuperscript{th} order formula: with these two generalisations the bound simply reads
\begin{align}
    \epsilon_{p,q} \left(T,\delta \right) \leq  C_1 \frac{T \delta^{q}}{\left(q+1\right)!} +  C_2 \frac{T}{\delta} \int_0^\delta q \int_0^1 \left(1-x\right)^{q-1}  \frac{x \tau^{q+1}}{q!} \ee^{x \tau N B_p} \dd x \dd\tau
\end{align}
with
\begin{align}
    C_1 &= \n q B_p^{2} \Lambda^{q-1} N \left(M H_p - B_p + B_p \left(\frac{N}{\Lambda}\right)\right)^{q-1} \left(\left(S_p M\right)^{2}-\left(S_p M\right)\right)\\
    C_2 &= \n  B_p^{2}\left( M H_p \Lambda\right)^{q} N \left(\left(S_p M\right)^{2}-\left(S_p M\right)\right).
\end{align}

\subsection*{A Taylor Bound on the Taylor Bound}\label{ap:taylor-of-taylor}
Another method to obtain a tighter bound on a Taylor expansion as used on $\op R(\tau)$ in \cref{eq:integral-representation-ap} and which can be used together with the more sophisticated commutator-based error bound from \cref{th:Trotter-Er-Commutator-ap} derived in the last section, can be obtained by performing a Taylor expansion of the remainder term, and in turn bounding \emph{its} Taylor remainder by some other method \cite[Rem.~4]{Bausch2013}.

We first establish the following technical lemma:
\begin{lemma}[Taylor Error Bound]\label{lem:taylor-error-bound-ap}
Let the setup be as in \cref{lem:trotter-tech1-ap}, and let $q>p$.
The error term $\epsilon$ from \cref{eq:trotter-epsilon-ap} satisfies
\begin{align}
    \epsilon_p(\delta) &\le \sum_{l=p}^q \frac{\delta^{l+1}}{(l+1)!} \| \op R^{(l)}(0) \| + \epsilon_{p,q+1}(\delta)\\
\intertext{where}
    \op R(0)^{(l)} &= \sum_{\alpha:\,|\alpha|=p+1} \binom{p+1}{\alpha} \op F(\alpha) + \ii \op H \sum_{\beta:\,|\beta|=p} \binom{p}{\beta} \op F(\beta),\\
\intertext{with $\op H=\sum_{i=1}^M \op H_i$, and}
    \op F(\alpha) :&\!\!= \prod_{j=1}^S\prod_{i=1}^M (-\ii \tcoeff_{ji} \op H_i)^{\alpha_{(i,j)}}.
\end{align}
\end{lemma}
\begin{proof}
  The expression for $\epsilon$ stems from Taylor-expanding \cref{eq:higher-trotter-error-1-ap} to order $q$ instead of $p$, and integrating over $\tau$.
  The last term $\epsilon_{q+1,p}$ is then simply the overall remainder $\op S$, as before; and we can use \cref{eq:p-q-error-1-ap} to obtain a bound on it.
  The bound on $\| \op R(0)^{(l)} \|$ is an immediate consequence of \cref{eq:R-bound-ap}, where we set $\tau=0$.
\end{proof}
This allows us to calculate a numerical bound on $\| \op R(0)^{(l)} \|$, by bounding $\| \op H_i \| \le \Lambda$ and allowing terms within the two sums over $\alpha$ and $\beta$ to cancel.
The benefit of this approach is that it is generically applicable to any given Trotter formula, and only depends on the non-commuting layers of $\op H$.

We can therefore derive the following bounds:
\begin{corollary}[Taylor Error Bound]\label{cor:taylor-error-bound-ap}
Let $\op H=\sum_{i=1}^M \op H_i$ with $\| \op H_i \| \le \Lambda$ for all $i$.
Then for $\epsilon$ from \cref{eq:higher-trotter-error-1-ap}, and for a $p$\textsuperscript{th} Trotter formula, we have
\begin{align}
    \epsilon_p(\delta) &\le \sum_{l=p}^q \frac{\delta^{l+1}\Lambda^{l+1}}{(l+1)!} f(p,M,l)
    + \epsilon_{p,q+1}(\delta),\\
\intertext{where}
    f(p,M,l) &= \left\| \sum_{\alpha:\, |\alpha|=p+1} \binom{p+1}{\alpha} \op v(\alpha) + \ii \sum_{j=1}^M \ket j \otimes \sum_{\beta:\,|\beta|=p} \binom{p}{\beta} \op v(\beta) \right\|_1,\\
\intertext{and}
    \op v(\alpha) :&\!\!= \bigotimes_{j=1}^S \bigotimes_{i=1}^M(-\ii \tcoeff_{ji} \ket i)^{\otimes \alpha_{(i,j)}}.
\end{align}
for a basis $\{ \ket1, \ldots, \ket M \}$ of $\field C^M$.
\end{corollary}
\begin{proof}
Follows immediately from \cref{lem:taylor-error-bound-ap}.
\end{proof}
A selection of the series coefficients $f(p,M,l)$ can be found in \cref{tab:coefficients-ap}.
\begin{table}[t]
    \centering
\begin{tabular}{@{}llllllll@{}}
\toprule
 &  & \multicolumn{6}{c}{$f(p,M,l)$ for $l=\cdot$}   \\
 & M & l=p & l=p+1 & l=p+2 & l=p+3 & l=p+4 & l=p+5 \\ \midrule
\multirow{4}{*}{p=1} & 2 & 2 & 6 & 14 & 30 & 62 & 126 \\
 & 3 & 6 & 26 & 90 & 290 & 906 & 2786 \\
 & 4 & 12 & 68 & 312 & 1340 & 5592 & 22988 \\
 & 5 & 20 & 140 & 800 & 4292 & 22400 & 115220 \\ \midrule
\multirow{4}{*}{p=2} & 2 & 3 & 9 & 22.75 & 50 &  108.344 & 225.531 \\
 & 3 & 13 & 57 & 213.25 & 711.25 & 2309.47 & 7283.06 \\
 & 4 & 34 & 198 & 980.5 & 4377.5 & 18926.6 & 79758 \\
 & 5 & 70 & 510 & 3141.5 & 17555 & 94765.3 & 499391 \\ \midrule
\multirow{4}{*}{p=4} & 2 & 4.89745 & 19.5277 & 79.5305 & 442.266 & 2312.73 & 11208.3 \\
 & 3 & 43.6604 & 277.994 & 1880.62 & 16924.7 &  &  \\
 & 4 & 194.476 & 1719.69 & 16226.8 &  &  &  \\
 & 5 & 610.187 & 6926.95 & 83775.9 &  &  &  \\ \bottomrule
\end{tabular}
    \caption{Trotter error coefficients $f(p,M,l)$ from \cref{cor:taylor-error-bound-ap}; values rounded to the precision shown.}
    \label{tab:coefficients-ap}
\end{table}
\Cref{cor:taylor-error-bound-ap} can then be applied in conjunction with e.g.\ the commutator error bound given in \cref{th:Trotter-Er-Commutator-ap} for the remaining term $\epsilon_{q+1}(\delta, \delta)$.

\section*{Spectral Norm of Fermionic Hopping Terms}\label{ap:normbound}
\def\ad{a^{\dagger}}
Let $\ad$ and $a$ be the standard fermionic creation and annihilation operators.
\begin{theorem}\label{thm:norm-bound}
  Let $\Omega = \{ij \}$ be a set of pairs of indices such that no two pairs share an index. Define:
  \begin{align}
    H_{\Omega} = \sum_{ij \in \Omega} h_{ij} \;,\;\; h_{ij}=\ad_i a_j + \ad_j a_i
  \end{align}
  Given a normalized fermionic state $\ket{\psi}$ such that $N \ket{\psi} = n \ket{\psi}$:
  \begin{align}
    \vert \bra{\psi} H_{\Omega} \ket{\psi} \vert \leq \textrm{min}(n, M-n, \vert \Omega \vert)
    \end{align}
  Where $M$ is the number of fermionic modes.
  This bound is tight.
\end{theorem}

\begin{proof}
Consider that $h_{ij}$ has eigenvalues in $\{-1,0,1\}$, since $h_{ij}^2 = (N_i-N_j)^2$ which has eigenvalues $\{0,1\}$. Suppose there existed a normalised state $\ket{\psi}$ such that $N \ket{\psi} = n \ket{\psi}$ and $H_{\Omega} \ket{\psi} = \lambda \ket{\psi}$ where $\vert \lambda \vert > \textrm{min}(n, M-n, \vert \Omega \vert)$ . Since $H_Y$, $N$ and all $h_{ij}$ are all mutually commuting, we may choose $\ket{\psi}$ to be an eigenstate of all $h_{ij}$ wlog (by convexity). Then it must be the case that $h_{ij}^2 \ket{\psi} = \ket{\psi}$ for at least $\vert \lambda \vert$ pairs $ij$, which implies that in the Fock basis $\ket{\psi} = a \ket{0_i, 1_j, ..} + b \ket{1_i, 0_j, ..}  $. Therefore for at least $\vert \lambda \vert$ pairs $ij \in \Omega$  we have $\bra{\psi} (N_i+N_j) \ket{\psi} = 1$. So $\bra{\psi}N \ket{\psi} \geq  \vert \lambda \vert $ and $M-\bra{\psi}N \ket{\psi} \geq  \vert \lambda \vert$. If $\textrm{min}(n, M-n, \vert \Omega \vert) =n$ then $\bra{\psi}N \ket{\psi} >n$ which is a contradiction. If $\textrm{min}(n, M-n, \vert \Omega \vert)  = M-n$ then  $M-\bra{\psi}N \ket{\psi} >M-n$ which is a contradiction. If $\textrm{min}(n, M-n, \vert \Omega \vert) = \vert \Omega \vert$ then $\vert \lambda \vert > \vert \Omega \vert$ which is a contradiction. This proves the bound.

Now we need only show the bound is tight. Consider the following state:
\begin{align}
  \ket{\phi_{ij}^{\pm}}=  (\ad_i\pm\ad_j) \Gamma_{s}\ket{0}
\end{align}
With $\Gamma_{s}$ composed of creation and annihilation operators which do not include $i$ or $j$. This state is an eigenstate of $h_{ij}=\ad_i a_j + \ad_j a_i$:

\begin{align}
  h_{ij} \ket{\phi_{ij}^{\pm}} = \pm  \ket{\phi_{ij}^{\pm}}
\end{align}

Observe:
\begin{align}
h_{ij} \ket{\phi_{ij}^{\pm}} &=  h_{ij}(\ad_i\pm\ad_j) \Gamma_{s} \ket{0}\\
 &= (\ad_i a_j\ad_i + \ad_j a_i \ad_i \pm \ad_i a_j \ad_j \pm \ad_j a_i\ad_j) \Gamma_{s} \ket{0}\\
 &=  ( \ad_j a_i \ad_i \pm \ad_i a_j \ad_j ) \Gamma_{s} \ket{0}\\
&= ( (-\ad_j \ad_i a_i  +\ad_j)\pm (-\ad_i \ad_j a_j + \ad_i) ) \Gamma_{s} \ket{0}\\
 &= (\ad_j\pm  \ad_i ) \Gamma_{s} \ket{0}\\
h_{ij} \ket{\phi_{ij}^{\pm}} &= \pm (\ad_i\pm  \ad_j ) \Gamma_{s} \ket{0}
\end{align}

Consider a set of pairs of indices $\omega \subseteq \Omega$. Choose an ordering on $\omega$ and define
\begin{align}
  \ket{\phi_{\omega}^b} = \prod_{ij \in \omega } (\ad_i +(-1)^{b_{ij}}\ad_j) \ket{0}
\end{align}
  with $b$ a bit-string indexed by $ij$. Note that $N \ket{\phi_{\omega}^b} = \vert \omega \vert \ket{\phi_{\omega}^b}$. We now argue that $b$ can always be chosen such that:
  \begin{align}
    H_{\Omega} \ket{\phi_{\omega}^b} = \vert \omega \vert \ket{\phi_{\omega}^b}.
   \end{align}
Choose a pair $ij \in \omega$, the state $\ket{\phi^b_{\omega}}$ can be expressed as:
\begin{align}
  \ket{\phi^b_{\omega}} = (\delta_i \ad_i + (-1)^{b_{ij}} \delta_j \ad_j) \Gamma_s \ket{0},\; \delta_i,\delta_j \in \{-1,1\}
\end{align}
With $\Gamma_{s}$ composed of creation and annihilation operators which do not include $i$ or $j$. So
\begin{align}
  \ket{\phi^b_{\omega}} = \delta_i \ket{\phi_{ij}^{\Delta_{ij}}}\;\;,\Delta_{ij} = \delta_i \delta_j (-1)^{b_{ij}}.
  \end{align}
Let us choose $b_{ij}$ such that $\Delta_{ij} = 1$. Noting that $\Delta_{ij}$ is independent of $\Delta_{pq}$ when $pq \neq ij$ we can do the same for all other $b_{pq}$. This gives:
\begin{align}
  H_{\Omega} \ket{\phi_{\omega}^b} &=\sum_{ij \in \omega} h_{ij}\delta_i \ket{\phi_{ij}^{+}} \\
H_{\Omega} \ket{\phi_{\omega}^b}&=\sum_{ij \in \omega} \delta_i \ket{\phi_{ij}^{+}}\\
H_{\Omega} \ket{\phi_{\omega}^b}&=\sum_{ij \in \omega} \ket{\phi_{\omega}^b} \\
H_{\Omega} \ket{\phi_{\omega}^b} &= \vert \omega \vert \ket{\phi_{\omega}^b}
\end{align}
Note that $n=\vert \omega \vert < \vert \Omega \vert $ and $\vert \Omega \vert < M/2$ and so the bound is shown to be tight in the case where $\textrm{min}(n, M-n, \vert Y \vert)=n$.

If we consider the case where $\textrm{min}(n, M-n, \vert \Omega \vert)=\vert \Omega \vert$, then we may always choose $\omega$ such that it is composed of a set of pairs of indices such that no two pairs share an index, and such that $\Omega \subseteq \omega$. In this case, by a similar argument
\begin{align}
  H_{\Omega} \ket{\phi_{\omega}^b} = \vert \Omega \vert \ket{\phi_{\omega}^b}.
\end{align}

Finally, in the case where $\textrm{min}(n, M-n, \vert \Omega \vert)=M-n$ one may choose the particle-hole symmetric state
\begin{align}
  \ket{\tilde{\phi}^b_{\omega}} = \prod_{ij \in \omega } (a_i +(-1)^{b_{ij}}a_j) \prod_{k=1}^M \ad_k \ket{0}
\end{align}
and a similar argument follows by particle hole symmetry.
\end{proof}

\section*{Simulating Fermi-Hubbard via Sub-Circuit Algorithms}\label{sec:FHanalysis}
\subsection*{Overview and Benchmarking of Analysis}
In the following sections we primarily adopt the per-time error model and associated metric for costing circuits, \Cref{def:time-cost}. We first establish asymptotic bounds on the run-time $\cost$ of performing a time-dynamics simulation of a 2D spin Fermi-Hubbard Hamiltonian using a $p$\textsuperscript{th}-order Trotter formula with $M=5$ Trotter layers, for a target time $T$ and target error $\epsilon_t$. We perform this analysis for both the compact and VC encodings and the results are summarised in \cref{th:FH-optimum-analogue}.

We first want to compare using sub-circuit vs standard circuit decompositions in a per-time error model.  We do this in conjunction with our Trotter bounds. To this end we establish the analytic bounds for the same simulation task, first using the standard conjugation method to generate evolution under higher weight interactions using only standard CNOT gates and single qubit rotations as opposed to a sub-circuit pulse sequence. We choose this method as it doesn't introduce any unfair and needless analytic error into the comparison. We decompose the Trotter steps into a standard gate set of CNOTs and single-qubit rotations which are gates of the form $\ee^{\pm \ii \pi/4 ZZ}$ up to single qubit rotations. We cost this with the same metric using a per-time error model, but do not allow the comparison to contain any gates of the form $\ee^{\pm \ii \delta ZZ}$ as this would constitute a sub-circuit gate.

In this comparative analytic expression we still use our Trotter error bounds, so \cref{th:FH-optimum-dig} only serves to evaluate the impact of differing Trotter step synthesis methods on the asymptotic scaling of the run-time $\cost$ in a per-time error model.

Later in this section we perform a tighter numerical analysis of both our proposal and our standard circuit model comparison. In these numerics we compare our Trotter bounds to readily applicable bounds from the literature \cite[Prop.~F.4.]{Childs2017}. We point out that these bounds do not exploit the underlying structure of the Hamiltonian or make use of the recent advances of \cite{Childs2019}, \cite{Childsnew2019}. However these bounds contained all constants, were applicable to 2D lattices and could easily be evaluated for arbitrary $p$ and allowed us to make use of \cref{thm:norm-bound}. We were able to compare our bounds to \cite{Childs2019} for the case of a simple 1D lattice and establish that our bounds are preferable for medium system sizes, not in asymptotic limits of system size, as was our intention in reformulating bounds for NISQ applications.

\begin{figure}[t]
    \begin{subfigure}[b]{0.44\linewidth}
    \centering
        \includegraphics[height=4.5cm]{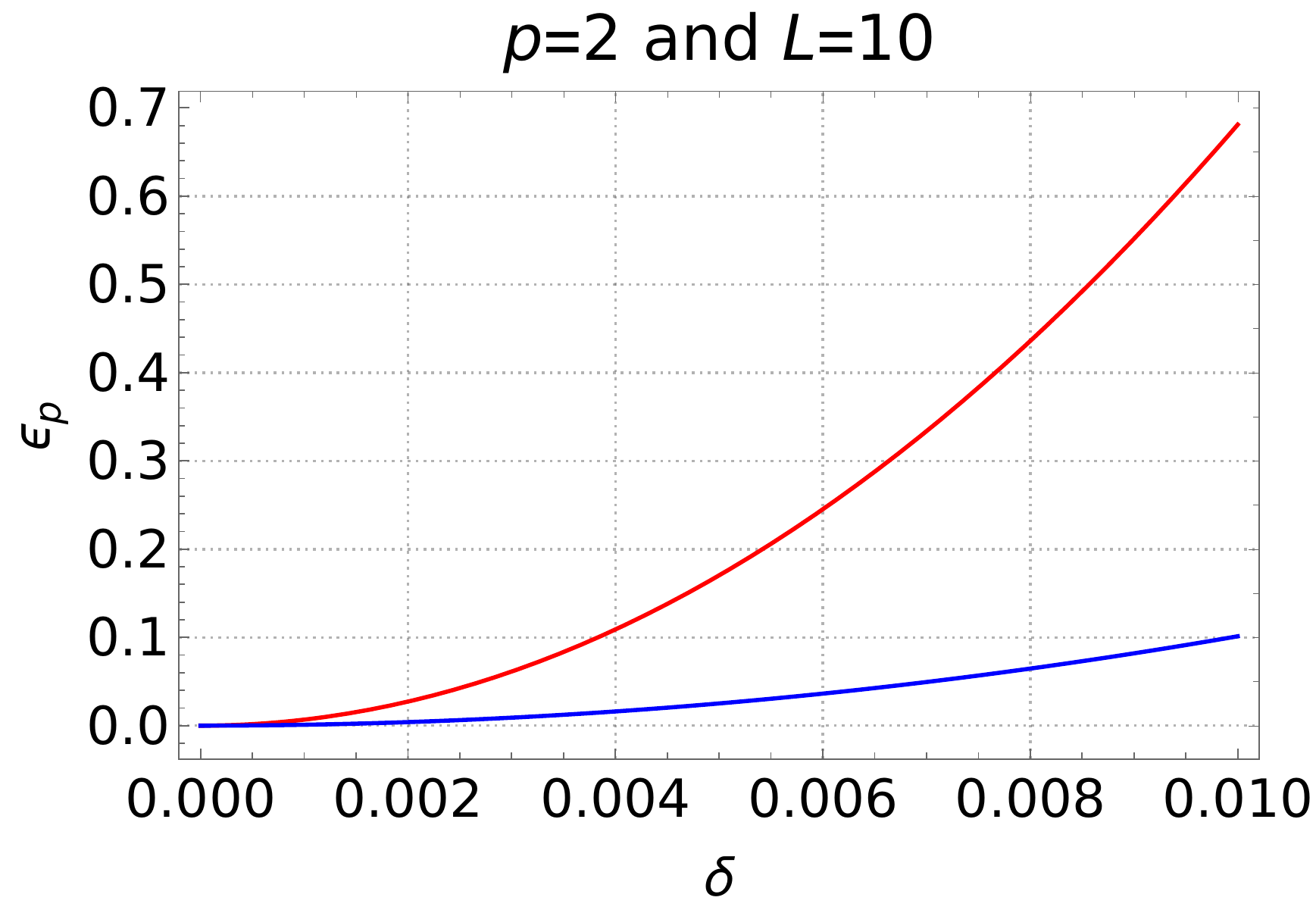}
    \end{subfigure}
    \hspace{1cm}
    \begin{subfigure}[b]{0.44\linewidth}
    \centering
        \includegraphics[height=4.5cm]{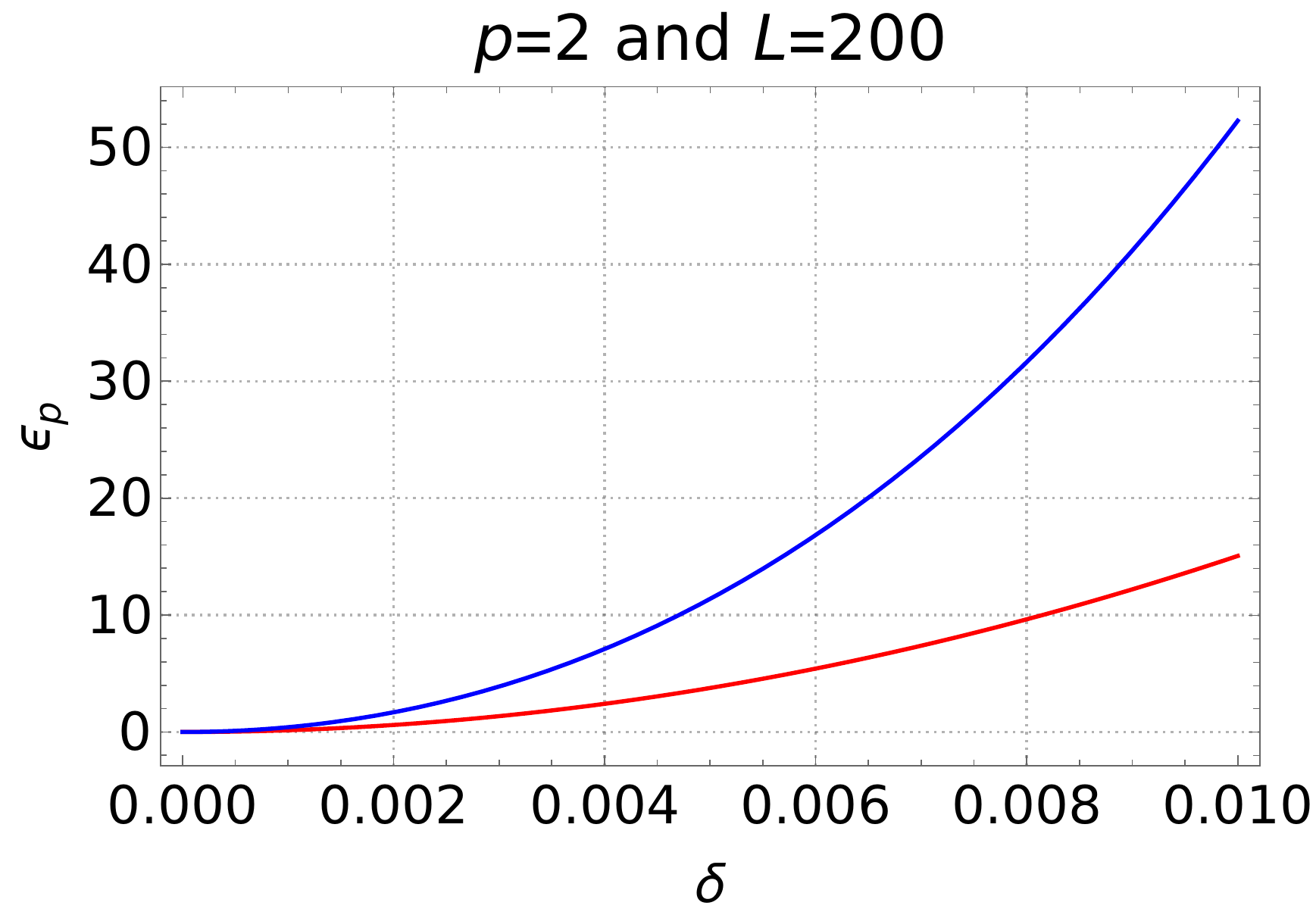}
    \end{subfigure}
    \caption{A comparison of $2$\textsuperscript{nd} order Trotter bounds $\epsilon_p(\delta,L)$ for a $1$D lattice Hamiltonian of length $L$ split into even and odd layers of interactions $H=H_{\text{even}}+H_{\text{odd}}$. The bounds are our \cref{th:Trotter-Er-Commutator-ap} (blue) and the main result stated in \cite{Childs2019} which we've evaluated from \cite[Sec.~B:eq.~57+58]{Childs2019} (red). This illustrates why we have derived Trotter bounds which sacrifice scaling in system size in favour of bounds with smaller constants in term of $p$. Here we evaluate \cref{th:Trotter-Er-Commutator-ap} with $\Lambda = L/2$ as this illustrative example is a generic, not fermionic, lattice Hamiltonian.
    }
    \label{fig:bounds-comparison}
\end{figure}

After this comparison in the framework of a per-time error model, we numerically analyse the impact of sub-circuit techniques in the per-gate error model, calculating the circuit depth of sub-circuit Trotter simulations. This is shown in \Cref{fig:cost-numeric-delta-ignore-pulse-lengths,fig:cost-analytic-delta-ignore-pulse-lengths}.

\subsection*{The Fermi-Hubbard Hamiltonian and Fermionic Encodings}
We consider a Fermi-Hubbard model on a 2D lattice of $N = L \times L$ fermionic sites. There is hopping between nearest neighbours only and on-site interactions between fermions of opposite spin. In terms of fermionic creation and annihilation operators the Hamiltonian for this system is
\begin{equation}\label{def:unencoded-H}
\op H_{\text{FH}} \coloneqq  \sum_{i=1}^{N} \op h_{\text{on-site}}^{(i)} \ + \sum_{i<j,\sigma} \op h_{\text{hopping}}^{(i,j,\sigma)}
\coloneqq  \os \sum_{i=1}^{N}   \ad_{i \uparrow} a_{i \uparrow} \ad_{i \downarrow} a_{i \downarrow}  +  \hop \sum_{i<j,\sigma}   \left(\ad_{i \sigma} a_{j \sigma} + \ad_{j\sigma} a_{i \sigma}\right),
\end{equation}
where $ \sigma \in \{ \uparrow, \downarrow \}$ and the sum over hopping terms runs over all nearest neighbour fermionic lattice sites $i$ and $j$. The interaction strengths are $\os$ and  $\hop$ and we assume that $\hop = 1$, and that they are bounded as $|\hop|,|\os| \leq r$. Before we proceed we have to choose how to encode this Hamiltonian in terms of spin operators. The choice of encoding has a significant impact on the run-time of the simulation. There are many encodings in the literature \cite{Jordan1928} but we will only analyse two, the Verstraete-Cirac (VC) encoding \cite{Verstraete2005}, and the recent compact encoding from \cite{DK}.

We choose our encoding in order to minimise the maximum Pauli weight of the encoded interaction terms. Using the VC and compact encodings this is constant at weight-$4$ and weight-$3$ respectively. In comparison the Jordan-Wigner encoding results in a maximum Pauli weight of the encoded interaction terms that scales with the lattice size as $\BigO\left(L\right)$, the Bravyi-Kitaev encoding \cite{Bravyi2002} has interaction terms of weight $O(\log L)$, and the Bravyi-Kitaev superfast encoding \cite{Bravyi2002} results in weight-$8$.

The encodings require the addition of ancillary qubits as well as two separate lattices encoding spin up and spin down fermions. For VC $4L^2$ qubits are needed to encode $L^2$ fermionic sites. In contrast compact requires $(L-1)^2$ ancillary qubits and $2L^2$ data qubits. The layout of these ancillary qubits are indicated in \cref{fig:ordering}. Note that we must also choose an ordering of the lattice sites. This is also indicated in \cref{fig:ordering}.

The two encodings map the Fermi-Hubbard Hamiltonian terms to interactions between qubits.
In both encodings, on-site interaction terms become
\begin{align}\label{eq:h-onsite}
     \op h_{\text{on-site}}^{(i)} \rightarrow \frac{\os}{4} \left(I - Z_{i \uparrow}\right)\left(I - Z_{i \downarrow}\right).
\end{align}
Only the encoded hopping terms differ. The exact expressions for hopping interactions depend on whether two nearest neighbour fermionic sites are horizontally or vertically connected on the lattice. The horizontally connected hopping terms are encoded as
\begin{align}
    \op h_{\text{hopping,hor}}^{(i,j,\sigma)} \rightarrow \frac{1}{2}
    \begin{cases}
      X_{i,\sigma}Z_{i',\sigma}X_{j,\sigma}+Y_{i,\sigma}Z_{i',\sigma}Y_{j,\sigma} & \text{VC}\\
     X_{i,\sigma}X_{j,\sigma}Y_{f_{ij}',\sigma}+Y_{i,\sigma}Y_{j,\sigma}Y_{f_{ij}',\sigma} & \text{compact}
    \end{cases},
\end{align}
while the vertically connected hopping terms are encoded as
\begin{align}
    \op h_{\text{hopping,vert}}^{(i,j,\sigma)} \rightarrow \frac{1}{2}
    \begin{cases}
      X_{i,\sigma}Y_{i',\sigma}Y_{j,\sigma}X_{j',\sigma}-Y_{i,\sigma}Y_{i',\sigma}X_{j,\sigma}X_{j',\sigma} & \text{VC} \\
        (-1)^{g(i,j)}\left( X_{i,\sigma}X_{j,\sigma}X_{f_{ij}',\sigma}+Y_{i,\sigma}Y_{j,\sigma}X_{f_{ij}',\sigma}\right) & \text{compact}
    \end{cases}.
\end{align}
In this notation $i$ labels the data qubit for lattice site $i$ and $\sigma$ its spin lattice. Dashed indices such as $i'$ refer to ancillary qubits. These are illustrated in grey in \Cref{fig:ordering}. In the VC encoding there is an ancillary qubit for every site on each spin lattice. In compact these ancillary qubits are laid out in a checker-board pattern on the faces on each spin lattice. Here $f_{ij}'$ labels the ancillary qubit to $i$ and $j$. There is also a sign determined by $g(i,j) = 0,1$. The details of this can be found in \cite{DK}.

\begin{figure}[t]
\begin{subfigure}[b]{0.5\linewidth}
\centering
\begin{tikzpicture}[x=2cm,y=2cm,scale=.75]
\draw[dashed][step=2cm,gray, thin] (0,0) grid (2,2);

\draw[solid][step=1cm,black, thick] (0,2) -- (2,2);
\draw[solid][step=1cm,black, thick] (0,1) -- (2,1);
\draw[solid][step=1cm,black, thick] (0,0) -- (2,0);

\draw[solid][step=1cm,black, thick] (2,2) -- (2,1);
\draw[solid][step=1cm,black, thick] (0,0) -- (0,1);

\node[label={[label distance=0.1mm]45:1}] at (0,2)[fill=black,circle,scale=0.5]{};

\node[label={[label distance=0.1mm, color=gray]45:1'}] at (0.5,1.5)[fill=gray,circle,scale=0.5]{};

\node[label={[label distance=0.1mm]45:2}] at (1,2)[fill=black,circle,scale=0.5]{};

\node[label={[label distance=0.1mm, color=gray]45:2'}] at (1.5,1.5)[fill=gray,circle,scale=0.5]{};

\node[label={[label distance=0.1mm]45:3}] at (2,2)[fill=black,circle,scale=0.5]{};

\node[label={[label distance=0.1mm, color=gray]45:3'}] at (2.5,1.5)[fill=gray,circle,scale=0.5]{};

\node[label={[label distance=0.1mm]45:4}] at (2,1)[fill=black,circle,scale=0.5]{};

\node[label={[label distance=0.1mm, color=gray]45:4'}] at (1.5,0.5)[fill=gray,circle,scale=0.5]{};

\node[label={[label distance=0.1mm]45:5}] at (1,1)[fill=black,circle,scale=0.5]{};

\node[label={[label distance=0.1mm, color=gray]45:5'}] at (0.5,0.5)[fill=gray,circle,scale=0.5]{};

\node[label={[label distance=0.1mm]45:6}] at (0,1)[fill=black,circle,scale=0.5]{};

\node[label={[label distance=0.1mm, color=gray]45:6'}] at (-0.5,0.5)[fill=gray,circle,scale=0.5]{};

\node[label={[label distance=0.1mm]45:7}] at (0,0)[fill=black,circle,scale=0.5]{};

\node[label={[label distance=0.1mm, color=gray]45:7'}] at (0.5,-0.5)[fill=gray,circle,scale=0.5]{};

\node[label={[label distance=0.1mm]45:8}] at (1,0)[fill=black,circle,scale=0.5]{};

\node[label={[label distance=0.1mm, color=gray]45:8'}] at (1.5,-0.5)[fill=gray,circle,scale=0.5]{};

\node[label={[label distance=0.1mm]45:9}] at (2,0)[fill=black,circle,scale=0.5]{};

\node[label={[label distance=0.1mm, color=gray]45:9'}] at (2.5,-0.5)[fill=gray,circle,scale=0.5]{};

\end{tikzpicture}
\caption{VC Encoding.}
\end{subfigure}
\begin{subfigure}[b]{0.5\linewidth}
\centering
\begin{tikzpicture}[x=2cm,y=2cm,scale=.75]
\draw[dashed][step=2cm,gray, thin] (0,0) grid (2,2);

\draw[solid][step=1cm,black, thick] (0,2) -- (2,2);
\draw[solid][step=1cm,black, thick] (0,1) -- (2,1);
\draw[solid][step=1cm,black, thick] (0,0) -- (2,0);

\draw[solid][step=1cm,black, thick] (2,2) -- (2,1);
\draw[solid][step=1cm,black, thick] (0,0) -- (0,1);

\node[label={[label distance=0.1mm]45:1}] at (0,2)[fill=black,circle,scale=0.5]{};

\node[label={[label distance=0.1mm, color=gray]45:$f'_{12}$}]at (0.5,1.5)[fill=gray,circle,scale=0.5]{};

\node[label={[label distance=0.1mm]45:2}] at (1,2)[fill=black,circle,scale=0.5]{};

\node[label={[label distance=0.1mm]45:3}] at (2,2)[fill=black,circle,scale=0.5]{};

\node[label={[label distance=0.1mm]45:4}] at (2,1)[fill=black,circle,scale=0.5]{};

\node[label={[label distance=0.1mm, color=gray]45:$f'_{54}$}] at (1.5,0.5)[fill=gray,circle,scale=0.5]{};

\node[label={[label distance=0.1mm]45:5}] at (1,1)[fill=black,circle,scale=0.5]{};

\node[label={[label distance=0.1mm]45:6}] at (0,1)[fill=black,circle,scale=0.5]{};

\node[label={[label distance=0.1mm]45:7}] at (0,0)[fill=black,circle,scale=0.5]{};

\node at (0.5,-0.5)[fill=white,circle,scale=0.5]{};

\node[label={[label distance=0.1mm]45:8}] at (1,0)[fill=black,circle,scale=0.5]{};

\node at (1.5,-0.5)[fill=white,circle,scale=0.5]{};

\node[label={[label distance=0.1mm]45:9}] at (2,0)[fill=black,circle,scale=0.5]{};

\node at (2.5,-0.5)[fill=white,circle,scale=0.5]{};

\end{tikzpicture}
\caption{compact Encoding.}
\end{subfigure}
\caption{The ordering of qubits for $L=2$ and the layout of ancillary qubits (grey) for each encoding. This figure only depicts a single spin lattice. Additionally, the first VC ancillary qubit $f'_{12}$ is also labelled by $f'_{25}$, $f'_{56}$ and $f'_{16}$. Similarly for the other (VC) ancillary qubit.}\label{fig:ordering}
\end{figure}
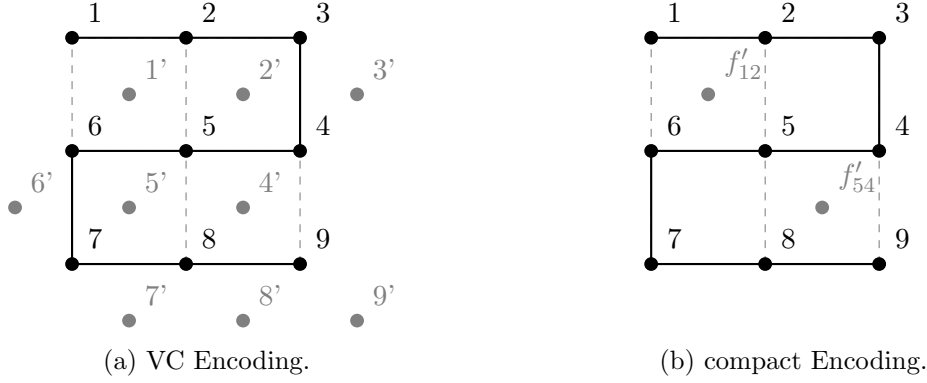

\subsection*{Choice of Trotter Layers}
We group the interactions into $5$ Trotter layers. Every pair of interactions within a layer must be disjoint. Under the assumptions of \Cref{def:short-pulse-circuit} all interactions within a single layer can then be implemented in parallel. For both encodings the five layers consist of all on-site interactions $\op H_5$, two alternating layers of horizontal hopping interactions $\op H_1$ and $\op H_2$ and, two alternating layers of vertical hopping interactions $\op H_3$ and $\op H_4$. Both cases are illustrated in \cref{fig:VC-layers-hor,fig:VC-layers-vert,fig:compact-layers-hor,fig:compact-layers-vert} for the case of $L=5$.

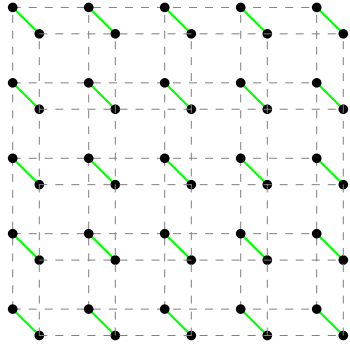
\begin{figure}[t]
\centering
\begin{tikzpicture}
\draw[dashed][step=1cm,gray,very thin] (0,0) grid (4,4);

\foreach \x in {0,1,2,3,4}{
    \foreach \y in {0,1,2,3,4}{
        \draw[solid][step=1cm,green,thick] (\x,\y) -- (\x+0.35,\y-0.35);
        }
    }

\foreach \x in {0,1,2,3,4}{
    \foreach \y in {0,1,2,3,4}{
        \node at (\x,\y)[fill=black,circle,scale=0.35]{};
        }
    }

\foreach \x in {0.35,1.35,2.35,3.35,4.35}{
    \foreach \y in {-0.35,0.65,1.65,2.65,3.65}{
        \node at (\x,\y)[fill=black,circle,scale=0.35]{};
        }
    }
\begin{scope}[shift={(0.35,-0.35)}]
\draw[dashed][step=1cm,gray,very thin] (0,0) grid (4,4);
\end{scope}
\end{tikzpicture}
\caption{The green lines connecting pairs of qubits represent a single on-site interaction term in either encoding $\op h_{\text{on-site}}^{(i)}$. }\label{fig:os-layer}
\end{figure}

The on-site interaction terms are the same in both cases and do not involve any ancillary qubits. They are shown in \Cref{fig:os-layer}, where the ancillary qubits are consequently not depicted. The hopping terms all act within a single spin lattice. They are shown for the VC encoding in \Cref{fig:VC-layers-hor} and \Cref{fig:VC-layers-vert} for a single spin lattice, and for compact these are shown in \Cref{fig:compact-layers-hor} and \Cref{fig:compact-layers-vert}

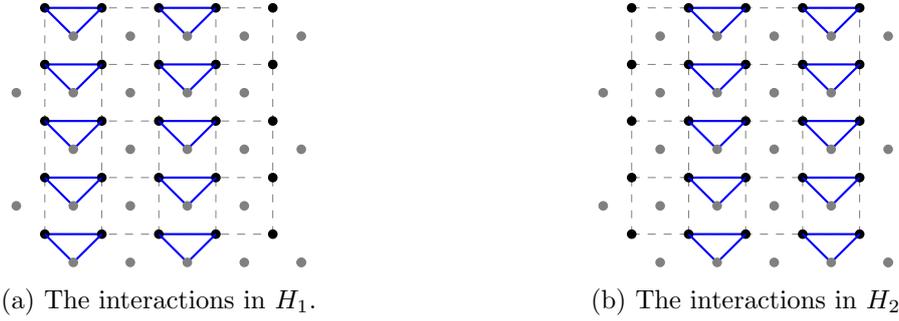
\begin{figure}[t]
\begin{subfigure}[b]{0.5\linewidth}
\centering
\begin{tikzpicture}[scale=0.75]
\draw[dashed][step=1cm,gray,very thin] (0,0) grid (4,4);

\foreach \x in {0,2}{
    \foreach \y in {0,1,2,3,4}{
        \draw[solid][step=1cm,blue,thick] (\x,\y) -- (\x+1,\y);
        }
    }

\foreach \x in {0,1,2,3,4}{
    \foreach \y in {0,1,2,3,4}{
        \node at (\x,\y)[fill=black,circle,scale=0.35]{};
        }
    }
\foreach \x in {0.5,2.5}{
    \foreach \y in {-0.5,0.5,1.5,2.5,3.5}{
        \draw[solid][step=1cm,blue,thick] (\x,\y) -- (\x-0.5,\y+0.5);
         \draw[solid][step=1cm,blue,thick] (\x,\y) -- (\x+0.5,\y+0.5);
        }
    }

\foreach \x in {0.5,1.5,2.5,3.5}{
    \foreach \y in {-0.5,0.5,1.5,2.5,3.5}{
        \node at (\x,\y)[fill=gray,circle,scale=0.35]{};
        }
    }
\node at (4.5,3.5)[fill=gray,circle,scale=0.35]{};
\node at (-0.5,2.5)[fill=gray,circle,scale=0.35]{};
\node at (4.5,1.5)[fill=gray,circle,scale=0.35]{};
\node at (-0.5,0.5)[fill=gray,circle,scale=0.35]{};
\node at (4.5,-0.5)[fill=gray,circle,scale=0.35]{};
\node at (4.5,-0.5)[fill=gray,circle,scale=0.35]{};
\end{tikzpicture}
\caption{The interactions in $\op H_1$.}
\end{subfigure}
\begin{subfigure}[b]{0.5\linewidth}
\centering
\begin{tikzpicture}[scale=0.75]

\draw[dashed][step=1cm,gray,very thin] (0,0) grid (4,4);

\foreach \x in {1,3}{
    \foreach \y in {0,1,2,3,4}{
        \draw[solid][step=1cm,blue,thick] (\x,\y) -- (\x+1,\y);
        }
    }

\foreach \x in {0,1,2,3,4}{
    \foreach \y in {0,1,2,3,4}{
        \node at (\x,\y)[fill=black,circle,scale=0.35]{};
        }
    }

\foreach \x in {1.5,3.5}{
    \foreach \y in {-0.5,0.5,1.5,2.5,3.5}{
        \draw[solid][step=1cm,blue,thick] (\x,\y) -- (\x-0.5,\y+0.5);
         \draw[solid][step=1cm,blue,thick] (\x,\y) -- (\x+0.5,\y+0.5);
        }
    }

\foreach \x in {0.5,1.5,2.5,3.5}{
    \foreach \y in {-0.5,0.5,1.5,2.5,3.5}{
        \node at (\x,\y)[fill=gray,circle,scale=0.35]{};
        }
    }
\node at (4.5,3.5)[fill=gray,circle,scale=0.35]{};
\node at (-0.5,2.5)[fill=gray,circle,scale=0.35]{};
\node at (4.5,1.5)[fill=gray,circle,scale=0.35]{};
\node at (-0.5,0.5)[fill=gray,circle,scale=0.35]{};
\node at (4.5,-0.5)[fill=gray,circle,scale=0.35]{};
\node at (4.5,-0.5)[fill=gray,circle,scale=0.35]{};
\end{tikzpicture}
\caption{The interactions in $\op H_2$}
\end{subfigure}
\caption{The blue lines connecting three qubits represent a single horizontal hopping term in the VC encoding: $\op h_{\text{hopping,hor}}^{(i,j,\sigma)}$}\label{fig:VC-layers-hor}
\end{figure}
\begin{figure}[t]
\begin{subfigure}[b]{0.5\linewidth}
\centering
\begin{tikzpicture}[scale=0.75]

\draw[dashed][step=1cm,gray,very thin] (0,0) grid (4,4);

\foreach \x in {0,1,2,3,4}{
    \foreach \y in {0,2}{
        \draw[solid][step=1cm,blue,thick] (\x,\y) -- (\x,\y+1);
        }
    }

\foreach \x in {0,1,2,3,4}{
    \foreach \y in {0,1,2,3,4}{
        \node at (\x,\y)[fill=black,circle,scale=0.35]{};
        }
    }

\foreach \x in {-0.5,0.5,1.5,2.5,3.5}{
    \foreach \y in {0.5,2.5}{
         \draw[solid][step=1cm,blue,thick] (\x,\y) -- (\x+0.5,\y+0.5);
        }
    }

\foreach \x in {0.5,1.5,2.5,3.5,4.5}{
    \foreach \y in {-0.5,1.5}{
         \draw[solid][step=1cm,blue,thick] (\x,\y) -- (\x-0.5,\y+0.5);
        }
    }

\foreach \x in {0.5,1.5,2.5,3.5}{
    \foreach \y in {-0.5,0.5,1.5,2.5,3.5}{
        \node at (\x,\y)[fill=gray,circle,scale=0.35]{};
        }
    }
\node at (4.5,3.5)[fill=gray,circle,scale=0.35]{};
\node at (-0.5,2.5)[fill=gray,circle,scale=0.35]{};
\node at (4.5,1.5)[fill=gray,circle,scale=0.35]{};
\node at (-0.5,0.5)[fill=gray,circle,scale=0.35]{};
\node at (4.5,-0.5)[fill=gray,circle,scale=0.35]{};
\node at (4.5,-0.5)[fill=gray,circle,scale=0.35]{};
\end{tikzpicture}
\caption{The interactions in $\op H_3$.}
\end{subfigure}
\begin{subfigure}[b]{0.5\linewidth}
\centering
\begin{tikzpicture}[scale=0.75]

\draw[dashed][step=1cm,gray,very thin] (0,0) grid (4,4);

\foreach \x in {0,1,2,3,4}{
    \foreach \y in {1,3}{
        \draw[solid][step=1cm,blue,thick] (\x,\y) -- (\x,\y+1);
        }
    }

\foreach \x in {0,1,2,3,4}{
    \foreach \y in {0,1,2,3,4}{
        \node at (\x,\y)[fill=black,circle,scale=0.35]{};
        }
    }

\foreach \x in {-0.5,0.5,1.5,2.5,3.5}{
    \foreach \y in {0.5,2.5}{
         \draw[solid][step=1cm,blue,thick] (\x,\y) -- (\x+0.5,\y+0.5);
        }
    }

\foreach \x in {0.5,1.5,2.5,3.5,4.5}{
    \foreach \y in {1.5,3.5}{
         \draw[solid][step=1cm,blue,thick] (\x,\y) -- (\x-0.5,\y+0.5);
        }
    }

\foreach \x in {0.5,1.5,2.5,3.5}{
    \foreach \y in {-0.5,0.5,1.5,2.5,3.5}{
        \node at (\x,\y)[fill=gray,circle,scale=0.35]{};
        }
    }
\node at (4.5,3.5)[fill=gray,circle,scale=0.35]{};
\node at (-0.5,2.5)[fill=gray,circle,scale=0.35]{};
\node at (4.5,1.5)[fill=gray,circle,scale=0.35]{};
\node at (-0.5,0.5)[fill=gray,circle,scale=0.35]{};
\node at (4.5,-0.5)[fill=gray,circle,scale=0.35]{};
\node at (4.5,-0.5)[fill=gray,circle,scale=0.35]{};

\end{tikzpicture}
\caption{The interactions in $\op H_4$.}
\end{subfigure}
\caption{The blue lines connecting four qubits represent a single vertical hopping term in the VC encoding: $\op h_{\text{hopping,vert}}^{(i,j,\sigma)}$}\label{fig:VC-layers-vert}
\end{figure}
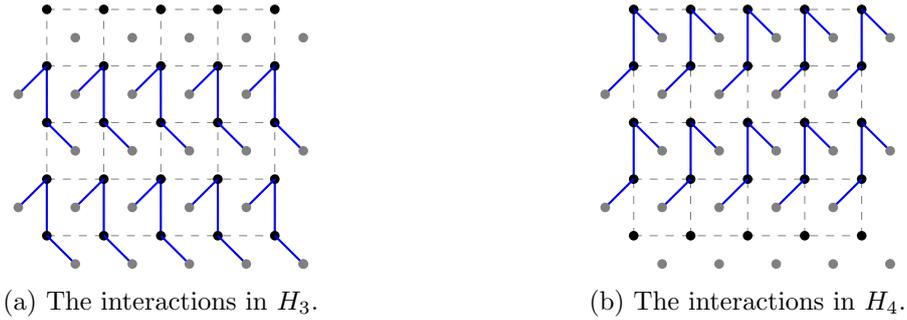

The alternating horizontal layers and alternating vertical layers are chosen to ensure that all pairs of interactions are disjoint and not just commuting. Note that we could have chosen to lay out the alternating horizontal and vertical layers in the VC encoding in the same fashion as the compact as depicted in \Cref{fig:compact-layers-hor} or \Cref{fig:compact-layers-vert}.

These are not the only choices of Trotter layers. In a later Supplementary Method called ``Regrouping Interaction Terms'' we show that we can implement a $p$\textsuperscript{th} order product formula with only $3$ Trotter layers. We do this for the compact encoding only as it is particularly neat. This is despite grouping the interactions in a way where not all interactions within a layer commute with one another. A combination of this and the previous results still enable us to directly implement each layer without incurring any further analytic error.

\begin{figure}[t]
\begin{subfigure}[b]{0.5\linewidth}
\centering
\begin{tikzpicture}[scale=0.75]

\draw[dashed][step=1cm,gray,very thin] (0,0) grid (4,4);

\foreach \x in {0,2}{
    \foreach \y in {0,2,4}{
        \draw[solid][step=1cm,red,thick] (\x,\y) -- (\x+1,\y);
        }
    }

\foreach \x in {1,3}{
    \foreach \y in {1,3}{
        \draw[solid][step=1cm,red,thick] (\x,\y) -- (\x+1,\y);
        }
    }

\foreach \x in {1.5,3.5}{
    \foreach \y in {0.5,2.5}{
        \draw[solid][step=1cm,red,thick] (\x,\y) -- (\x-0.5,\y+0.5);
         \draw[solid][step=1cm,red,thick] (\x,\y) -- (\x+0.5,\y+0.5);
        }
    }

\foreach \x in {0.5,2.5}{
    \foreach \y in {1.5,3.5}{
        \draw[solid][step=1cm,red,thick] (\x,\y) -- (\x-0.5,\y+0.5);
         \draw[solid][step=1cm,red,thick] (\x,\y) -- (\x+0.5,\y+0.5);
        }
    }

\foreach \x in {1.5,3.5}{
    \foreach \y in {0.5,2.5}{
        \node at (\x,\y)[fill=gray,circle,scale=0.35]{};
        }
    }

\foreach \x in {0.5,2.5}{
    \foreach \y in {1.5,3.5}{
        \node at (\x,\y)[fill=gray,circle,scale=0.35]{};
        }
    }

\foreach \x in {0,1,2,3,4}{
    \foreach \y in {0,1,2,3,4}{
        \node at (\x,\y)[fill=black,circle,scale=0.35]{};
        }
    }
\end{tikzpicture}
\caption{The interactions in $\op H_1$.}
\end{subfigure}
\begin{subfigure}[b]{0.5\linewidth}
\centering
\begin{tikzpicture}[scale=0.75]

\draw[dashed][step=1cm,gray,very thin] (0,0) grid (4,4);

\foreach \x in {1,3}{
    \foreach \y in {0,2,4}{
        \draw[solid][step=1cm,red,thick] (\x,\y) -- (\x+1,\y);
        }
    }

\foreach \x in {0,2}{
    \foreach \y in {1,3}{
        \draw[solid][step=1cm,red,thick] (\x,\y) -- (\x+1,\y);
        }
    }

\foreach \x in {1.5,3.5}{
    \foreach \y in {0.5,2.5}{
        \draw[solid][step=1cm,red,thick] (\x,\y) -- (\x-0.5,\y-0.5);
         \draw[solid][step=1cm,red,thick] (\x,\y) -- (\x+0.5,\y-0.5);
        }
    }

\foreach \x in {0.5,2.5}{
    \foreach \y in {1.5,3.5}{
        \draw[solid][step=1cm,red,thick] (\x,\y) -- (\x-0.5,\y-0.5);
         \draw[solid][step=1cm,red,thick] (\x,\y) -- (\x+0.5,\y-0.5);
        }
    }

\foreach \x in {1.5,3.5}{
    \foreach \y in {0.5,2.5}{
        \node at (\x,\y)[fill=gray,circle,scale=0.35]{};
        }
    }

\foreach \x in {0.5,2.5}{
    \foreach \y in {1.5,3.5}{
        \node at (\x,\y)[fill=gray,circle,scale=0.35]{};
        }
    }

\foreach \x in {0,1,2,3,4}{
    \foreach \y in {0,1,2,3,4}{
        \node at (\x,\y)[fill=black,circle,scale=0.35]{};
        }
    }
\end{tikzpicture}
\caption{The interactions in $\op H_2$.}
\end{subfigure}
\caption{The red lines connecting three qubits represent a single horizontal hopping term in the compact encoding $\op h_{\text{hopping,hor}}^{(i,j,\sigma)}$.}\label{fig:compact-layers-hor}
\end{figure}
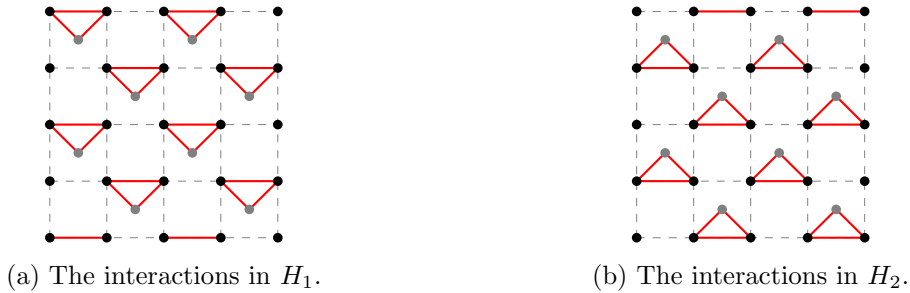
\begin{figure}[t]
\begin{subfigure}[b]{0.5\linewidth}
\centering
\begin{tikzpicture}[scale=0.75,rotate=90,x=-1cm]

\draw[dashed][step=1cm,gray,very thin] (0,0) grid (4,4);

\foreach \x in {0,2}{
    \foreach \y in {0,2,4}{
        \draw[solid][step=1cm,red,thick] (\x,\y) -- (\x+1,\y);
        }
    }

\foreach \x in {1,3}{
    \foreach \y in {1,3}{
        \draw[solid][step=1cm,red,thick] (\x,\y) -- (\x+1,\y);
        }
    }

\foreach \x in {1.5,3.5}{
    \foreach \y in {0.5,2.5}{
        \draw[solid][step=1cm,red,thick] (\x,\y) -- (\x-0.5,\y+0.5);
         \draw[solid][step=1cm,red,thick] (\x,\y) -- (\x+0.5,\y+0.5);
        }
    }

\foreach \x in {0.5,2.5}{
    \foreach \y in {1.5,3.5}{
        \draw[solid][step=1cm,red,thick] (\x,\y) -- (\x-0.5,\y+0.5);
         \draw[solid][step=1cm,red,thick] (\x,\y) -- (\x+0.5,\y+0.5);
        }
    }

\foreach \x in {1.5,3.5}{
    \foreach \y in {0.5,2.5}{
        \node at (\x,\y)[fill=gray,circle,scale=0.35]{};
        }
    }

\foreach \x in {0.5,2.5}{
    \foreach \y in {1.5,3.5}{
        \node at (\x,\y)[fill=gray,circle,scale=0.35]{};
        }
    }

\foreach \x in {0,1,2,3,4}{
    \foreach \y in {0,1,2,3,4}{
        \node at (\x,\y)[fill=black,circle,scale=0.35]{};
        }
    }
\end{tikzpicture}
\caption{The interactions in $\op H_3$.}
\end{subfigure}
\begin{subfigure}[b]{0.5\linewidth}
\centering
\begin{tikzpicture}[scale=0.75,rotate=90,x=-1cm]

\draw[dashed][step=1cm,gray,very thin] (0,0) grid (4,4);

\foreach \x in {1,3}{
    \foreach \y in {0,2,4}{
        \draw[solid][step=1cm,red,thick] (\x,\y) -- (\x+1,\y);
        }
    }

\foreach \x in {0,2}{
    \foreach \y in {1,3}{
        \draw[solid][step=1cm,red,thick] (\x,\y) -- (\x+1,\y);
        }
    }

\foreach \x in {1.5,3.5}{
    \foreach \y in {0.5,2.5}{
        \draw[solid][step=1cm,red,thick] (\x,\y) -- (\x-0.5,\y-0.5);
         \draw[solid][step=1cm,red,thick] (\x,\y) -- (\x+0.5,\y-0.5);
        }
    }

\foreach \x in {0.5,2.5}{
    \foreach \y in {1.5,3.5}{
        \draw[solid][step=1cm,red,thick] (\x,\y) -- (\x-0.5,\y-0.5);
         \draw[solid][step=1cm,red,thick] (\x,\y) -- (\x+0.5,\y-0.5);
        }
    }

\foreach \x in {1.5,3.5}{
    \foreach \y in {0.5,2.5}{
        \node at (\x,\y)[fill=gray,circle,scale=0.35]{};
        }
    }

\foreach \x in {0.5,2.5}{
    \foreach \y in {1.5,3.5}{
        \node at (\x,\y)[fill=gray,circle,scale=0.35]{};
        }
    }

\foreach \x in {0,1,2,3,4}{
    \foreach \y in {0,1,2,3,4}{
        \node at (\x,\y)[fill=black,circle,scale=0.35]{};
        }
    }
\end{tikzpicture}
\caption{The interactions in $\op H_4$.}
\end{subfigure}
\caption{The red lines connecting three qubits represent a single vertical hopping term in the compact encoding: $\op h_{\text{hopping,vert}}^{(i,j,\sigma)}$. }\label{fig:compact-layers-vert}
\end{figure}
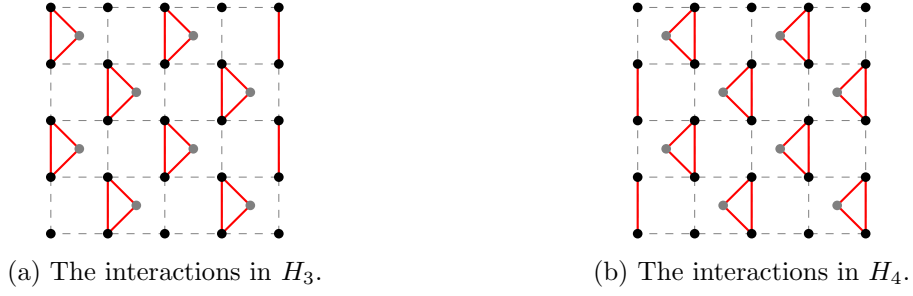

The norm of these layers appears in the Trotter bounds we derive. We bound these as $ \| \op H_i \|  \leq \Lambda$ for all $i$. In \Cref{thm:norm-bound} it is shown that $\Lambda$ can be related to fermion number and this fact is used to obtain tighter bounds on the Trotter error in the numerics we perform. We confine ourselves to a sector of $5$ fermions in all our numerical calculations. We do this both for pragmatic reasons, since the Hilbert space dimension is just large enough to be classically hard, and because for a $5 \times 5$ lattice a roughly quarter-filling already implies interesting crossover phenomena appear \cite{Keller2001,Kaneko2014}. We also leave $\Lambda$ explicit in our analytic bounds so that we can explore different parameter regions in later work.

\subsection*{Analytic Run-Time Bounds for Simulating Fermi Hubbard}
Now we can proceed with obtaining analytic bounds on the run-time of this simulation for each encoding. Throughout this section we assume a per-time error model. For
the recursive product formula in \cref{eq:P-2k-ap} with either $p=1$ or $p=2k$ for $k\ge1$ and $M$ non-commuting Hamiltonian layers $\op H=\sum_{i=1}^M \op H_i$, the cost of the simulation in terms of the single most expensive Trotter layer is
\begin{align}\label{eq:trotter-time}
\cost\left( \P_p\left(\delta\right)^{T/\delta} \right) \le \frac{MT}{\delta} \times \cost\left(\Umax\left(\op H,\delta B_p\right)\right) \times \begin{cases}
    1& p=1 \\
    2\times 5^{p/2-1} & \text{$p=2k$ for $k\ge1$},
\end{cases}
\end{align}
where $\Umax\left(\op H, \tau \right) \coloneqq \argmax_{ \op U_i } \{ \cost\left( \op U_i \right) \}$ for $\op U_i \coloneqq \exp\left(\ii \tau \op H_i\right)$ and $B_p$ given in \cref{rem:bounds-on-trotter-coefficients-ap}. This follows from the definitions of the product formula in \cref{eq:P-2k-ap}.

We proceed by obtaining bounds on the run-time of the most costly Trotter layer in each encoding. This expression depends on whether we use our methods (summarised in the main text in \cref{eq:intro-gatedec-3,eq:intro-gatedec-4}) or the conjugation method (see \cref{eq:conj-method} in the main text) to implement each Trotter step.

For the VC encoding the most costly layers will be the the vertical hopping layers $\op H_{3}$ and $\op H_{4}$. As they are both a sum of disjoint terms which we assume can be performed in parallel this is simply given by the cost of implementing a single vertical hopping interaction. Remembering that the interaction strengths satisfy $| \hop |$, $| \os| \leq r$, we bound this as
\begin{align}
    \cost \left(\Umax\left(\op H_{\text{VC}}, \delta B_p \right) \right)&\leq \cost\left( \ee^{\ii B_p \delta \frac{r}{2} \left( XYYX-YYXX\right) } \right) \\
    & = 2 \cost\left( \ee^{\ii B_p \delta \frac{r}{2} Z^{\otimes 4} } \right) \\
    &\leq 2 \times  \begin{cases}
            7 ( \frac{r}{2} B_p \delta)^{1/3} & \text{sub-circuit}\\
            6 \left(\frac{\pi}{4}\right) & \text{standard}
           \end{cases}
\end{align}
The second simplification follows from both terms in the interaction commuting, thus allowing them to be performed sequentially. The same is true for the compact encoding and so we have
\begin{align}
    \cost \left(\Umax\left(\op H_{\text{compact}}, \delta B_p \right) \right)&\leq\cost\left( \ee^{\ii B_p \delta \frac{r}{2} \left( XXY+YYY\right) } \right) \\
    & = 2 \cost\left( \ee^{\ii B_p \delta \frac{r}{2} Z^{\otimes 3} } \right)\\
    & \leq 2 \times\begin{cases}
            2 (  2 \frac{r}{2} B_p \delta)^{1/2} & \text{sub-circuit}\\
             4 \left(\frac{\pi}{4}\right) & \text{standard}
           \end{cases}
\end{align}
The final expressions now depend only on how we decompose local Trotter steps, either in terms of CNOT gates and single qubit rotations or using circuits such as those in \cref{fig:circuits}. The concrete bounds on $\cost \left(\Umax\left(\op H, \delta B_p \right) \right)$ are summarised and simplified in \cref{tab:maxInteractionCost}.
\begin{figure}[t]
    \centering
    \begin{tabular}{ l p{3cm} p{3cm}}
        \toprule
        & \multicolumn{2}{c}{$\cost\left(\Umax\left(\op H,\delta B_p\right)\right)$} \\
        \midrule
         Encoding  & standard & sub-circuit \\
        \midrule
         compact & $2 \pi$ & $4 \sqrt{B_p r \delta}$\\
         VC & $3 \pi$ & $12 \sqrt[3]{B_p r \delta}$ \\
        \bottomrule
    \end{tabular}
    \caption{Cost of implementing the highest weight interaction term in the encoded Fermi-Hubbard Hamiltonian in a per-time error model. Decomposing a $k$-local evolution in terms of the standard CNOT conjugation method has overhead $2 (k-1) \times \pi/4$. The overhead associated with sub-circuit synthesis follows from \cref{eq:decomp-bounds}. }
    \label{tab:maxInteractionCost}
\end{figure}

Substituting the bounds in \cref{tab:maxInteractionCost} into  \cref{eq:trotter-time} results in run-times $\propto \BigO(\delta^{-1})$ and $\propto \BigO(\delta^{\frac{2-k}{k-1}} )$, assuming decomposition via \cref{eq:conj-method} or \cref{eq:intro-gatedec-3,eq:intro-gatedec-4} in the main text, respectively. As both of these expressions diverge as $\delta \rightarrow 0$ it is optimal to maximise $\delta$ with respect to an allowable analytic Trotter error $\epsilon_t$. This is captured in the following lemma which uses the simplest bounds on Trotter error established previously.

\begin{lemma}[Optimal FH $\delta$]\label{lem:FH-max-delta}
For a target error rate $\epsilon_t$, the maximum Trotter step for a $p$\textsuperscript{th} formula saturating the error bound in \Cref{th:trotter-error-ap} is
\begin{align}
    \delta_0 = \left(\frac{\epsilon_t}{T M^{p+1} \Lambda^{p+1}}\right)^{1/p} \times \begin{cases}
     \displaystyle 1 &  p=1 \\begin{align}10pt]
     \displaystyle \left(\frac{(p+1)!}{2}\right)^{1/p} \left( \frac{3}{10} \right)^{p/2-1/2-1/p} &  p=2k\ \text{for}\ k\ge 1.
    \end{cases}
\end{align}
\end{lemma}
\begin{proof}
Follows from \cref{th:trotter-error-ap} by solving for $\delta$.
\end{proof}

Now we can obtain the final analytic bounds on the total run-time of simulating the Fermi-Hubbard Hamiltonian for each of these four cases.

\begin{corollary}[Standard-circuit Minimised Run-time]\label{th:FH-optimum-dig}
If standard synthesis techniques are used to implement local Trotter steps in terms of CNOT gates and single-qubit rotations with an optimal Trotter step size $\delta_0$ saturating \Cref{lem:FH-max-delta}, the simulation cost for the Fermi-Hubbard Hamiltonian with a $p$\textsuperscript{th} order Trotter formula with maximum error $\epsilon_t$ is as follows
\begin{align}
\cost(\P_p(\delta_0)^{T/\delta_0}) &\leq
    \begin{cases}
    f_p  \  M^{2 + \frac{1}{p}} \Lambda^{1 + \frac{1}{p}} T^{1+\frac{1}{p}} \epsilon^{-1/p}_t  & \text{VC}  \\
    g_p  \  M^{2 + \frac{1}{p}} \Lambda^{1 + \frac{1}{p}} T^{1+\frac{1}{p}} \epsilon^{-1/p}_t & \text{compact}
    \end{cases}
\end{align}
with
\begin{align}
    f_p
    &=3 \pi \times \begin{cases}
     \displaystyle 1 &  p=1 \\[10pt]
     \displaystyle 2^{\frac{p+1}{2}} 3^{-\frac{p}{2}+\frac{1}{p}+\frac{1}{2}} 5^{p-\frac{1}{p}-\frac{3}{2}}  (p+1)!^{-\frac{1}{p}} &  p=2k\ \text{for}\ k\ge 1
    \end{cases}
\end{align}
and
\begin{align}
    g_p
    &=2 \pi \times \begin{cases}
     1 &  p=1 \\[10pt]
     \displaystyle 2^{\frac{p+1}{2}} 3^{-\frac{p}{2}+\frac{1}{p}+\frac{1}{2}} 5^{p-\frac{1}{p}-\frac{3}{2}}  (p+1)!^{-\frac{1}{p}} &  p=2k\ \text{for}\ k\ge 1.
    \end{cases}
\end{align}
\end{corollary}
\begin{proof}
The proof follows by choosing $\delta$ such that the bound obtained in \cref{lem:FH-max-delta} is saturated and substituting this and the respective expressions in \cref{tab:maxInteractionCost} into \Cref{eq:trotter-time}.
\end{proof}

\begin{corollary}[Sub-Circuit Minimised Run-time]\label{th:FH-optimum-analogue}
If sub-circuit synthesis techniques are used to implement local Trotter steps with an optimal Trotter step size $\delta_0$ saturating \Cref{lem:FH-max-delta}, the simulation cost for the Fermi-Hubbard Hamiltonian with a $p$\textsuperscript{th} order Trotter formula with maximum error $\epsilon_t$ is as follows
\begin{align}
\cost(\P_p(\delta_0)^{T/\delta_0}) &\leq
    \begin{cases}
    f_p \ r^{1/3}   M^{5/3 + 2/(3p)} \Lambda^{2/3 + 2/(3p)} T^{1+2/(3p)} \epsilon^{-2/(3p)}_t  & \text{VC}  \\
    g_p \ r^{1/2}   M^{3/2 + 1/(2p)} \Lambda^{1/2 +1/(2p)} T^{1 + 1/(2p)} \epsilon^{-1/(2p)}_t & \text{compact}
    \end{cases}
\end{align}
with
\begin{align}
    f_p
    &=12 \times \begin{cases}
     \displaystyle 1 &  p=1 \\[10pt]
     \displaystyle 2^{\frac{p}{2}} 3^{\frac{1}{6} \left(-3 p+\frac{4}{p}+4\right)} 5^{\frac{1}{6} \left(5 p-\frac{4}{p}-8\right)} (p+1)!^{-\frac{2}{3 p}} &  p=2k\ \text{for}\ k\ge 1
    \end{cases}
\end{align}
and
\begin{align}
    g_p
    &=4 \times \begin{cases}
     1 &  p=1 \\[10pt]
     \displaystyle 2^{\frac{p}{2}-\frac{1}{4}} 3^{\frac{1}{4} \left(-2 p+\frac{2}{p}+3\right)} 5^{\frac{1}{4} \left(3 p-\frac{2}{p}-5\right)}  (p+1)!^{-\frac{1}{2 p}} &  p=2k\ \text{for}\ k\ge 1.
    \end{cases}
\end{align}
\end{corollary}
\begin{proof}
The proof follows by choosing $\delta$ such that the bound obtained in \Cref{lem:FH-max-delta} is saturated and substituting this and the respective expressions in \cref{tab:maxInteractionCost} into \Cref{eq:trotter-time}.
\end{proof}

\subsection*{Trivial Stochastic Error Bound}\label{sec:stochasticerror}

\let\O\relax
\newcommand{\O}{\mathcal O}
\DeclareDocumentCommand{\gateA}{m m}{%
    \draw[fill=white] (#1,#2-0.15) rectangle (#1+0.5,#2+1.15);
}
\DeclareDocumentCommand{\gateRow}{m m O{2} O{7}}{%
    \foreach \x in {0,...,#4}{
        \gateA{\x*#3+#1}{#2}
    }
}
\DeclareDocumentCommand{\errorRow}{m m O{2}}{%
    \foreach \x in {0,...,8}{%
        \draw[fill=white] (\x*#3+#1,#2) circle[radius=.35] node[] {$\mathcal E$};
    }
}
\begin{figure}
\centering
\begin{tikzpicture}[x=6.6mm,y=6.6mm]
\begin{scope}
\clip (-2.5,-2.25) rectangle (17.6,7.75);
\foreach \l/\ll in {-2/-2,-1/-1,0/,1/+1,2/+2}{
    \draw (-1,\l) node[left,xshift=-1]{$\l$} -- (17.5,\l);
    \errorRow{-.35}{\l}
    \draw (-1,\l+5.5) node[left,xshift=-1]{$n\ll$} -- (17.5,\l+5.5);
    \errorRow{-.35}{\l+5.5}
}
\foreach \l in {-2,0,2}{
    \gateRow{.25}{\l}[4][3]
    \gateRow{.25}{\l+5.5}[4][3]
}
\foreach \l in {-3,-1,1}{
    \gateRow{2.25}{\l}[4][3]
    \gateRow{2.25}{\l+5.5}[4][3]
}
\draw[fill=white,draw=none] (-1,2.32) rectangle (16,3.18);
\draw[fill=white,draw=none] (-1,2.18) rectangle (16,2.25);
\draw[fill=white,draw=none] (-1,3.24) rectangle (16,3.33);

\draw[fill=white,draw=none] (7.2,-3) rectangle (8.8,9);
\draw[fill=white,draw=none] (7.0,-3) rectangle (7.1,9);
\draw[fill=white,draw=none] (8.9,-3) rectangle (9.0,9);
\draw[fill=white] (16.7,5.5) circle[radius=.35] node[] {$\O$};
\end{scope}
\end{tikzpicture}
\caption{Saturated circuit model with intermediate errors $\mathcal E$, e.g.\ depolarizing noise $\mathcal E = \cN_q$ for some noise parameter $p$ given in \cref{eq:dep-channel}.
At the end of the circuit, an observable $\mathcal O$ is measured.
Drawn is a one-dimensional circuit; naturally, a similar setup can be derived for a circuit on a $2$-dimensional qudit lattice, for interactions shown in \cref{fig:compact-layers-hor,fig:compact-layers-vert,fig:VC-layers-hor,fig:VC-layers-vert,fig:os-layer}.}\label{fig:circuit-dense-2}
\end{figure}
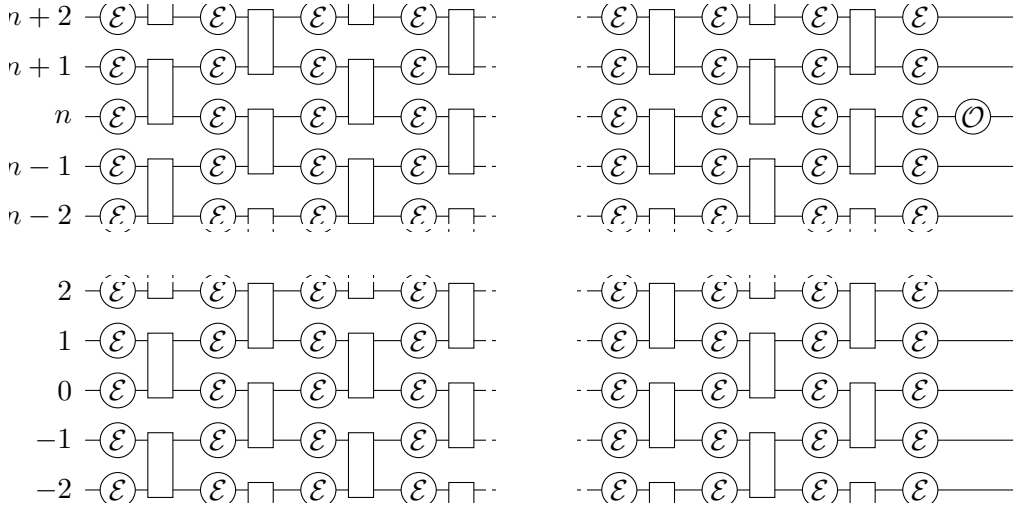

So far we have only considered the unitary error introduced by approximating the real Hamiltonian evolution with a Trotterized approximation.
However, in a near-term quantum device without error correction in place, we expect the simulated evolution to be noisy.
We model the noise by interspersing each circuit gate in the product formula by an iid channel $\mathcal E$; for simplicity we will assume that $\mathcal E$ is a single qubit depolarising channel $\mathcal E = \cN_q$, defined as
\begin{equation}\label{eq:dep-channel}
    \cN_q=(1-q)I+q \mathcal T,
    \quad
    \rho \longmapsto (1-q)\rho + \frac{q}{d}I
\end{equation}
for a noise parameter $q\in[0,1]$.\footnote{Strictly speaking $\cN_q$ defines a completely positive trace preserving map for all $p\le1+1/(d^2-1)$.
We emphasise that the error analysis which follows also works for a more general channel than the depolarising one.}
Here $I$ denotes the identity channel and $\mathcal T$ takes any state to the maximally mixed state $\tau=I/d$.

A trivial error bound for a circuit as in \cref{fig:circuit-dense-2} can then be found by just calculating the probability of no error occuring at all; disregarding the beneficial effects of a causal lightcone behind the observable $\O$, and denoting with $\mathcal U := \Ucirc^\dagger \cdot \Ucirc$ the clean circuit, and with $\mathcal U'$ the circuit saturated with intermediate errors, we get the expression
\begin{align}\label{eq:trivial-error}
    \epsilon = \left|\Tr\left[ (\mathcal U(\rho) - \mathcal U'(\rho))\O \right]\right|
    \le 1 - (1-q)^V,
\end{align}
where $V$ is the circuit's volume (i.e.\ the number of $\mathcal E$ interspersed in $\mathcal U'$).
It is clear to see that this error bound asymptotically approaches $1$, and does so exponentially quickly.
Thus, to stay below a target error rate $\epsilon_\text{tar}$, a sufficient condition is that
\begin{align}
    \label{eq:err-V-bound}
    1-(1-q)^V < \epsilon_\text{tar}
    \quad&\Longleftrightarrow\quad
    V < \log\left(\frac{1-\epsilon_\text{tar}}{1-q}\right),
    \\
\intertext{or alternatively}
    \label{eq:err-q-bound}
    &\Longleftrightarrow\quad
    q < 1 - \sqrt[V]{1-\epsilon_\text{tar}}.
\end{align}

Instead of assuming that each error channel $\mathcal E$ in \cref{fig:circuit-dense-2} has the same error probability $q$, we can analyse the case where $q$ is proportional to the pulse length of the preceding or anteceding gate; corresponding relations as given in \cref{eq:err-V-bound,eq:err-q-bound} can readily be derived numerically.

\subsection*{Error Mapping under Fermionic Encodings}\label{sec:2ndorder}

In \cite{Error-Mapping}, the authors analyse how noise on the physical qubits translates to errors in the fermionic code space.
To first order and in the W3 encoding, all of $\{X, Y, Z\}$ errors on the face, and $\{X, Y\}$ on the vertex qubits can be detected.
$Z$ errors on the vertex qubits -- as evident from the form of $\op h_\mathrm{on-site}$ from \cref{eq:h-onsite} -- result in an undetectable error; as shown in \cite[Sec.~3.2]{Error-Mapping}, this $Z$ error induces fermionic phase noise.

It is therefore a natural extension to the notion of simulation to allow for some errors occur -- if they correspond to physical noise in the fermionic space.
And indeed, as discussed more extensively in \cite[Sec.~2.4]{Error-Mapping}, phase noise is a natural setting for many fermionic condensed matter systems coupled to a phonon bath \cite{Ng2015,Kauch2020,Zhao2017,Melnikov2016,Openov2005,Fedichkin2004,Scully1993} and \cite[Ch.~6.1\&eq.~6.17]{Wellington2014}.

We further assume that we can measure all stabilisers (including a global parity operator) once at the end of the entire circuit.
We could imagine measuring these stabilisers after each individual gate of the form $\ee^{\ii \delta h_i}$ -- where $h_i$ is any term in the Hamiltonian. However, as every stabiliser commutes with every term in the Hamiltonian the outcome of the stabiliser measurement is unaffected and so we need only measure all stabilisers once at the end of the entire circuit. It is evident 
that measuring the stabilisers can be done by dovetailing an at most depth $4$ circuit to the end of our simulation -- much like measuring the stabilisers of the Toric code. 
It is thus a negligible overhead to the cost of simulation $\cost$.

However, errors may occur within the decomposition of gates of the form $\ee^{\ii \delta h_i}$ into single qubit rotations and two-qubit gates of the form $\ee^{\ii t ZZ}$. The stabilisers do not generally commute with these one- and two-local gates. In spite of this we can commute a Pauli error which occurs within a decomposition past the respective gates that make up that decomposition. For example, consider $h_i = X_1 Z_2 X_3$. Then we would decompose $\ee^{\ii \delta h_i}$ first via $\ee^{\ii \delta X_1 Z_2 X_3} \approx \ee^{\ii \delta^{1/2} X_1 X_2}\ee^{-\ii \delta^{1/2} Y_2 X_3}\ee^{-\ii \delta^{1/2} X_1 X_2}\ee^{\ii \delta^{1/2} Y_2 X_3}$ and then furthermore using identities such as $\ee^{\ii \delta^{1/2} X_1 X_2} = H_1 H_2 \ee^{\ii \delta^{1/2} Z_1 Z_2} H_1 H_2$. Suppose a Pauli X error were to occur on the second qubit during one of these steps, say for example instead leading us to perform the circuit $H_1 H_2 \ee^{\ii \delta^{1/2} Z_1 Z_2} X_2 H_1 H_2 \ee^{-\ii \delta^{1/2} Y_2 X_3}\ee^{-\ii \delta^{1/2} X_1 X_2}\ee^{\ii \delta^{1/2} Y_2 X_3}$. By commuting the error $X_2$ to the end of the circuit we can see that this is still $\mathcal{O}(\delta^{1/2}$) close to performing $Z_2 \ee^{\ii \delta X_1 Z_2 X_3}$ as $H Z = X H$, which is an error we can detect as previously discussed. As a Pauli Z error is equally likely to occur this doesn't introduce any noise bias. Therefore in cases where $\delta$ is sufficiently small, we can use error detection to reduce the fidelity requirements needed in the device by accounting for this additional error term. This is done numerically and shown in the dashed lines of \Cref{fig:cost-analytic-delta-ignore-pulse-lengths,fig:cost-analytic-delta}.

This means that while the fermionic encodings do not provide error correction, they do allow error detection to some extent;
we summarise all first order error mappings in \cref{tab:error-mapping}.
This means we can numerically simulate the occurrence of depolarising noise throughout the circuit, map the errors to their respective syndromes, and classify the resulting detectable errors, as well as non-detectable phase, and non-phase noise.

This means we can analyse $\cost$ with a demonstrably suppressed error, by allowing the non-detectable non-phase noise to saturate a target error bound.
The resulting simulation is such that it corresponds to a faithful simulation of the fermionic system, but where we allow fermionic phase error occurs -- where we emphasise that since detectable non-phase errors occur roughly with the same probability as non-detectable phase errors we know that, in expectation, only $O(1)$ phase error occurs throughout the simulation; in brief, it is not a very noisy simulation after all.

The resulting required depolarising noise parameters for various FH setups we summarise in \cref{fig:cost-analytic-delta,fig:cost-numeric-delta,fig:cost-analytic-delta-ignore-pulse-lengths,fig:cost-numeric-delta-ignore-pulse-lengths}, and the resulting post-selection probabilities in \cref{fig:postsel-probs}.

\begin{table}[t]
\centering
\begin{tabular}{rrl}
    \toprule
    Location & Syndrome & Effect \\
    \midrule
    \multirow{3}{*}{Vertex} & $X$ & detectable \\
     & $Y$ & detectable \\
     & $Z$ & detectable \\
    \midrule
    \multirow{3}{*}{\shortstack[r]{Face\\(Ancilla)}} & $X$ & detectable \\
     & $Y$ & detectable \\
     & $Z$ & phase noise \\
     \bottomrule
\end{tabular}
\caption{Error mapping from first order physical noise to the encoded fermionic code space, under the W3 encoding, by \cite{Error-Mapping}.
All but $Z$ errors on the faces are detectable; the latter result in fermionic phase noise.}\label{tab:error-mapping}
\end{table}

\subsection*{Numeric Results}\label{ap:numeric-results}
We can tighten the preceding analysis in several ways. First of, instead of crudely upper bounding the cost of individual gates we can sum these pulse times exactly. To this end, we use both the explicitly defined Trotter formulae coefficients $h_{ij}$, and also the exact formulae for the pulse times derived previously.

Secondly, we can use tighter bounds for Trotter error, which take into account the commutation relations between pairs of interactions across Trotter layers $\op H_i$, the coefficients defined in \Cref{rem:bounds-on-trotter-coefficients-ap}. Additionally, we obtain a bound which rewrites the Trotter error as a Taylor series, and then bound the Taylor remainder using methods which usually bound the total Trotter error; which bound is tighter depends on the order $p$ of the formula, and the target simulation time.
As explained in \cref{lem:FH-max-delta}, a tighter Trotter error allows us to choose a larger $\delta$ while achieving the same error $\epsilon_t$, reducing the total cost of the simulation.

Finally, as explained in the Supplementary Method entitled ``Trivial Stochastic Error Bound'', a certain simulation circuit size will determine the amount of stochastic error present within the circuit.
We assume that the depolarising noise precedes each Trotter layer; in the error-per-time model we assume that the noise parameter scales with the pulse length; in the more traditional error-per-gate model we assume that the noise is independent of the pulse length.

Furthermore, instead of taking the analytic Trotter bounds, we can also numerically calculate the optimal Trotter pulse length by calculating $\epsilon_p(T,\delta)$ from \cref{eq:trotter-epsilon-ap} explicitly, and maximizing $\delta$ till $\epsilon_p$ saturates an upper target error rate.
Naturally, this is computationally costly; so we perform this calculation for FH lattices up to a size $3\times 3$, and extrapolate these numbers by qubit count to the desired lattice lengths up to $10\times 10$; the dependence of $\delta_0$ on the target time and number of qubits was extracted from the asymptotic Trotter bound in \cref{lem:FH-max-delta}, i.e.
\begin{equation}\label{eq:fitted-delta-dependency}
    \delta_0 = \left(a_0 + \frac{b_0}{T^{1/p}}\right)\left(a_1 + \frac{b_1}{\Lambda^{(p+1)/p} } \right)
\end{equation}
for four fit parameters $a_0, a_1, b_0, b_1$; all other dependencies are assumed constant.
We show such numerically fitted data in \cref{fig:numeric-extrapolation}; a similar analysis was made for Trotter orders $p=1, 2, 4, 6$, and target error rates $\epsilon_t = 0.1, 0.05, 0.01$.
These numerical bounds are much tigher than the analytical bounds, but can still be assumed to overestimate the actual error: the operator norm distance between Trotter evolution and true evolution likely overestimates the error that would occur if starting from a particular initialized state.
We found that randomizing the Trotter layer order does not yield any advantage over keeping it fixed, and similarly we did not obtain an advantage by permuting the fixed order further.

\begin{figure}[t]
    \centering
    \includegraphics[width=\textwidth]{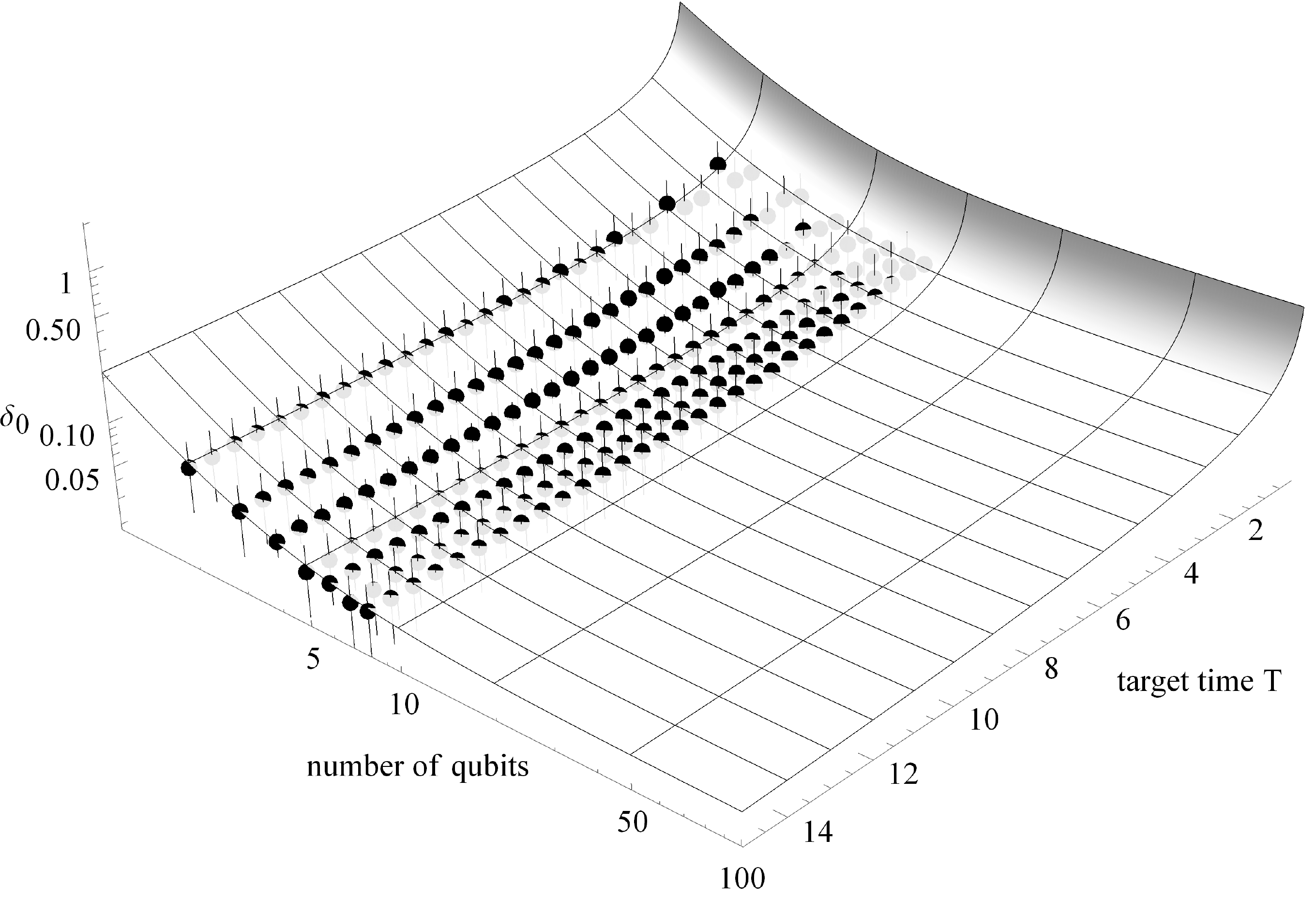}
    \caption{Extraploated optimal Trotter error step size $\delta_0$.
    The data points stem from the simulation of FH lattices up to size $3\times 3$, for a target error rate $\epsilon_t=0.1$, Trotter order $p=2$, and for times $T=0.5,1.0,\ldots,15.0$.
    The fitted surface follows the formula in \cref{eq:fitted-delta-dependency}.}
    \label{fig:numeric-extrapolation}
\end{figure}

\Cref{tab:NumericCost} compares these numerics for the case of a lattice with $L=5$ and for a sufficiently long time in terms of units set by the lattice spacing as we assume $\hop =1$. So for any $L$ we choose $T=  \lfloor \sqrt{2}L  \rfloor$. We choose an analytic error of $\epsilon_t = 0.1$ as there is no point making the analytic error smaller than the experimental error present in NISQ-era gates.
\begin{table}[]
    \centering
    \begin{tabular}{c c c c}
    \toprule
     & Trotter Bounds & Standard & Subcircuit \\
    \midrule
     \multirow{2}{*}{VC}
      & analytic & 95,409 & 17,100 \\
      & numeric & 4,234 & 1,669 \\
     \midrule
     \multirow{2}{*}{compact}
         & analytic & 77,236 & 1,686 \\
         & numeric & 3,428 & 259 \\
    \end{tabular}
    \caption{Per-time. A comparison of the run-time $\cost$ for lattice size $L\times L$ with $L=5$, overall simulation time $T=7$ and target Trotter error $\epsilon_\mathrm{target} = 0.1$, with $\Lambda=5$ fermions and coupling strengths $|\os|, |\hop|\le r=1$.
    Obtained by minimising over product formulas up to $4$\textsuperscript{th} order.
    $\cost=\cost(\P_p(\delta_0)^{T/\delta_0})$ for per-time error model.
    In either gate decomposition case---standard and sub-circuit---we account single-qubit rotations as a free resource; the value of $\cost$ depends only on the two-qubit gates/interactions. Two-qubit unitaries are counted by their respective pulse lengths. Here compact and VC denote the choice of fermionic encoding.}
    \label{tab:NumericCost}
\end{table}
\begin{table}[]
    \centering
    \begin{tabular}{c c c c}
    \toprule
     & Trotter Bounds & Standard & Subcircuit \\
    \midrule
     \multirow{2}{*}{VC}
      & analytic & 121,478 & 95,447 \\
      & numeric & 5391 & 4236 \\
     \midrule
     \multirow{2}{*}{compact}
         & analytic & 98,339 & 72,308 \\
         & numeric & 4,364 & 3,209 \\
    \end{tabular}
    \caption{Per-gate. A comparison of the run-time $\cost$ for lattice size $L\times L$ with $L=5$, overall simulation time $T=7$ and target Trotter error $\epsilon_\mathrm{target} = 0.1$, with $\Lambda=5$ fermions and coupling strengths $|\os|, |\hop|\le r=1$.
    Obtained by minimising over product formulas up to $4$\textsuperscript{th} order.
    $\cost=$~circuit-depth for per-gate error model.
    In either gate decomposition case---standard and sub-circuit---we account single-qubit rotations as a free resource; the value of $\cost$ depends only on the two-qubit gates/interactions. Two-qubit unitaries are counted by unit time per gate in the per gate error model.  Here compact and VC denote the choice of fermionic encoding.}
\end{table}
As the compact encoding results in the smallest run-time we investigate how $\cost$ varies with $\epsilon_t$ for $L=3$, $5$ and $10$ below. In these numerics we choose the order $p$ which minimise $\cost$ at each value of $T$.

\begin{figure}[t]
    \centering
    \vspace{-1cm}
    \centerline{\begin{minipage}{18cm}
    \includegraphics[width=\textwidth]{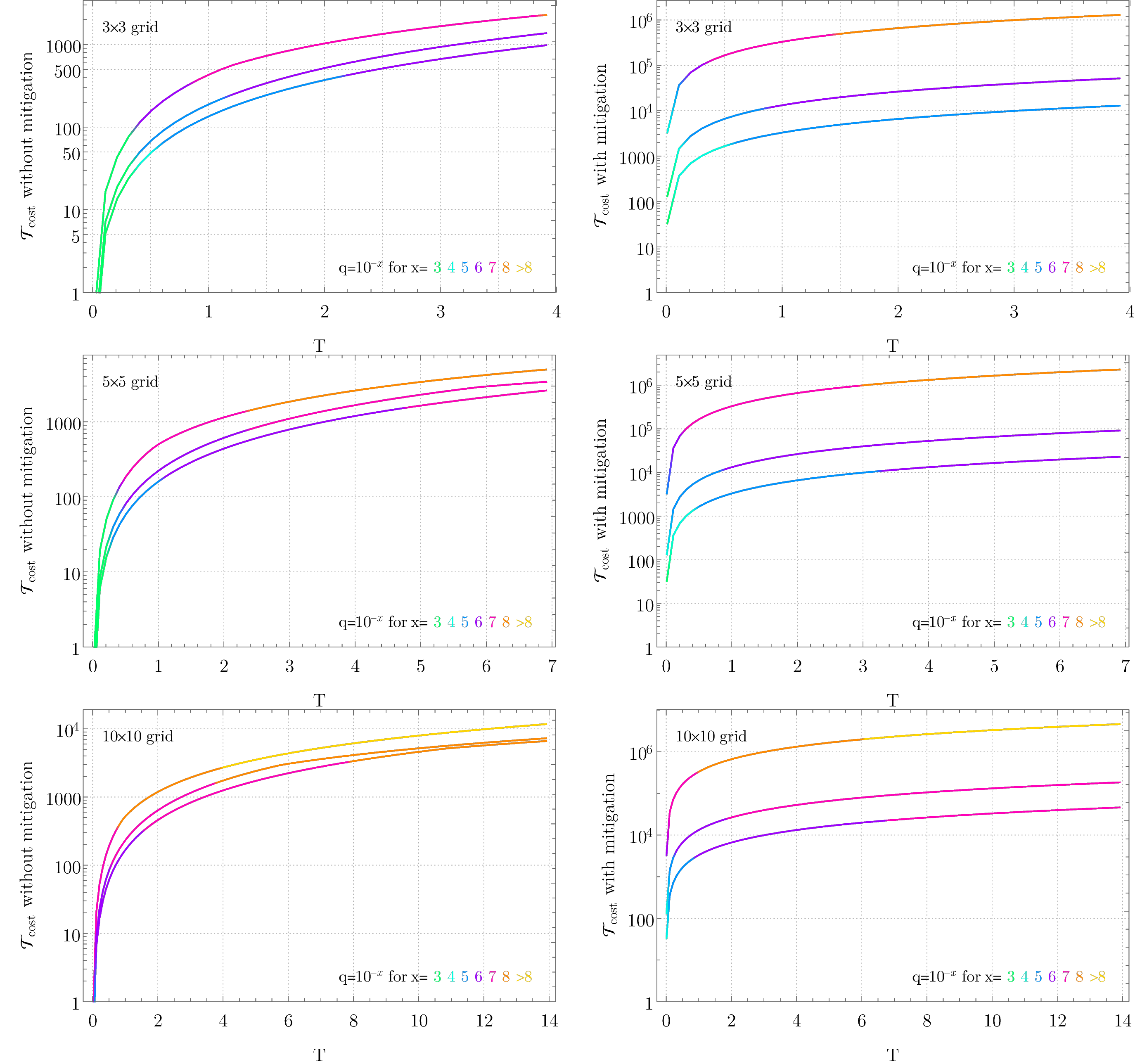}
    \end{minipage}}
    \caption{$\cost$ vs.~target time $T$ both without (left) and with (right) error mitigation, for lattice sizes $3\times3$, $5\times5$, and $10\times10$, for the Fermi-Hubbard Hamiltonian $\op H_{\text{FH}}$ from \cref{def:unencoded-H} in the compact encoding.
      These plots show a per-gate error model where $\cost$ equals circuit depth, and we synthesise local gates with a depth-3 gate decomposition by conjugation when error mitigation is not used (left) and our depth-4 decomposition (see 
      main text) if error mitigation is used (right).
    In each plot, three lines (top to bottom) represent $1\%$, $5\%$, and $10\%$ Trotter error $\epsilon$ given in \cref{eq:trotter-epsilon-ap}, where we minimize over the product formula order $p\in\{1,2,4,6\}$ and extrapolate a numerical simulation of FH lattices up to size $3\times3$ to obtain the optimal Trotter step size $\delta_0$.
    The line color corresponds to the intervals wherein the \emph{decoherence} error of the circuit is upper-bounded by the \emph{Trotter} error of the simulation, given a specific depolarizing noise parameter $q$. E.g.\ the blue sections of the top line within each row indicate that the decoherence error remains below $1\%$, for $q=10^{-5}$.
    }
    \label{fig:cost-numeric-delta-ignore-pulse-lengths}
\end{figure}

\begin{figure}[t]
    \centering
    \vspace{-1cm}
    \centerline{\begin{minipage}{18cm}
    \includegraphics[width=\textwidth]{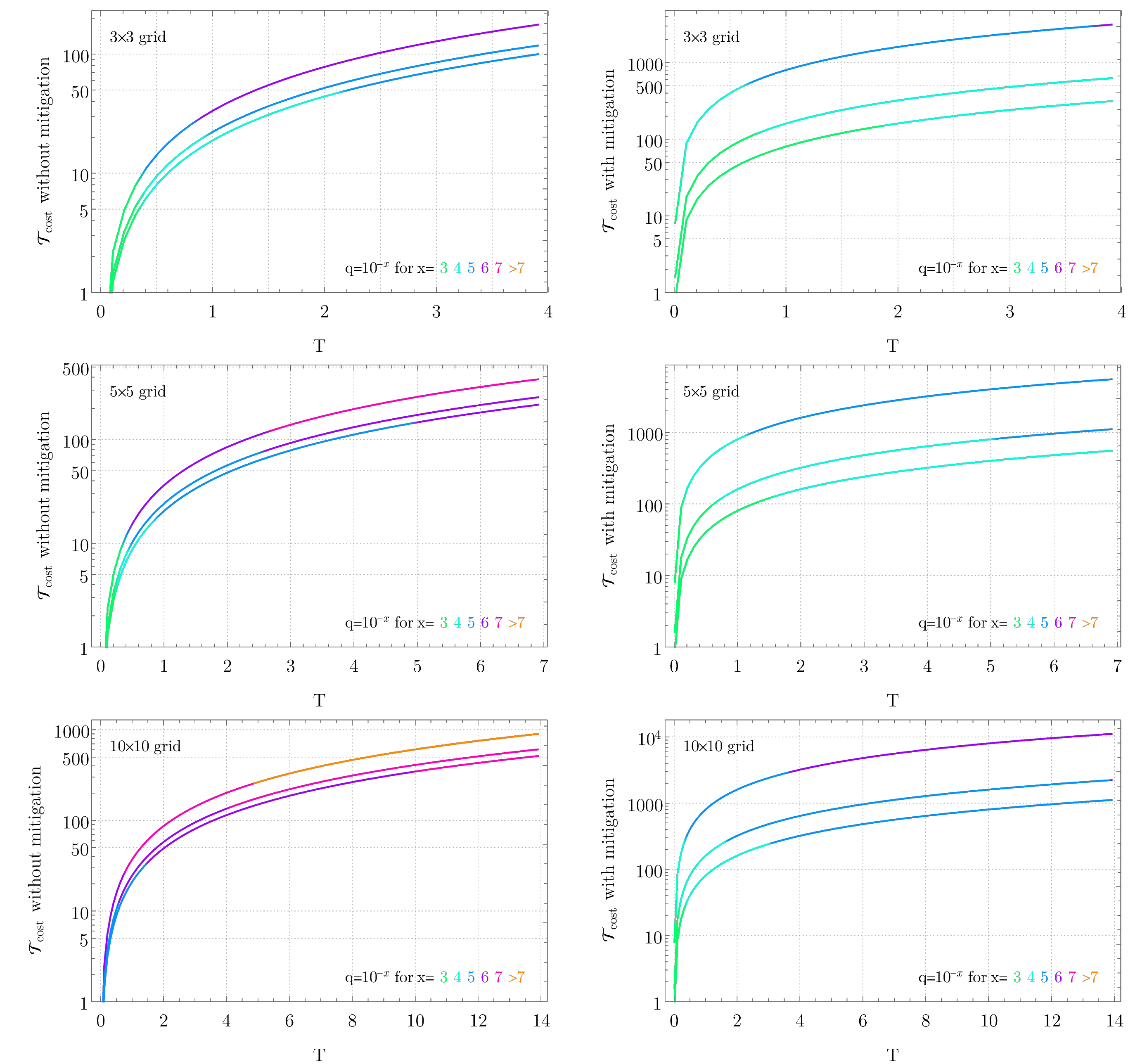}
    \end{minipage}}
    \caption{$\cost$ vs.~target time $T$ both without (left) and with (right) error mitigation , for lattice sizes $3\times3$, $5\times5$, and $10\times10$, for the Fermi-Hubbard Hamiltonian $\op H_{\text{FH}}$ from \cref{def:unencoded-H} in the compact encoding.
    Assuming the per-time error model where the noise probability is proportional to the gate lengths.
    Here $\cost$ denotes the circuit depth equivalent (i.e.\ the sum of all pulse times from \cref{def:cost-circuit depth}; all 3-local gates in the compact encoding are decomposed with the depth-4 method from \cref{lem:Weight-Increase-Depth4-ap}.
    In each plot, three lines represent $1\%$, $5\%$, and $10\%$ Trotter error $\epsilon$ given in \cref{eq:trotter-epsilon-ap}, where we minimize over the product formula order $p\in\{1,2,4,6\}$ and extrapolate a numerical simulation of FH lattices up to size $3\times3$ to obtain the optimal Trotter step size $\delta_0$.}
    \label{fig:cost-numeric-delta}
\end{figure}

\begin{figure}[t]
    \centering
    \vspace{-1cm}
    \centerline{\begin{minipage}{18cm}
    \includegraphics[width=\textwidth]{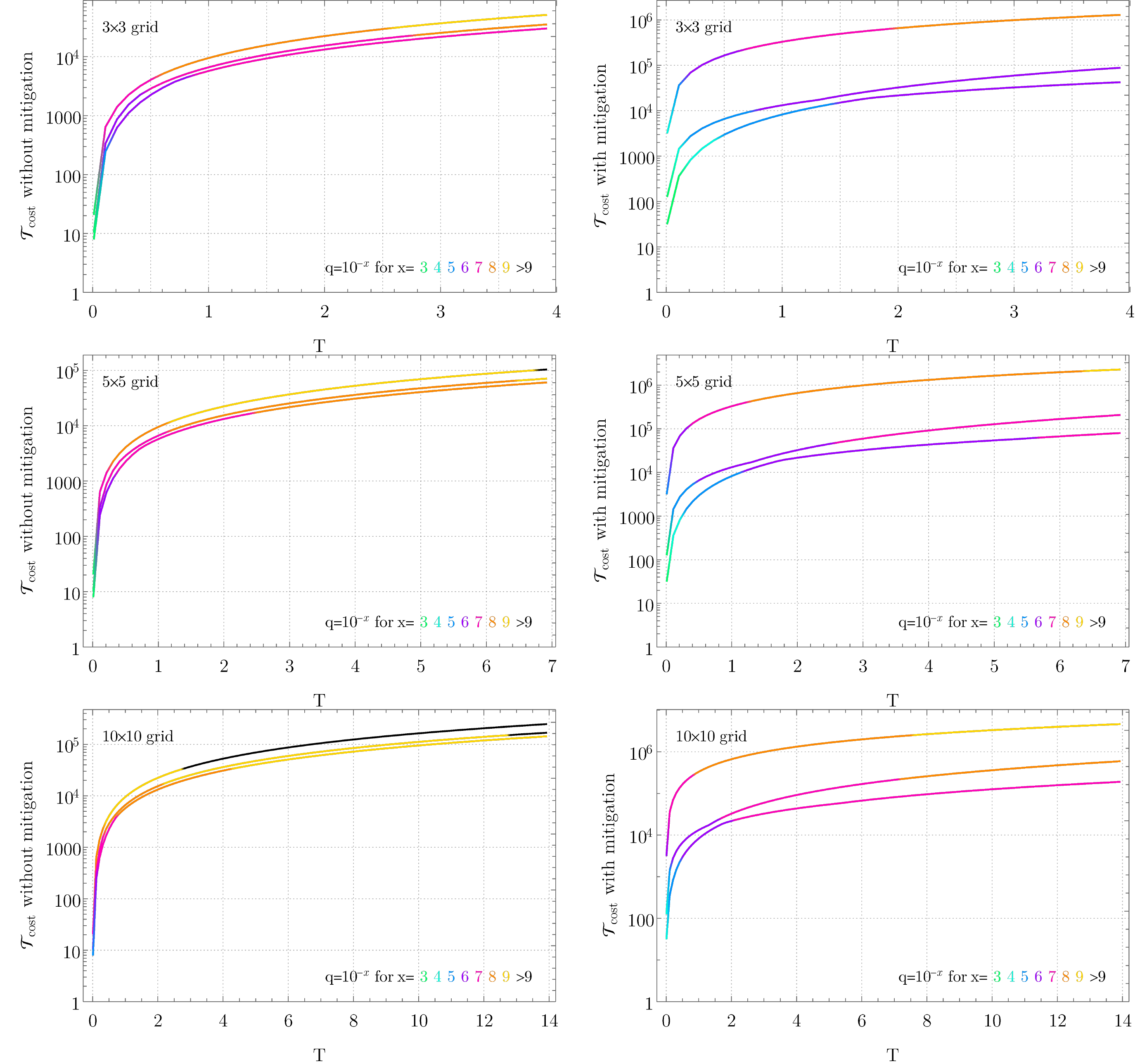}
    \end{minipage}}
    \caption{Setup as in \cref{fig:cost-numeric-delta-ignore-pulse-lengths} with the per-gate error model, but we use the tightest analytic (instead of numeric) error expression from \cref{cor:trotter-error-ap,th:Trotter-Er-Commutator-ap,cor:taylor-error-bound-ap}.
    We use a depth-3 gate decomposition when error mitigation is not used (left) and a depth-4 decomposition if error mitigation is used (right).}
    \label{fig:cost-analytic-delta-ignore-pulse-lengths}
\end{figure}

\begin{figure}[t]
    \centering
    \vspace{-1cm}
    \centerline{\begin{minipage}{18cm}
    \includegraphics[width=\textwidth]{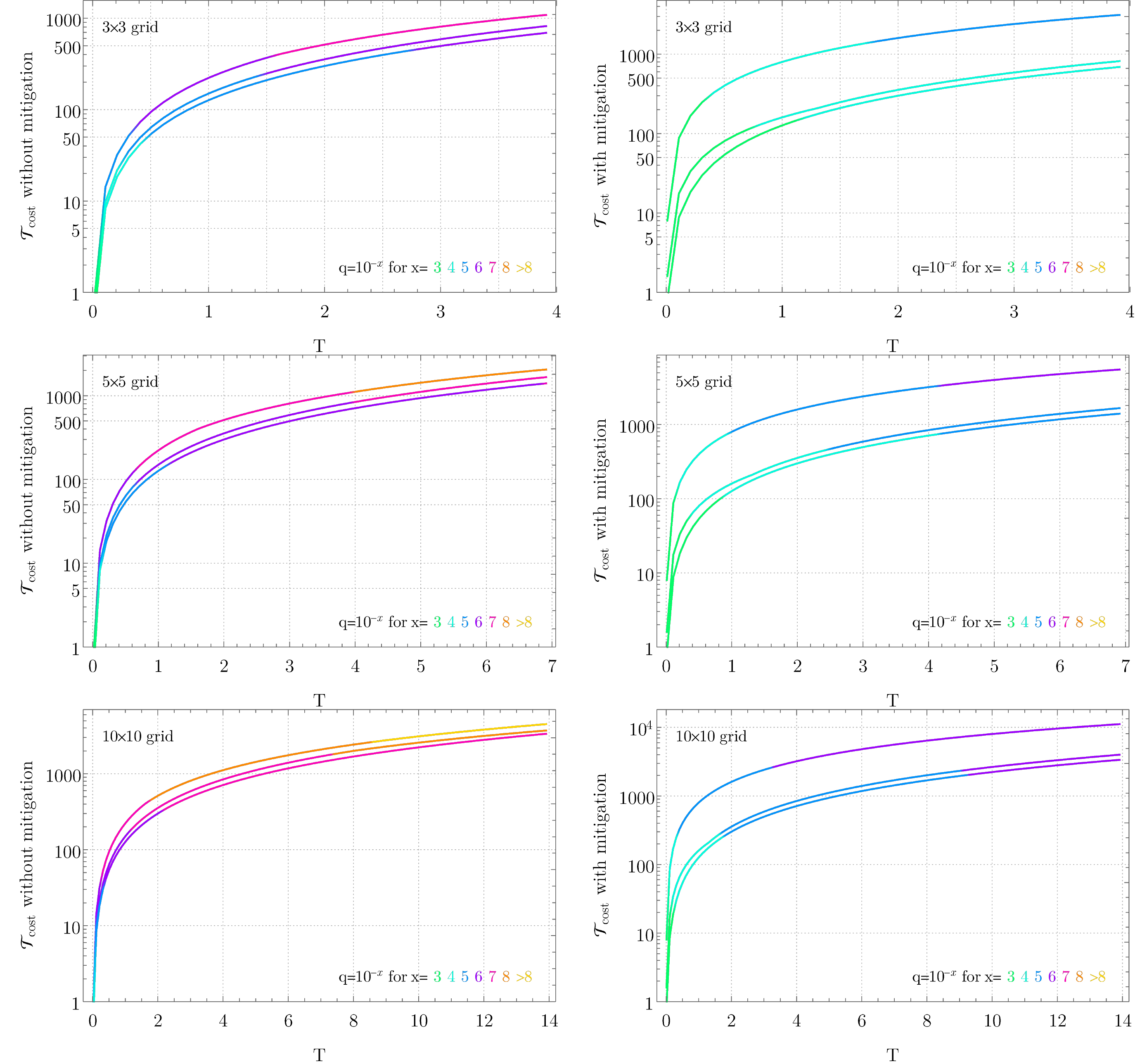}
    \end{minipage}}
    \caption{Setup as in \cref{fig:cost-numeric-delta} with the per-time error model, but we use the tightest error expression from \cref{cor:trotter-error-ap,th:Trotter-Er-Commutator-ap,cor:taylor-error-bound-ap}. All 3-local gates in the compact encoding are decomposed with the depth-4 method from \cref{lem:Weight-Increase-Depth4-ap}.}
    \label{fig:cost-analytic-delta}
\end{figure}

\begin{figure}[t]
    \centering
    \vspace{-1cm}
    \centerline{\begin{minipage}{18cm}
    \includegraphics[width=\textwidth]{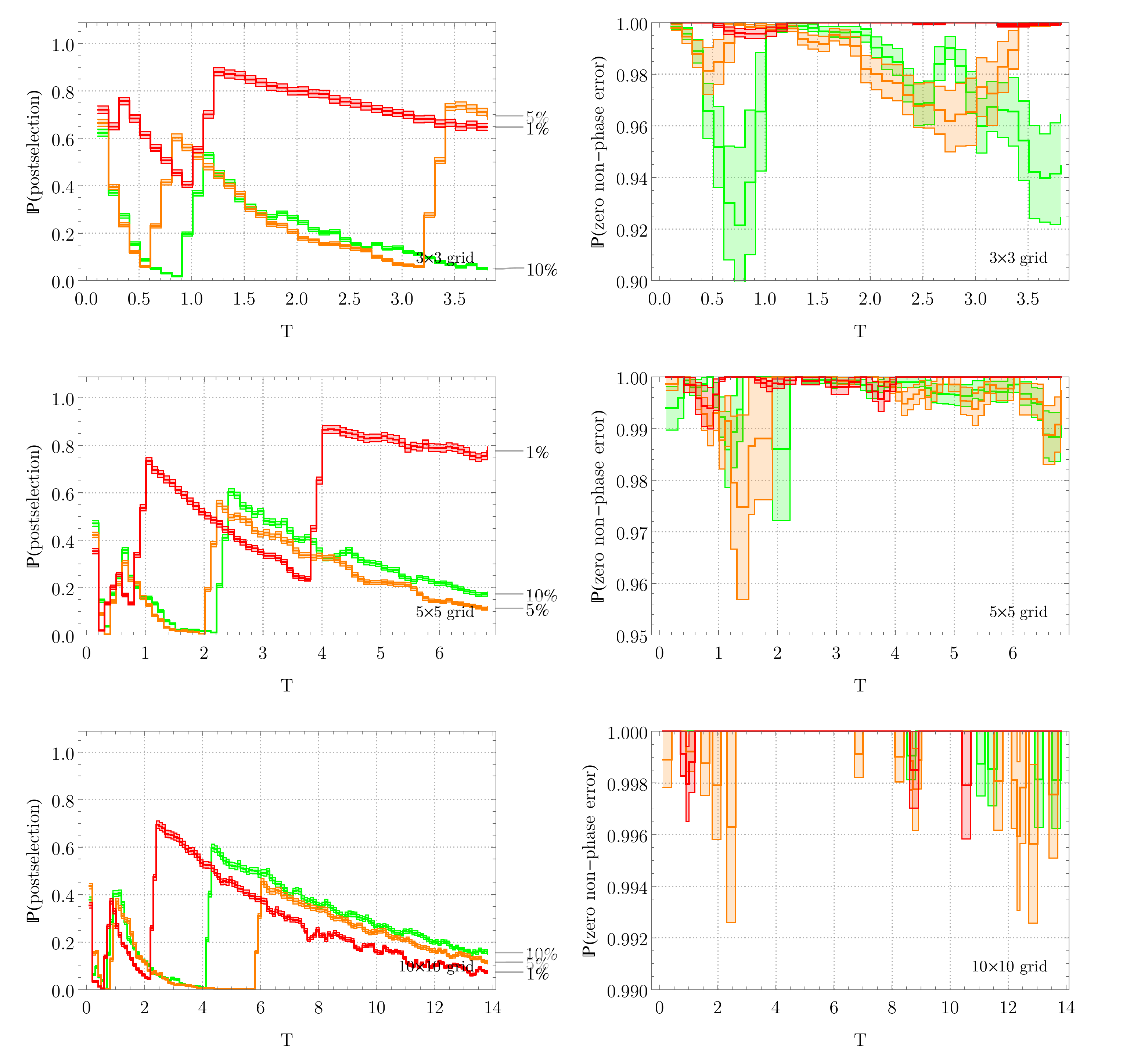}
    \end{minipage}}
    \caption{
    Post-selection probabilities (left column) and probability of zero undetectable non-phase error after post-selection (right column) for lattice sizes $3\times3$, $5\times5$, and $10\times10$, for the Fermi-Hubbard Hamiltonian $H_\mathrm{FH}$ from \cref{eq:FH-H-intro} (main text) in the compact encoding, to go alongside \cref{fig:cost-analytic-delta}.
    The choice between $1\%$, $5\%$ and $10\%$ Trotter error is made according to the colouring shown in \cref{fig:cost-analytic-delta}.
    }
    \label{fig:postsel-probs}
\end{figure}
\section*{Simulating Fermi-Hubbard with Three Trotter Layers}\label{ap:regroup}
\subsection*{Further Circuit Decompositions}
In this section we show that we can actually simulate a 2D spin Fermi-Hubbard model with $M=3$ Trotter layers as opposed to the previous $M=5$. First we need to introduce another circuit decomposition in the same spirit as before.
\begin{lemma}[Depth 3 Decomposition]\label{lem:Rank-Increase-Depth3-ap}
  Let $\op U(t)=\ee^{\ii t \left(\cos(\theta) \op h_1 + \sin(\theta) \op h_2\right)}$ be the time-evolution operator for time $t$ under a Hamiltonian $\op H_{\theta}=\cos(\theta) \op h_1 + \sin(\theta) \op h_2$.
  If $\op h_1$ and $\op h_2$ anti-commute and both square to identity, $\op U(t)$ can be decomposed as
  \begin{align}
    \op U(t) &=\ee^{\ii t_1 \op h_1} \ee^{\ii t_2 \op h_2} \ee^{\ii t_1 \op h_1}
  \end{align}
  where the pulse times $t_1,t_2$ as a function of the target time $t$ are given by
  \begin{align}
    t_1 &=\frac{1}{2} \tan^{-1}\left(\pm\frac{\cos (t)}{\sqrt{1-\sin^2(\theta ) \sin^2(t)}}, \; \pm\frac{\cos (\theta ) \sin (t)}{\sqrt{1-\sin^2(\theta ) \sin^2(t)}}\right)+\pi  c\\
    t_2 &=\tan^{-1}\left(\pm\sqrt{1-\sin^2(\theta ) \sin^2(t)}, \;\ \pm\sin (\theta ) \sin (t)\right)+2 \pi  c
  \end{align}
  where $c \in \mathbb{Z}$ and signs are taken consistently throughout.
\end{lemma}
\begin{proof}
  Since $h_1,h_2$ square to identity by assumption, we have
  \begin{align}
    \ee^{\ii t_1  h_1} \ee^{\ii t_2  h_2} \ee^{\ii t_1  h_1}  &= I \cos \left(2 t_1\right) \cos \left(t_2\right) + \ii \op h_1 \sin \left(2 t_1\right) \cos \left(t_2\right)+ \ii \op h_2 \sin \left(t_2\right),
  \end{align}
  and
  \begin{align}
    \ee^{\ii t \left(\cos(\theta) \op h_1 + \sin(\theta) \op h_2\right)} = I \cos \left(t \right) + \ii \sin \left(t \right) \left( \cos(\theta) \op h_1 + \sin(\theta) \op h_2 \right).
  \end{align}
  Equating these and solving for $t_1$ and $t_2$ gives the expressions in the \namecref{lem:Rank-Increase-Depth3-ap}.
\end{proof}
We then need to establish the overhead associated with implementing this decomposition. We will see that for a target time $t$, the pulse times in \cref{lem:Rank-Increase-Depth3-ap} are as $t_i(t) \propto t$.
\begin{lemma}\label{lem:Rank-Increase-Depth3-param-bounds-ap}
  Let $\op H = \cos(\theta) \op h_1 + \sin(\theta) \op h_2$ be as in \cref{lem:Rank-Increase-Depth3-ap}.
  For $0 \leq t \leq \pi/2$ and $0<\theta<\pi/2$, the pulse times $t_1,t_2$ in \cref{lem:Rank-Increase-Depth3-ap} can be bounded by
  \begin{align}
    |t_1|  &\leq \frac{t}{2}, \\
    |t_2|  &\leq t \theta.
  \end{align}
\end{lemma}
\begin{proof}
  \Cref{lem:Rank-Increase-Depth3-ap} gives two valid sets of solutions for $t_1,t_2$.
  Choose the following solution:
  \begin{align}
    t_1 &=\frac{1}{2} \tan^{-1}\left(\frac{\cos (t)}{\sqrt{1-\sin^2(\theta ) \sin^2(t)}},\frac{\cos (\theta ) \sin (t)}{\sqrt{1-\sin^2(\theta ) \sin^2(t)}}\right)\\
    t_2 &=\tan^{-1}\left(\sqrt{1-\sin^2(\theta ) \sin^2(t)},\sin (\theta ) \sin (t)\right).
  \end{align}
  Taylor expanding these functions about $t=0$ and $\theta=0$, we have
  \begin{align}
    t_1 &= 
          \frac{t}{2} + R_1\left(t, \theta\right), \\
    t_2 &= 
           t \theta  + R_2\left(t, \theta\right),
  \end{align}
  Basic calculus shows that the Taylor remainders $R_1,R_2$ are always negative for the stated range of $t$, giving the stated bounds.
\end{proof}

\subsection*{Regrouping Interaction Terms}
Now we apply this and the previous decompositions to simulate $H_{\text{FH}}$ as encoded using the compact encoding, using only three Trotter layers: $\{H_0, H_1, H_2\}$. The first of these layers consists of all the on-site interactions:
\begin{align}
    \op H_0 := \frac{\os}{4} \sum_{i=1}^{N}\left(I - Z_{i \uparrow}\right)\left(I - Z_{i \downarrow}\right),
\end{align}
and the other two layers are a \emph{mix} of horizontal and vertical hopping terms. Each has the same form, but consists of different sets of interactions as shown in \cref{fig:M=3-Split-ap}.
\begin{align}
    \op H_{1/2} &:=+ \frac{\hop}{2}\sum_{i<j}\sum_{\sigma \in \{\uparrow,\downarrow\}} \left( X_{i,\sigma}X_{j,\sigma}Y_{f(i,j),\sigma}+Y_{i,\sigma}Y_{j,\sigma}Y_{f(i,j),\sigma}\right) \\
    &+\frac{\hop}{2}\sum_{i<j}\sum_{\sigma \in \{\uparrow,\downarrow\}} g(i,j)\left( X_{i,\sigma}X_{j,\sigma}X_{f(i,j),\sigma}+Y_{i,\sigma}Y_{j,\sigma}X_{f(i,j),\sigma}\right).
\end{align}
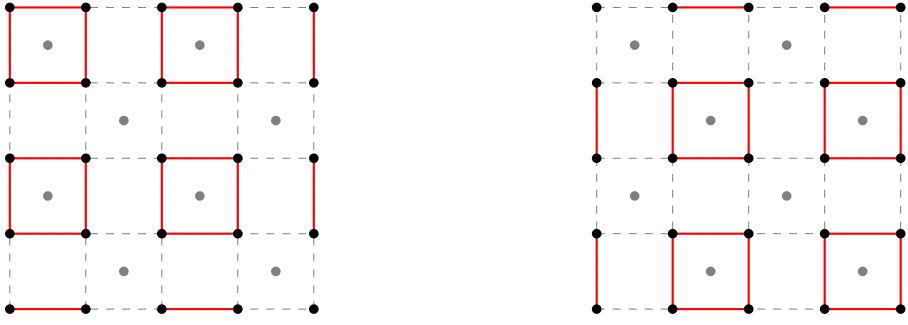
\begin{figure}
\begin{subfigure}[b]{.5\linewidth}
\centering
\begin{tikzpicture}

\draw[dashed][step=1cm,gray,very thin] (0,0) grid (4,4);

\draw[solid][step=1cm,red,thick] (0,0) -- (1,0);
\draw[solid][step=1cm,red,thick] (0,1) grid (1,2);
\draw[solid][step=1cm,red,thick] (0,3) grid (1,4);
\draw[solid][step=1cm,red,thick] (2,0) -- (3,0);
\draw[solid][step=1cm,red,thick] (2,1) grid (3,2);
\draw[solid][step=1cm,red,thick] (2,3) grid (3,4);
\draw[solid][step=1cm,red,thick] (4,1) -- (4,2);
\draw[solid][step=1cm,red,thick] (4,3) -- (4,4);

\foreach \x in {1.5,3.5}{
    \foreach \y in {0.5,2.5}{
        \node at (\x,\y)[fill=gray,circle,scale=0.35]{};
        }
    }

\foreach \x in {0.5,2.5}{
    \foreach \y in {1.5,3.5}{
        \node at (\x,\y)[fill=gray,circle,scale=0.35]{};
        }
    }

\foreach \x in {0,1,2,3,4}{
    \foreach \y in {0,1,2,3,4}{
        \node at (\x,\y)[fill=black,circle,scale=0.35]{};
        }
    }
\end{tikzpicture}
\end{subfigure}
\begin{subfigure}[b]{.5\linewidth}
\centering
\begin{tikzpicture}

\draw[dashed][step=1cm,gray,very thin] (0,0) grid (4,4);

\draw[solid][step=1cm,red,thick] (0,0) -- (0,1);
\draw[solid][step=1cm,red,thick] (1,0) grid (2,1);
\draw[solid][step=1cm,red,thick] (3,0) grid (4,1);
\draw[solid][step=1cm,red,thick] (0,2) -- (0,3);
\draw[solid][step=1cm,red,thick] (1,2) grid (2,3);
\draw[solid][step=1cm,red,thick] (3,2) grid (4,3);
\draw[solid][step=1cm,red,thick] (1,4) -- (2,4);
\draw[solid][step=1cm,red,thick] (3,4) -- (4,4);

\foreach \x in {1.5,3.5}{
    \foreach \y in {0.5,2.5}{
        \node at (\x,\y)[fill=gray,circle,scale=0.35]{};
        }
    }

\foreach \x in {0.5,2.5}{
    \foreach \y in {1.5,3.5}{
        \node at (\x,\y)[fill=gray,circle,scale=0.35]{};
        }
    }

\foreach \x in {0,1,2,3,4}{
    \foreach \y in {0,1,2,3,4}{
        \node at (\x,\y)[fill=black,circle,scale=0.35]{};
        }
    }
\end{tikzpicture}
\end{subfigure}
\caption{The interactions in $\op H_1$ (left) and those in $\op H_2$ (right). Gray nodes represent ancillary qubits, all non-gray qubits encode fermionic sites of a particular spin.}
\label{fig:M=3-Split-ap}
\end{figure}

Now we show that $\ee^{\ii \delta \op H_{1}}$ can be implemented directly; the case of $\ee^{\ii \delta \op H_{2}}$ follows similarly. As $\op H_{1}$ consists of interactions on disjoint sets of qubits, each forming a square on four qubits in \cref{fig:M=3-Split-ap}, we only need to show that we can implement evolution under the interactions making up one of the squares. We will denote these by $\op h_i$. We will label an example $\op h_1$ as shown in \cref{fig:interaction-square-ap} and demonstrate that this can be done. Using this labelling, the square of interactions is given by
\begin{align}
    \op h_1 &= \frac{1}{2} \left( X_1 X_2 + Y_1 Y_2 \right) Y_a +  \frac{1}{2} \left( X_3 X_4 + Y_3 Y_4 \right) Y_a \\
    &+\frac{1}{2} \left( X_1 X_4 + Y_1 Y_4 \right) X_a -  \frac{1}{2} \left( X_2 X_3 + Y_2 Y_3 \right) X_a.
\end{align}
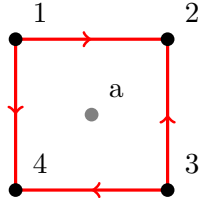
\begin{figure}
    \centering
\begin{tikzpicture}
\begin{scope}[red,very thick,decoration={
    markings,
    mark=at position 0.5 with {\arrow{>}}}
    ]

\draw[dashed][step=2cm,black,very thin] (0,0) grid (2,2);

\draw[postaction={decorate}] (2,0)--(0,0);
\draw[postaction={decorate}] (0,2)--(2,2);
\draw[postaction={decorate}] (2,0)--(2,2);
\draw[postaction={decorate}] (0,2)--(0,0);

\node[label={[black,label distance=0.1mm]45:a}] at (1,1)[fill=gray,circle,scale=0.5]{};

\node[label={[black,label distance=0.1mm]45:4}] at (0,0)[fill=black,circle,scale=0.5]{};
\node[label={[black,label distance=0.1mm]45:1}] at (0,2)[fill=black,circle,scale=0.5]{};
\node[label={[black,label distance=0.1mm]45:3}] at (2,0)[fill=black,circle,scale=0.5]{};
\node[label={[black,label distance=0.1mm]45:2}] at (2,2)[fill=black,circle,scale=0.5]{};
\end{scope}
\end{tikzpicture}
    \caption{The Hamiltonian $\op h_1$: a \emph{sum} of each Pauli interaction represented by a red line connecting a pair of qubits. Upward pointing arrows indicate $g(i,j)=-1$ and downward, left and right pointing arrows indicate $g(i,j)=1$ (See \cite{DK}). $\op H_1$ is a sum of disjoint Hamiltonians of this form, shown in \cref{fig:M=3-Split-ap}}
    \label{fig:interaction-square-ap}
\end{figure}
Now we will regroup these interactions in such away that we can use the methods of \cref{lem:Rank-Increase-Depth3-ap} to decompose $\ee^{\ii \delta \op h_1}$.
To do this we group the terms as $\op h_1= \op a_1 + \op a_2 + \op b_1 +\op b_2$, where
\begin{align}
    \op a_1 &= \frac{1}{\sqrt{2}}\left(\frac{X_1 X_2 Y_a - X_a X_2 X_3}{\sqrt{2}}\right) \\
    \op a_2 &= \frac{1}{\sqrt{2}}\left(\frac{Y_1 Y_4 X_a + Y_a Y_4 Y_3}{\sqrt{2}}\right) \\
    \op b_1 &= \frac{1}{\sqrt{2}}\left(\frac{Y_1 Y_2 Y_a - X_a Y_2 Y_3}{\sqrt{2}}\right) \\
    \op b_2 &= \frac{1}{\sqrt{2}}\left(\frac{X_1 X_4 X_a + Y_a X_4 X_3}{\sqrt{2}}\right).
\end{align}

We have reordered terms in order to easily verify the following commutation and anti-commutation relations: (i).~Every $\op a_i$ and $\op b_i$ squares to something proportional to the identity and consists of two anti-commuting Pauli terms. (ii).~The only pairs of $\op a_i$ and $\op b_i$ which don't commute, instead anti-commute:
\begin{align}
    \{\op a_1, \op b_2 \} &=0 \\
    \{\op a_2, \op b_1 \} &=0.
\end{align}
It is easy to verify that all other pairs of $\op a_i$ and $\op b_i$ commute. This allows us to implement the target evolution as follows
\begin{align}
    \ee^{\ii \delta \op h_1} &= \ee^{\ii \delta \left(\op a_1 + \op b_2 \right)} \ee^{\ii \delta \left(\op a_2 + \op b_1 \right)}.
\end{align}

Now we only need to show that we can implement $\ee^{\ii \delta \left(\op a_1 + \op b_2 \right)}$; the implementation of  $\ee^{\ii \delta \left(\op a_2 + \op b_1 \right)}$ follows similarly.
Consider the Hamiltonian
\begin{align}
  \op a_1 + \op b_2&= \frac{1}{\sqrt{2}}\left(\frac{X_1 X_2 Y_a - X_a X_2 X_3}{\sqrt{2}}\right) + \frac{1}{\sqrt{2}}\left(\frac{X_1 X_4 X_a + Y_a X_4 X_3}{\sqrt{2}}\right).
\end{align}
This meets the criteria to decompose $\ee^{\ii \delta \left(\op a_1 + \op b_2 \right)}$ with two applications of \cref{lem:Rank-Increase-Depth3-ap}.
The first application gives the decomposition
\begin{align}\label{eq:three_decomp}
    \ee^{\ii \delta \left(\op a_1 + \op b_2 \right)} &= \ee^{\ii t_1 \frac{X_1 X_2 Y_a - X_a X_2 X_3}{\sqrt{2}}} \ee^{\ii t_2\frac{X_1 X_4 X_a + Y_a X_4 X_3}{\sqrt{2}} }\ee^{\ii t_1\frac{X_1 X_2 Y_a - X_a X_2 X_3}{\sqrt{2}}}
\end{align}
where the pulse times are a function of the target time $t_i=t_i(\delta)$ as defined in \cref{lem:Rank-Increase-Depth3-ap}.
Then we apply \cref{lem:Rank-Increase-Depth3-ap} again, to each of the individual gates in \cref{eq:three_decomp}. The first gate in \cref{eq:three_decomp} decomposes as
\begin{align}
    \ee^{\ii t_1 \frac{X_1 X_2 Y_a - X_a X_2 X_3}{\sqrt{2}}} &= \ee^{t_1 ( t_1) X_1 X_2 Y_a} \ee^{t_2 ( t_1) X_a X_2 X_3} \ee^{t_1 ( t_1) X_1 X_2 Y_a},
\end{align}
and the others decompose similarly.
Now we need only apply \cref{lem:Weight-Increase-Depth4-ap} to decompose these three qubit unitaries into evolution under two-local Pauli interactions. As \cref{lem:Rank-Increase-Depth3-param-bounds-ap} shows that up until now the pulse times have remained $\propto \delta$, it is only this final step which introduces a root overhead. Hence the run-time remains asymptotically as \cref{th:FH-optimum-analogue}.

\section*{Feasible simulation time as a function of hardware error rate}
In this Appendix we analyse the simulation performance from the perspective of hardware error rate, rather than circuit depth. Specifically, we analyse the following question: given a maximum allowable simulation error $\epsilon_{tar}$ and a given hardware noise rate $q$, for how long can we simulate a given Fermi Hubbard Hamiltonian? That is, what is the maximum simulation time $T_{tar}$ such that the total simulation error remains below the target threshold, $\epsilon \leq \epsilon_{tar}$? In \cref{tab:extra} we consider this question for an $L=5$ Fermi Hubbard model with target error $\epsilon_{tar} = 10\%$.

When error mitigation is unused the only contribution to the total error $\epsilon$ is from Trotter error $\epsilon_t$ and error due to stochastic noise $\epsilon_s$. That is we have $\epsilon = \sqrt{\epsilon^2_{t} + \epsilon^2_{s}} \leq \epsilon_{tar}$. When mitigation is used there is an additional contribution $\epsilon_c$ from commuting errors past short-time gates. That is $\epsilon = \sqrt{\epsilon^2_{t} + \epsilon^2_{s} + \epsilon^2_{c}} \leq \epsilon_{tar}$. All Trotter bounds used to produce \cref{tab:extra} are numeric. We are considering a fermion occupation number of $\Lambda = 5$.

In \Cref{tab:extra} we consider both per-time and per-gate error models, as well as standard circuit decompositions and subcircuit decompositions. Recall that the former only allows a standard gate set, whereas the latter allows access to a continuous family of two-qubit gates.

In the per-gate error model, the error mitigation techniques described in the Supplementary Method entitled ``Trivial Stochastic Error Bound'' do not apply, and we choose the decomposition which yields the lowest depth circuit. In both the standard and subcircuit decompositions, this implies using the conjugation technique \cref{eq:conj-method} rather than the decomposition techniques of before.
For example, when decomposing 3-local terms, using a standard gate set this conjugation decomposition must be applied twice to leave us with only single qubit rotations and gates that are equivalent to CNOT gates (up to single-qubit rotations). Whereas if we are in a per-gate error model but allowing a subcircuit gate-set, then we only need to decompose via conjugation once.

Where error mitigation can be applied in the per-gate model, we decompose all gates via our subcircuit decompositions as otherwise error mitigation is not possible. Similarly, in the per-time error model, we always use these decompositions, as they are always more efficient than conjugation decompositions in that model.

We see that in the per-time error model using the compact encoding, using our subcircuit synthesis with the error mitigation techniques they enable yields the best performance across the hardware error rates considered. For $q=10^{-6}$ admitting subcircuit gates and allowing for error mitigation takes us from a maximum simulation time of $T_{tar}=3.17$ to $T_{tar}=896$.
For the VC encoding, error mitigation does not yield an improvement. However subcircuit gates and our subcircuit synthesis techniques do, taking us from $T_{tar}=2.25$ to $T_{tar}=3.81$ for $q=10^{-6}$.

For both encodings smaller improvements are seen for $q=10^{-5}$ and $q=10^{-4}$.
In the per-gate error model for the compact encoding, we find that our subcircuit decomposition techniques and error mitigation strategy yields an improvement, but over error rates of $q=10^{-6}$ and $q=10^{-5}$. For the VC encoding we see that error mitigation does not help. However for $q=10^{-6}$ subcircuit gates still yield an improvement, taking us from $T_{tar}=1.63$ to $T_{tar}=2.2$.

In cases where error mitigation is used we include the classical post-selection overhead. We have deliberately capped this at a maximum of $\approx 10^{4}$, to disallow excessive post-selection overhead. In some cases this cap is reached, indicating that our error mitigation techniques could yield further improvement still, but at a potentially unreasonably high post-selection overhead.
\newpage
\clearpage
\begin{table}[h]
\thisfloatpagestyle{empty}
\hspace*{-1.0cm}
\begin{tabular}{@{}lllllllll@{}}
\toprule
error  model               & encoding            & decomp & mitigation & \begin{tabular}[c]{@{}l@{}}post\\selection\\ overhead\end{tabular} & \begin{tabular}[c]{@{}l@{}} $T_{tar}$\end{tabular} & \begin{tabular}[c]{@{}l@{}} $\delta$\end{tabular}   & \begin{tabular}[c]{@{}l@{}}$\lceil T_{tar}/\delta\rceil$\end{tabular} & $\cost$                                                     \\
\midrule
\multirow[t]{6}{*}{per time} & \multirow[t]{3}{*}{compact} & standard      & False      & n/a                                                                          & \begin{tabular}[c]{@{}l@{}}3.17\\ 0.74\\ 0.23 \\[3mm] \end{tabular}                              & \begin{tabular}[c]{@{}l@{}}0.133\\ 0.258\\ 0.456\\[3mm]\end{tabular}       & \begin{tabular}[c]{@{}l@{}}24\\ 3\\ 1\\[3mm]\end{tabular}                           & \begin{tabular}[c]{@{}l@{}}1273\\ 155\\ 27\\[3mm]\end{tabular}     \\
                          &                     & subcircuit    & False      & n/a                                                                          & \begin{tabular}[c]{@{}l@{}}28.2\\ 3.77\\ 0.55\\[3mm]\end{tabular}                              & \begin{tabular}[c]{@{}l@{}}0.054\\ 0.123\\ 0.299\\[3mm]\end{tabular}       & \begin{tabular}[c]{@{}l@{}}520\\ 31\\ 2\\[3mm]\end{tabular}                         & \begin{tabular}[c]{@{}l@{}}1377\\ 123\\ 12\\[3mm]\end{tabular}     \\
                          &                     & subcircuit    & True       & \begin{tabular}[c]{@{}l@{}}9.1e3\\ 9.8e3\\ 9.1e3\end{tabular}                & \begin{tabular}[c]{@{}l@{}}896.\\ 89.8\\ 8.84\end{tabular}                              & \begin{tabular}[c]{@{}l@{}}0.005\\ 0.005\\ 0.005\end{tabular}       & \begin{tabular}[c]{@{}l@{}}179e3\\ 18.0e3\\ 1770\end{tabular}                & \begin{tabular}[c]{@{}l@{}}144e3\\ 14e3\\ 1416\end{tabular} \\
                          \cmidrule{2-9}
                          & \multirow[t]{3}{*}{VC} & standard      & False      & n/a                                                                          & \begin{tabular}[c]{@{}l@{}}2.25\\ 0.48\\ 0.23\\[3mm]\end{tabular}                              & \begin{tabular}[c]{@{}l@{}}0.155\\ 0.318\\ 0.456\\[3mm]\end{tabular}       & \begin{tabular}[c]{@{}l@{}}15\\ 2\\ 1\\[3mm]\end{tabular}                           & \begin{tabular}[c]{@{}l@{}}956\\ 100\\ 33\\[3mm]\end{tabular}      \\
                          &                     & subcircuit    & False      & n/a                                                                          & \begin{tabular}[c]{@{}l@{}}3.81\\ 1.00\\ 0.23\\[3mm]\end{tabular}                              & \begin{tabular}[c]{@{}l@{}}0.123\\ 0.226\\ 0.456\\[3mm]\end{tabular}       & \begin{tabular}[c]{@{}l@{}}31\\ 5\\ 1\\[3mm]\end{tabular}                           & \begin{tabular}[c]{@{}l@{}}803\\ 141\\ 21\\[3mm]\end{tabular}      \\
                          &                     & subcircuit    & True       & \begin{tabular}[c]{@{}l@{}}1.8e4\\ 8.5e3\\ 9.8e3\end{tabular}                & \begin{tabular}[c]{@{}l@{}}0.27\\ 0.04\\ 0.00416\end{tabular}                           & \begin{tabular}[c]{@{}l@{}}1.65e-6\\ 1.65e-6\\ 1.65e-6\end{tabular} & \begin{tabular}[c]{@{}l@{}}165e3\\ 25.0e3\\ 2530\end{tabular}                & \begin{tabular}[c]{@{}l@{}}95e3\\ 14e3\\ 1451\end{tabular}  \\
\midrule
\multirow[t]{6}{*}{per gate} & \multirow[t]{3}{*}{compact} & standard      & False      & n/a                                                                          & \begin{tabular}[c]{@{}l@{}}2.87\\ 0.51\\ 0.23\\[3mm]\end{tabular}                              & \begin{tabular}[c]{@{}l@{}}0.139\\ 0.308\\ 0.456\\[3mm]\end{tabular}       & \begin{tabular}[c]{@{}l@{}}21\\ 2\\ 1\\[3mm]\end{tabular}                           & \begin{tabular}[c]{@{}l@{}}1407\\ 113\\ 34\\[3mm]\end{tabular}     \\
                          &                     & subcircuit    & False      & n/a                                                                          & \begin{tabular}[c]{@{}l@{}}3.53\\ 0.83\\ 0.23\\[3mm]\end{tabular}                              & \begin{tabular}[c]{@{}l@{}}0.127\\ 0.245\\ 0.456\\[3mm]\end{tabular}       & \begin{tabular}[c]{@{}l@{}}28\\ 4\\ 1\\[3mm]\end{tabular}                           & \begin{tabular}[c]{@{}l@{}}1393\\ 170\\ 25\\[3mm]\end{tabular}     \\
                          &                     & subcircuit    & True       & \begin{tabular}[c]{@{}l@{}}1.2e4\\ 2.1e4\\ 1.6e4\end{tabular}                & \begin{tabular}[c]{@{}l@{}}11.1\\ 1.08\\ 0.11\end{tabular}                              & \begin{tabular}[c]{@{}l@{}}0.005\\ 0.005\\ 0.005\end{tabular}       & \begin{tabular}[c]{@{}l@{}}2211\\ 216\\ 22\end{tabular}                      & \begin{tabular}[c]{@{}l@{}}146e3\\ 14e3\\ 1465\end{tabular} \\
                          \cmidrule{2-9}
                          & \multirow[t]{3}{*}{VC} & standard      & False      & n/a                                                                          & \begin{tabular}[c]{@{}l@{}}1.63\\ 0.48\\ 0.23\\[3mm]\end{tabular}                              & \begin{tabular}[c]{@{}l@{}}0.180\\ 0.318\\ 0.456\\[3mm]\end{tabular}       & \begin{tabular}[c]{@{}l@{}}9\\ 2\\ 1\\[3mm]\end{tabular}                            & \begin{tabular}[c]{@{}l@{}}761\\ 126\\ 42\\[3mm]\end{tabular}      \\
                          &                     & subcircuit    & False      & n/a                                                                          & \begin{tabular}[c]{@{}l@{}}2.20\\ 0.50\\ 0.23\\[3mm]\end{tabular}                              & \begin{tabular}[c]{@{}l@{}}0.157\\ 0.312\\ 0.456\\[3mm]\end{tabular}       & \begin{tabular}[c]{@{}l@{}}14\\ 2\\ 1\\[3mm]\end{tabular}                           & \begin{tabular}[c]{@{}l@{}}929\\ 106\\ 33\\[3mm]\end{tabular}      \\
                          &                     & subcircuit    & True       & \begin{tabular}[c]{@{}l@{}}2.6e4\\ 5.1e3\\ 8.0e3\end{tabular}                & \begin{tabular}[c]{@{}l@{}}0.000932\\ 0.000135\\ 0.0000140\end{tabular}                 & \begin{tabular}[c]{@{}l@{}}1.65e-6\\ 1.65e-6\\ 1.65e-6\end{tabular} & \begin{tabular}[c]{@{}l@{}}567\\ 82\\ 9\end{tabular}                         & \begin{tabular}[c]{@{}l@{}}96e3\\ 14e3\\ 1440
                         \end{tabular} \\
                         \bottomrule
\end{tabular}
\caption{For $L=5$ with $\epsilon_{tar} = 10\%$ and $q = 10^{-6},\ 10^{-5}$ and $10^{-4}$ from top to bottom.}\label{tab:extra}
\end{table}

\thispagestyle{empty}
\FloatBarrier
\bibliographystyle{unsrt}
\bibliography{bibliography}
\end{document}